\documentclass[english,12pt]{article}

\usepackage{authblk}
\usepackage[T1]{fontenc}
\usepackage{inputenc}
\usepackage{geometry}
\usepackage{babel}
\usepackage{amsthm}
\usepackage{amsmath}
\usepackage{amssymb}
\usepackage{mathrsfs}
\usepackage{mathtools}
\usepackage{algorithm}
\usepackage[noEnd=false]{algpseudocodex}
\usepackage{bbm}
\usepackage{babel}
\usepackage{natbib}
\usepackage{fnpct}
\usepackage{IEEEtrantools}
\usepackage[title]{appendix}  
\usepackage{pdfpages}
\usepackage{xcolor}
\usepackage[open,openlevel=1]{bookmark}
\usepackage{IEEEtrantools}
\hypersetup{
colorlinks,
linkcolor={red!50!black},
citecolor={blue!50!black},
urlcolor={blue!80!black}
}
\usepackage{multirow}
\usepackage{tikz}
\usetikzlibrary{snakes,arrows,shapes}
\usepackage{comment}
\usepackage{booktabs}

\usepackage{subcaption} 
\usepackage{graphicx}
\usepackage{thmtools, thm-restate}  

\declaretheoremstyle[
spaceabove=6pt, spacebelow=6pt,
headfont=\normalfont\bfseries,
notefont=\mdseries, notebraces={(}{)},
bodyfont=\normalfont,
postheadspace=1em,
qed=$\blacksquare$
]{examplestyle}

\declaretheoremstyle[
spaceabove=6pt, spacebelow=6pt,
headfont=\normalfont\bfseries,
notefont=\mdseries, notebraces={(}{)},
bodyfont=\itshape,
postheadspace=1em
]{theorem}

\declaretheoremstyle[
spaceabove=4pt, spacebelow=4pt,
headfont=\itshape\bfseries,
notefont=\mdseries, notebraces={(}{)},
bodyfont=\itshape,
postheadspace=0.2em,
qed=\qedsymbol
]{remark}

\declaretheorem[style=theorem]{theorem}  
\declaretheorem[style=theorem, numbered=no, name=Theorem]{theorem*}  
\declaretheorem[style=remark,name=Remark]{remark}  
\declaretheorem[style=plain,name=Assumption]{assumption}  
\declaretheorem[style=definition,name=Definition]{definition}  
\declaretheorem[style=theorem, name=Lemma]{lemma}  
\declaretheorem[style=theorem, name=Proposition]{proposition}  
\declaretheorem[numbered=no,style=definition,name=Question]{question*}  
\declaretheorem[style=definition,name=Definition, numbered=no]{definition*}  

\declaretheorem[style=examplestyle]{example}
\renewcommand\thmcontinues[1]{continued}

\newcommand{\real}{\mathbb{R}}

\newcommand{\indicator}{\mathbbm{1}}

\newcommand{\optchoice}[2]{\max\left(#1; #2\right)}

\newcommand{\prob}{\mathbb{P}}

\newcounter{steps}

\newcommand\independent{\protect\mathpalette{\protect\independenT}{\perp}}
\def\independenT#1#2{\mathrel{\rlap{$#1#2$}\mkern2mu{#1#2}}}

\usepackage[nodisplayskipstretch]{setspace}
\usepackage{enumitem}
\setenumerate[1]{label=(\roman*)}
\setenumerate[2]{label=\normalfont{(\alph*)}}

\title{Robust Counterfactuals in Centralized Schools Choice Systems: Addressing Gender Inequality in STEM Education\thanks{
We are grateful to 
Marc Henry,
Marinho Bertanha,
Margaux Luflade,
and
Bumin Yemnez
for valuable comments and discussions.  
We thank Rana  Mohie El din for excellent research assistance.
The paper has also benefited from feedback received during presentations at
Chicago,
Duke, 
Johns Hopkins,
Wash U, 
Upenn,
and various conferences. 
}

}

\author[1]{Lixiong Li\footnote{Department of Economics, Johns Hopkins University. Email: lixiong.li@jhu.edu
}}
\author[2]{Isma\"el Mourifi\'e \footnote{Department of Economics,
  Washington University in St. Louis \& NBER. Address: One Brookings Drive
St. Louis, MO 63130-4899, USA. Email: ismaelm@wustl.edu.
}}
\affil[1]{Johns Hopkins University}
\affil[2]{Washington University in St. Louis \& NBER}
\date{\today }

\begin{document}
\maketitle
\vspace{-1em}

\begin{center}
\textbf{Abstract}
\end{center}

Counterfactual analysis is central to education market design and provides a foundation for credible policy recommendations. We develop a novel methodology for counterfactual analysis in Gale-Shapley deferred-acceptance (DA) assignment mechanisms under a weaker set of assumptions than those typically imposed in existing empirical works. Instead of fully specifying utility functions or students' beliefs about admission probabilities, we rely on interpretable restrictions on behavior that yield an incomplete but flexible model of preferences. We address the core challenge that partial identification poses for counterfactual analysis by showing that sharp bounds on counterfactual stable matching outcomes can be computed efficiently through a combination of algorithmic techniques and integer programming.  We illustrate the methodology by evaluating policies aimed at increasing female enrollment in STEM fields in Chile.

\vspace{0cm}

\vspace{0.5cm}

\noindent \textbf{Keywords:} Matching, School choice, Strategic reporting, Consideration sets, Partial identification, Linear program, STEM Gap. 

\noindent \textbf{JEL Classification:}  C12, C21, C26.

\normalsize

\newpage

\section{Introduction}

Counterfactual analysis is central to education market design and informs policy recommendations.  An area that has received significant attention in recent years is the problem of school choice. In market design, school choice refers to the challenge of assigning students to schools while accounting for preferences, school capacities, and policy objectives. A standard school choice procedure has three components: (i) preference elicitation, typically through students' rank-ordered lists (ROLs); (ii) a priority structure that assigns each student a priority (e.g., via scores or tie-breaking rules) which schools use to rank applicants; and (iii) an assignment algorithm that allocates seats based on ROLs, priorities, and school capacities.

In most cases, priorities and capacities are determined by policy decisions. However, translating broad policy goals -- such as diversity objectives -- into specific priorities, capacity constraints, or assignment algorithms is nontrivial. This is where counterfactual analysis becomes an essential tool, enabling policymakers to evaluate the potential effects of different policy choices on student assignments.
While researchers observe school priorities, capacities, ROLs, and the assignment algorithm in use, a fundamental challenge in counterfactual analysis is that students' \emph{true} preferences are typically unobserved. Identifying these preferences is an empirical question that requires careful modeling.

As noted by \cite{agarwal_demand_2018} and \cite{fack_beyond_2019}, many real-world assignment mechanisms create incentives for students to misreport their preferences. Consequently, observed ROLs need not reveal students' true preferences. 
To address this issue, \cite{agarwal_demand_2018} proposed a fully structural model that rationalizes how students generate their ROLs. However, this approach assumes a high level of strategic sophistication, requiring students to best respond to others' behavior while resolving complex uncertainties. In large-scale settings such as college admissions, researchers must solve an extensive computational problem, which necessitates strong parametric assumptions about the utility functions of all participants.\footnote{In markets with a large number of schools, a nonparametric analysis based on cardinal preferences may be infeasible. See, for instance, \cite{agarwal_demand_2018}, Theorem A.1.}
Beyond the risk of misspecification, another concern is empirical evidence that students frequently make mistakes, even in strategically simple environments. As \cite{Artemov2023} documents, students sometimes play unambiguously dominated strategies in settings where optimal behavior is straightforward.\footnote{This can occur when weakly dominated strategies yield the same assignment outcome in equilibrium.} Assuming that observed ROLs result from fully optimal behavior may therefore introduce substantial bias in empirical analysis.

An alternative is to rely on the \textit{stability assumption}. Stability plays a key role in the empirical analysis of matching markets. Under DA type mechanisms, stability implies that each student is assigned to her most-preferred school within her feasible set. A school is \textit{feasible} for a student if the student's placement score meets or exceeds the school's admission cutoff; see \cite{azevedo_supply_2016}.

Stability delivers revealed preference implications that rely solely on assignment outcomes and are agnostic to the specific mechanism or behavioral assumptions generating them. However, assignment data alone do not reveal preferences over schools outside a student's feasible set, nor does stability pin down any preference ordering within the feasible set beyond the relation involving the assigned school. As a result, some preferences, often including those most relevant for counterfactuals, remain unidentified, limiting the empirical content of stability by itself. See \cite{kapor2024aftermarket}.

To infer full preference profiles under stability, researchers often extrapolate from students with larger feasible sets. Moreover, this approach, developed in \cite{fack_beyond_2019} and used in subsequent work, typically assumes that latent preferences are independent of scores conditional on observed covariates.\footnote{Similar extrapolation methods appear in \cite{Akyol2016}, \cite{Bucarey2018}, \cite{ngo_preferences_2024}, \cite{barahona2021}, among others.} This assumption is fragile, particularly in post-secondary settings. See Section \ref{subsec:para} for discussion. 



In this paper, we propose a new method for counterfactual analysis in Gale--Shapley DA assignment mechanisms. Unlike the fully structural approach, we do not model the entire ROL as a deterministic function of cardinal utilities and beliefs about admission probabilities. Conversely, unlike the stability-only approach, we do not advocate disregarding the information in ROLs, focusing only on the final assignment, and relying on strong extrapolation. Our aim is to deliver informative counterfactual results while avoiding the restrictive assumptions underlying these two approaches, which, as noted in \citet{Agarwal2020}, can lead to non-robust counterfactual predictions.

Our first contribution is a general framework for extracting partial information about student preferences from observed behavior. We model students' preferences as total orders. Adopting a revealed-preference perspective, we show that the empirical content of several intuitive behavioral and rationality assumptions can be characterized by \textit{partial orders} that represent preference relations revealed by ROLs and matching outcomes. These assumptions include stability and the ``dropping strategy'' studied in \cite{kojima_incentives_2009} and \cite{haeringer_constrained_2009}.

We also introduce and analyze a novel behavioral restriction, the \textit{Robust Undominated Strategy}, which posits that a student avoids dominated strategies relative to her \emph{consideration set}\footnote{Consideration sets have also been studied in diverse contexts, including \cite{ben-akiva1973}, \cite{barseghyan2021}, and \cite{cattaneo2020}.}-- schools she regards as potentially within her reach. We derive sharp empirical implications of the Robust Undominated Strategy when only a subset of a student's consideration set is observed (e.g., using historical cutoffs, personal scores, and peer experiences). The entire consideration set needs not be known. In addition, we examine the empirical implications of beliefs that some schools are more selective than others. Finally, we show how to aggregate partial orders revealed under different assumptions into a single, most informative partial order, and use this structure to characterize the identified set of preference profiles.

As our second contribution, we provide a tractable procedure to bound the set of counterfactual stable matching outcomes when students' preferences are only partially identified. 
The standard approach is to estimate preferences (typically under parametric assumptions) and then rerun the DA algorithm in the counterfactual environment. This is infeasible here because the admissible preference profile is set-valued rather than point-identified.
An alternative is to exploit integer/linear-programming characterizations of stable matchings (e.g., \cite{Rothblum1992,Roth1993,Teo1998,Baiou2000}), but these formulations also require preferences to be point-identified. To address this, we propose two complementary methods.

First, we derive a system of linear inequalities that any counterfactual matching must satisfy to be stable with respect to some admissible preferences within the derived bounds. When admissible preferences are the ones compatible with our inferred partial orders, these constraints are necessary and sufficient. This characterization thus delivers sharp bounds on counterfactual outcomes via an integer programming optimization whose decision variables scale with the number of student-school pairs.

Second, we extend the Gale--Shapley DA algorithm to settings where students' preferences are only partially known. The extended procedure yields upper and lower bounds on school-specific cutoff scores in the counterfactual. We use these bounds to screen out unstable matching allocations and to fix or eliminate variables in the integer programming problem, 
substantially reducing its dimensionality prior to optimization. In our application, this screening reduces the number of decision variables by about $98\%$, making computation feasible even in large-scale matching markets.

Our final contribution is empirical. The under-representation of women in STEM fields is a persistent, well-documented global challenge. Beyond equity concerns, this imbalance has substantial economic consequences and is frequently cited as a contributor to the gender wage gap  (\cite{daymont1984, zafar2013}). The problem is especially pronounced in Latin America and the Caribbean (LAC), where the STEM gender gap exceeds that of many other regions (\cite{bello2020stem, uribe2021gender}). Recognizing the long-term implications, international organizations such as UNESCO have repeatedly called for targeted interventions to identify and remove barriers that deter women from pursuing STEM careers in LAC countries.

In this paper, we focus on Chile, where university admissions rely on a composite score that weights standardized high-stakes exam scores and high-school GPA. The high-stakes exams include subject-specific assessments in Mathematics, Science, Language, and History, typically administered on a single day, whereas GPA aggregates performance across multiple years of secondary education. Each academic program sets its own weights.

We document two stylized facts in the Chilean data. First, programs place substantially more weight on standardized exams than on GPA. Typically, 60--80\% of the total weight falls on exams. Second, male students generally score higher on most standardized exams, whereas female students tend to have higher GPAs. Similar patterns appear in other competitive educational contexts, see \cite{JurajdaMunich2011,Saygin2019,MontolioTaberner2021, IriberriReyBiel2019, ArenasCalsamiglia2025} and references therein.
These weighting choices, combined with gender differences in high-stakes exams performance, can unintentionally amplify disparities. We therefore study two counterfactual policies using our framework: (i) reallocating weight from standardized exams toward GPA, and (ii) within-gender standardization of exam scores (evaluating performance relative to same-gender score distributions).

We find that both policies could reduce gender gaps in admissions to STEM programs, with larger effects among students near the top of the score distribution and smaller effects among lower-scoring students. These results indicate that adjustments to priority weights and within-gender standardization can mitigate the impact of gender differences in exam and GPA distributions and narrow gender gaps in STEM enrollment.

The rest of this paper proceeds as follows. Section \ref{sec:Fram} introduces our analytical framework and matching environment. 
Section~\ref{sec:revealed_preference} discusses how to infer students' preferences based on two assumptions commonly adopted in the literature. We then turn to counterfactual analysis in Section~\ref{sec:counterfactual}, where we show how inferred preferences can inform policy evaluation. 
In Section~\ref{sec:more_revealed_preference}, we delve into more advanced techniques for recovering preferences from observed data and the ROLs.
We illustrate our approach with the Chilean data in Section \ref{Sec:App}.
The last section concludes. The appendix presents all proofs for the paper.

\subsubsection*{Notation and preliminaries}
A binary relation $\succ$ defined on a discrete set $\mathcal{J}$ is a \emph{(strict) total order} if it is irreflexive, asymmetric, transitive and connected. A binary relation $\succ^P$ is a \emph{(strict) partial order} if it satisfies irreflexivity, asymmetry, and transitivity, but not necessarily connectedness. The \emph{domain} of a strict partial order $\succ^P$, denoted as $\text{domain}(\succ^P)$, is the subset of elements in $\mathcal{J}$ involved in at least one relation under $\succ^P$. Formally, $j \in \text{domain}(\succ^P)$ if there exists some $j'\in \mathcal{J}$ such that either $j \succ^P j'$ or $j' \succ^P j$. We use $|\succ^P|$ to denote the number of elements in the domain of $\succ^P$. We also use $|\mathcal{J}|$ to denote the number of elements in set $\mathcal{J}$.

Throughout the paper, all binary orders are strict unless explicitly stated otherwise. For brevity, we henceforth refer to strict total orders simply as total orders and strict partial orders as partial orders. 

Let $\succ^P$ be a partial order such that its restriction to its domain is a total order. Then, for any subset $F \subseteq \mathcal{J}$ with $F \cap \text{domain}(\succ^P) \neq \emptyset$, there exists a unique maximal element in the restriction of $\succ^P$ to $F \cap \text{domain}(\succ^P)$. Formally, there exists a unique $j \in F \cap \text{domain}(\succ^P)$ such that $j \succ^P j'$ for all $j' \in F \cap \text{domain}(\succ^P)$ with $j' \neq j$. We denote this maximal element by $\optchoice{\succ^P}{F}$.

The \emph{transitive closure} of a binary relation $\succ$ is the smallest transitive relation that contains $\succ$. Given a binary relation $\succ$, its transitive closure $\succ'$ can be constructed as follows: for any $j,j'\in \mathcal{J}$, $j \succ' j$ if and only if there exist elements $j_1, \dots, j_N \in \mathcal{J}$ such that: (\emph{a}) $j_1 = j$, (\emph{b}) $j_N = j'$, and (\emph{c}) for each $n=1,\dots,N-1$, $j_n \succ j_{n+1}$.

We say that a total order $\succ$ and a binary relation $\succ'$ are \emph{compatible} if, for any $j_1, j_2 \in \mathcal{J}$ with $j_1 \succ' j_2$ and no $j_2 \succ' j_1$, we have $j_1 \succ j_2$. Similarly, two binary relations $\succ_1$ and $\succ_2$ are \emph{compatible} if there exists at least one total order compatible simultaneously with both. Given two compatible partial orders $\succ^P_1$ and $\succ^P_2$, their \emph{join}, denoted by $\succ^P_1 \vee \succ^P_2$, is the partial order obtained by taking the transitive closure of the union of their binary relations. Formally, for $j, j' \in \mathcal{J}$, we have $j (\succ^P_1 \vee \succ^P_2) j'$ if and only if there exist elements $j_1, \dots, j_N \in \mathcal{J}$ such that: (\emph{a}) $j_1 = j$, (\emph{b}) $j_N = j'$, and (\emph{c}) for each $n=1,\dots,N-1$, either $j_n \succ^P_1 j_{n+1}$ or $j_n \succ^P_2 j_{n+1}$.


\section{Analytical framework and matching environment}\label{sec:Fram}

We consider a two-sided matching environment consisting of a set of students $\Omega$ and a set of school programs $\mathcal{J}$. Additionally, we include an outside option labeled as program $0$ and define $\mathcal{J}_0 = \mathcal{J}\cup\{0\}$. 

Each student is matched to one program in $\mathcal{J}_0$, whereas each program may accommodate multiple students. Following \cite{azevedo_supply_2016}, we normalize the measure of the student population to $1$, and interpret the capacity $q_j$ of each program $j$ as the share of students it can accommodate. The outside option has a capacity of $1$, reflecting that it never has a binding capacity constraint. 

Each student $\omega$ has a true preference ordering $\succ^Q_\omega$, which is a total order over the set $\mathcal{J}_0$. Programs ranked higher than the outside option by student $\omega$ are called \emph{acceptable} to her. Additionally, each student $\omega$ is associated with a vector of priority scores $S_\omega = (S_{\omega j}: j\in \mathcal{J}_0)$.\footnote{For simplicity, we assume students have priority scores for all programs in $\mathcal{J}_0$, including the outside option. The priority scores for the outside option can be set arbitrarily since the capacity constraint is never binding for the outside option.} A program $j$ prefers student $\omega$ over student $\omega'$ if and only if $S_{\omega j} > S_{\omega j'}$. Without loss of generality, we assume priority scores $S_{\omega j}$ are distributed on the interval $[0, 1]$. We also assume the priority scores are continuously distributed, ensuring ties occur with probability zero.

A \emph{matching} is a measurable function $\mu:\Omega \to \mathcal{J}_0$ assigning each student to a single program, while respecting programs' capacity constraints. Formally, a matching must satisfy $\mathbb P \{\omega: \mu(\omega) = j \} \leq q_j$ for every program $j \in \mathcal{J}_0$. When a student is assigned to the outside option ($\mu(\omega)=0$), we interpret this as the student being unmatched with any program in $\mathcal{J}$.

A centralized matching mechanism determines a matching using three types of information: (\emph{i}) capacities $(q_j:j\in \mathcal{J}_0)$; (\emph{ii}) priority scores $(S_{\omega j}: j\in \mathcal{J}_0, \omega\in \Omega)$; and (\emph{iii}) students' ranked order lists (ROLs), indicating their preferences over acceptable programs. We denote student $\omega$'s submitted ROL as $\succ^R_\omega$. In practice, students often cannot rank all acceptable programs due to institutional constraints imposed by the matching mechanism. Thus, we assume the length of each ROL, i.e., $|\succ^R_\omega|$, is bounded by a constant $K$. If $K < |\mathcal{J}_0|$, the ROL $\succ^R_\omega$ is a partial order. Nonetheless, the submitted ROL must be a total order on its domain, with the outside option always placed last. The following example illustrates these concepts.

\begin{example}\label{running_example}
Suppose $\mathcal{J}_0 = \{0, 1, 2, \dots, 6\}$. Consider a student $\omega$ whose true preference ordering is given by:
$$ 1 \succ^Q_\omega 2 \succ^Q_\omega 3 \succ^Q_\omega 4 \succ^Q_\omega 5 \succ^Q_\omega 0 \succ^Q_\omega 6.$$ 
Assume the maximum number of programs listed in the ROL is $K = 5$. Given this constraint, student $\omega$ could submit the ROL $2 \succ^R_\omega 3 \succ^R_\omega 4 \succ^R_\omega 5 \succ^R_\omega 0$. Shorter lists are also permissible. For instance, both $2 \succ^R_\omega 3 \succ^R_\omega 5 \succ^R_\omega 0$ and $2 \succ^R_\omega 0$ are valid submissions. She is also allowed to submit a ROL entirely different from her true preferences, such as $6 \succ^R_\omega 5 \succ^R_\omega 0$, although such a submission might not be strategically optimal to her. However, submitting an ROL containing more than $5$ programs is prohibited under this constraint.
\end{example}

Our analysis focuses on the Gale-Shapley DA mechanism and its variants. As discussed in \cite{fack_beyond_2019}, the DA mechanism has become the leading centralized matching approach for student placement in many countries. Following \cite{azevedo_supply_2016}, we represent the DA through its \emph{cutoff characterization}.

Specifically, the outcome of the DA mechanism can be summarized by a vector of \emph{cutoffs} $c = (c_j : j\in\mathcal{J}_0)$, where each cutoff $c_j$ corresponds to the threshold priority score required for admission to program $j$. By convention, the cutoff for the outside option is always $0$. Given these cutoffs, each student $\omega$ faces a feasible set defined as:
\[
F(S_\omega, c) \coloneqq \{j\in \mathcal{J}_0: S_{\omega j} \ge c_j \}.
\]
In the absence of ambiguity, we adopt the shorthand notation $F_\omega(c)$.
Note that the outside option is always included in $F_\omega$ by construction. 
Under the DA mechanism, student $\omega$ is matched to her most preferred feasible program according to her submitted ROL. Formally, student $\omega$ is matched to program $j \in \mathcal{J}_0$ if and only if
\begin{equation}\label{eq:DA_assign_rule}
j = \optchoice{\succ^R_\omega}{F_\omega(c)},
\end{equation}
where $\optchoice{\succ^R_\omega}{F_\omega(c)}$ denotes the maximal element of $\succ^R_\omega$ restricted to the intersection of its domain and $F_\omega(c)$, as defined in the notation section. Consequently, students can only be matched to feasible programs listed in their ROLs. In particular, if no listed program other than the outside option is feasible, the student must accept the outside option.

The equilibrium cutoffs $c$ are determined by requiring that, for each program $j$ with strictly positive cutoff ($c_j>0$), the measure of students matched to that program equals its capacity $q_j$. Programs not reaching their capacity have their cutoff set to zero. We refer to \cite{azevedo_supply_2016} for more details. Given a matching $\mu$ generated by the DA mechanism, its associated cutoff $c$ can be recovered as
\begin{equation}\label{eq:equilibrium_cutoff}
c_j = 
\begin{cases}
	\inf_{\mu(\omega) = j} S_{\omega j}, & \text{if } \prob(\mu(\omega) = j) = q_j, \\[6pt]
	0, & \text{if } \prob(\mu(\omega) = j) < q_j.
\end{cases}
\end{equation}

The following example illustrates how a student's matching outcome is jointly determined by her feasible set and submitted ROL.
\begin{example}[continues=running_example]
Suppose student $\omega$'s feasible set at the equilibrium cutoffs is $F_\omega(c) = \{0, 4, 5, 6\}$.
\begin{itemize}
    \item If her submitted ROL is $2 \succ^R_\omega 3 \succ^R_\omega 4 \succ^R_\omega 5 \succ^R_\omega 0$, then she will be matched to program $4$, since $4$ is the maximal element of $\succ^R_\omega$ restricted to $F_\omega(c)$.
    
    \item If her submitted ROL is $2 \succ^R_\omega 0$, then she will be matched to the outside option, because program $2$ is not within her feasible set.
    
    \item If her submitted ROL is $6 \succ^R_\omega 5 \succ^R_\omega 0$, then she will be matched to program $6$.
\end{itemize}
Note that in the last scenario, she is matched to program $6$ regardless of whether her true preferences rank program $5$ or even the outside option higher. Thus, matching outcomes depend solely on submitted ROLs and the feasible set determined by equilibrium cutoffs, not directly on students' true preferences.
\end{example}

In this paper, we address two central questions: Given access to both the information used by the centralized matching mechanism and the realized match outcomes, what can we learn about students' true preferences? How can these inferred preferences be used to evaluate counterfactual policy scenarios?

\section{Revealed preferences and compatible partial orders}\label{sec:revealed_preference}

In this section, we discuss how to infer students' preferences based on two assumptions commonly adopted in the literature. Using observed data, our approach constructs partial orders that are compatible with students' true preferences under these assumptions. As we demonstrate, this method provides a transparent and theoretically grounded framework for preference revealing. We also discuss the limitations of these assumptions and address them more comprehensively in Section \ref{sec:more_revealed_preference}.

\subsection{Stability}\label{sec:stable}
In the literature, it is common to assume that the observed matching is stable, particularly when analyzing the DA mechanism and its variants. This stability assumption is widely adopted for identifying students' preferences over colleges; see, for example, \cite{fack_beyond_2019}, \cite{Akyol2016}, \cite{Bucarey2018}, and \cite{ngo_preferences_2024}, among many others. Stability is formally defined as follows:

\begin{definition}[Stability]\label{ass:stable}
A matching $\mu: \Omega \to \mathcal{J}_0$ is said to be \emph{stable} if the following three conditions hold for every student $\omega \in \Omega$:
\begin{itemize}
    \item[(i)] \emph{Individual rationality}: Either $\mu(\omega) \succ^Q_\omega 0$ or $\mu(\omega)=0$.
    \item[(ii)] \emph{No waste}: For any program $j \in \mathcal{J}$, if a positive measure of students prefers $j$ over their matched program, then $j$ must be filled to capacity.
    \item[(iii)] \emph{No justified envy}: For any program $j \in \mathcal{J}$ that is filled to capacity, if student $\omega$ prefers $j$ over her matched program $\mu(\omega)$, and another student $\omega'$ is matched to $j$, then it must be the case that $S_{\omega'j} > S_{\omega j}$.
\end{itemize}
\end{definition}

To facilitate referencing the stability assumption explicitly, we write it formally below:

\begin{assumption}[Stability]\label{assu:stability}
The observed matching $\mu$ is stable.
\end{assumption}

We can equivalently characterize stability using the cutoff characterization of DA. Given a matching outcome $\mu$ and its equilibrium cutoffs $c$, the three conditions in Definition \ref{ass:stable} imply that each student is matched to her most preferred feasible program according to her true preference. Formally, for each student $\omega$, we have $\mu(\omega) = \optchoice{\succ^Q_\omega}{F_\omega(c)}$. Combining this result with equation \eqref{eq:DA_assign_rule}, stability ensures that
\begin{equation}\label{eq:stability_immediate}
\optchoice{\succ^Q_\omega}{F_\omega(c)} = \optchoice{\succ^R_\omega}{F_\omega(c)}.
\end{equation}
This equality provides a basis for inferring information about students' true preferences from the stability assumption. To unify the analysis of this section with the analyses of other assumptions discussed in subsequent sections, we reinterpret equation \eqref{eq:stability_immediate} as a compatibility condition involving a partial order constructed below.

\begin{proposition}[Partial order induced by stability] \label{Prop:PS}
	Let $\mu$ be the matching outcome and let $c$ be its corresponding equilibrium cutoffs. Define for each student $\omega$ a partial order $\succ^{P_s}_\omega$ as follows: for every program $j\in F_{\omega}(c) \setminus \{\mu(\omega)\}$, we have $\mu(\omega) \succ^{P_s}_\omega j$. Then, Assumption \ref{assu:stability} holds if and only if, for each student $\omega$, the partial order $\succ^{P_s}_\omega$ is compatible with the student's true preference $\succ^Q_\omega$.
\end{proposition}

The Proof of Proposition \ref{Prop:PS} is presented in Appendix \ref{proof:stability}. Proposition \ref{Prop:PS} provides a necessary and sufficient characterization, meaning there is no loss of information implied by the stability assumption when we represent it through the compatibility condition stated in the proposition. To illustrate the partial order constructed in Proposition \ref{Prop:PS}, we revisit the running example.

\begin{example}[continues=running_example]
As before, suppose $\mathcal{J}_0 = \{0, 1, 2, ..., 6\}$ and student $\omega$'s feasibility set $F_{\omega}(c) = \{4, 5, 6 , 0\}$. In addition, suppose she is assigned to $4$. Then, her $\succ^{P_s}_\omega$ implied by the stability assumption only consists of three relations:  $4 \succ^{P_s}_\omega 5$, $4 \succ^{P_s}_\omega 6$ and $4 \succ^{P_s}_\omega 0$, as illustrated in Figure \ref{fig:stability_illustration}.
\end{example}

\begin{figure}[h]
    \centering
    \begin{tikzpicture}[>=latex,line join=bevel,]
  \pgfsetlinewidth{1bp}
\begin{scope}
  \pgfsetstrokecolor{black}
  \definecolor{strokecol}{rgb}{1.0,1.0,1.0};
  \pgfsetstrokecolor{strokecol}
  \definecolor{fillcol}{rgb}{1.0,1.0,1.0};
  \pgfsetfillcolor{fillcol}
  \filldraw (0.0bp,0.0bp) -- (0.0bp,65.48bp) -- (360.0bp,65.48bp) -- (360.0bp,0.0bp) -- cycle;
\end{scope}
  \pgfsetcolor{black}
  \draw [->] (194.77bp,28.728bp) .. controls (199.24bp,30.754bp) and (203.71bp,31.669bp)  .. (219.29bp,28.704bp);
  \draw [->] (190.46bp,33.125bp) .. controls (196.83bp,40.951bp) and (205.71bp,49.738bp)  .. (216.0bp,54.0bp) .. controls (230.78bp,60.123bp) and (237.22bp,60.123bp)  .. (252.0bp,54.0bp) .. controls (258.75bp,51.203bp) and (264.9bp,46.457bp)  .. (277.54bp,33.125bp);
  \draw [->] (190.46bp,33.125bp) .. controls (196.83bp,40.951bp) and (205.71bp,49.738bp)  .. (216.0bp,54.0bp) .. controls (252.96bp,69.307bp) and (269.04bp,69.307bp)  .. (306.0bp,54.0bp) .. controls (312.75bp,51.203bp) and (318.9bp,46.457bp)  .. (331.54bp,33.125bp);
\begin{scope}
  \definecolor{strokecol}{rgb}{0.0,0.0,0.0};
  \pgfsetstrokecolor{strokecol}
  \draw (18.0bp,18.0bp) ellipse (18.0bp and 18.0bp);
  \draw (18.0bp,18.0bp) node {$1$};
\end{scope}
\begin{scope}
  \definecolor{strokecol}{rgb}{0.0,0.0,0.0};
  \pgfsetstrokecolor{strokecol}
  \draw (72.0bp,18.0bp) ellipse (18.0bp and 18.0bp);
  \draw (72.0bp,18.0bp) node {$2$};
\end{scope}
\begin{scope}
  \definecolor{strokecol}{rgb}{0.0,0.0,0.0};
  \pgfsetstrokecolor{strokecol}
  \draw (126.0bp,18.0bp) ellipse (18.0bp and 18.0bp);
  \draw (126.0bp,18.0bp) node {$3$};
\end{scope}
\begin{scope}
  \definecolor{strokecol}{rgb}{0.0,0.0,0.0};
  \pgfsetstrokecolor{strokecol}
  \draw (180.0bp,18.0bp) ellipse (18.0bp and 18.0bp);
  \draw (180.0bp,18.0bp) node {$4$};
\end{scope}
\begin{scope}
  \definecolor{strokecol}{rgb}{0.0,0.0,0.0};
  \pgfsetstrokecolor{strokecol}
  \draw (234.0bp,18.0bp) ellipse (18.0bp and 18.0bp);
  \draw (234.0bp,18.0bp) node {$5$};
\end{scope}
\begin{scope}
  \definecolor{strokecol}{rgb}{0.0,0.0,0.0};
  \pgfsetstrokecolor{strokecol}
  \draw (288.0bp,18.0bp) ellipse (18.0bp and 18.0bp);
  \draw (288.0bp,18.0bp) node {$6$};
\end{scope}
\begin{scope}
  \definecolor{strokecol}{rgb}{0.0,0.0,0.0};
  \pgfsetstrokecolor{strokecol}
  \draw (342.0bp,18.0bp) ellipse (18.0bp and 18.0bp);
  \draw (342.0bp,18.0bp) node {$0$};
\end{scope}
\end{tikzpicture}
     \vspace{-2em}
    \caption{Illustration of $\succ^{P_s}_\omega$}
    \label{fig:stability_illustration}
\end{figure}
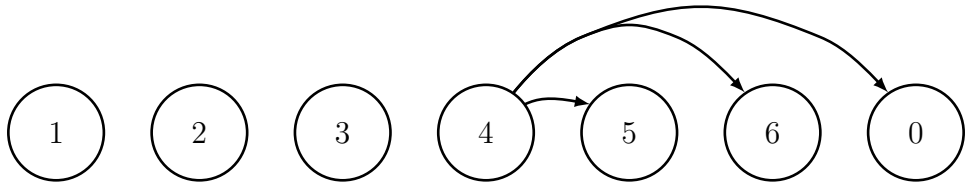

In practice, given the observed matching outcomes and students' priority scores, one can always determine the equilibrium cutoffs using equation \eqref{eq:equilibrium_cutoff} and then construct the partial order $\succ^{P_s}_\omega$ for each student accordingly.

\begin{remark}
	Although it provides straightforward implications, the stability assumption has two notable limitations: (\emph{i}) It is often not very informative without other assumptions. Indeed, stability alone only reveals preference relations between the matched program and other feasible programs. As a result, it generates a limited number of binary relations compared to the full preference order $\succ^Q_\omega$, particularly for students who have lower priority scores and thus smaller feasible sets. (\emph{ii}) 
Stability is a desirable property of the equilibrium outcome, not a primitive assumption about the model—such as assumptions on student behavior. While \cite{fack_beyond_2019} provides an asymptotic justification for stability, researchers may often prefer to impose alternative assumptions on how students behave within a given matching environment.

Importantly, our theoretical results do not rely on the stability assumption. Rather, we develop a flexible framework capable of incorporating various assumptions regarding revealed preferences, with stability being just one possible choice. This flexibility enables researchers to adapt our approach across diverse empirical contexts.
\end{remark}

\subsection{Undominated Strategy}\label{sec:undominance}
A widely used assumption in the theoretical matching literature is that students submit undominated ROLs. A student's assignment depends on the ROLs and priority scores of all students via the induced equilibrium cutoffs and hence her feasible set, so it is natural to rule out dominated submissions. Intuitively, a ROL is dominated if there exists another ROL that yields a weakly better assignment for every potential feasible set and a strictly better assignment for at least one feasible set, according to the student's true preferences. We formalize this notion below.

\begin{definition}\label{def:undominance}
Given student $\omega$'s preference ordering $\succ^Q_\omega$, a ROL $\succ^R_\omega$ is \emph{dominated} by another ROL $\succ^{R'}_\omega$ if the following two conditions hold:
\begin{enumerate}
    \item For every feasible set $F\subseteq \mathcal{J}_0$ with $0\in F$, we have either $ \optchoice{\succ^{R'}_\omega}{F} = \optchoice{\succ^{R}_\omega}{F}$ or $\optchoice{\succ^{R'}_\omega}{F} \succ^Q_\omega \optchoice{\succ^{R}_\omega}{F}$.
    \item There exists at least one feasible set $F\subseteq \mathcal{J}_0$ with $0\in F$ for which $\optchoice{\succ^{R'}_\omega}{F} \succ^Q_\omega \optchoice{\succ^{R}_\omega}{F}$.
\end{enumerate}
A ROL $\succ^R_\omega$ is \emph{undominated} if no other ROL dominates it. A ROL $\succ^R_\omega$ is \emph{dominant} if it dominates every other possible ROL.
\end{definition}

When there is no constraint on the number of programs listed in students' ROLs, \cite{dubins_machiavelli_1981} and \cite{roth_economics_1982} show that truthfully reporting ROLs as students' true preferences is a dominant strategy under the DA mechanism. This property, known as \emph{strategy-proofness}, is a fundamental characteristic of the DA mechanism.\footnote{See also \cite{azevedo_strategy-proofness_2018}, who advocates a robust notion of strategy-proofness in large economies.} In practice, however, many real-world implementations of DA impose nontrivial upper bounds on $|\succ^R_\omega|$, restricting students to submit ROLs of length $|\succ^R_\omega|\le K$ for some predetermined constant $K < |\mathcal{J}_0|$. Under such constraints, the DA mechanism is no longer strategy-proof, and no dominant strategy exists.\footnote{Strategy-proofness also fails if, instead of a length constraint, students incur application costs dependent on the number of schools listed \cite{fack_beyond_2019}.}

Although no dominant strategy exists under constrained DA mechanisms, \emph{dominated} strategies still exist. \cite{haeringer_constrained_2009} provides a detailed characterization of these dominated strategies. Specifically, Lemma 4.2 in \cite{haeringer_constrained_2009} establishes that if a mechanism is strategy-proof in the absence of constraints on $|\succ^R_\omega|$, then in settings with a constraint $|\succ^R_\omega| \le K < |\mathcal{J}_0|$, any ROL $\succ^R_\omega$ incompatible with the true preference ordering $\succ^Q_\omega$ is dominated by another ROL $\succ^{R'}_\omega$ that lists the same set of programs according to $\succ^Q_\omega$. Moreover, if $|\succ^{R'}_\omega| = K$, the ROL $\succ^{R'}_\omega$ represents an undominated strategy. Using such a strategy is known as the ``dropping strategy'' in \cite{kojima_incentives_2009}.

Let us look at some examples before proceeding.
\begin{example}\label{ex:undominance}
Suppose $\mathcal{J}_0 = \{0, 1, 2, \dots, 6\}$ and $K = 5$. Consider a student $\omega$ whose true preference ordering is given by:
\[
1 \succ^Q_\omega 2 \succ^Q_\omega 3 \succ^Q_\omega 4 \succ^Q_\omega 5 \succ^Q_\omega 0 \succ^Q_\omega 6.
\]

Consider the following scenarios:
\begin{itemize}
\item The ROL $3 \succ^R_\omega 2 \succ^R_\omega 4 \succ^R_\omega 5 \succ^R_\omega 0$ is dominated by the alternative ROL $2 \succ^{R'}_\omega 3 \succ^{R'}_\omega 4 \succ^{R'}_\omega 5 \succ^{R'}_\omega 0$. Indeed, for any feasible set $F$ containing both programs $\{2, 3\}$, we have $\optchoice{\succ^{R'}_\omega}{F} = 2 \succ^Q_\omega 3 = \optchoice{\succ^R_\omega}{F}$. For all other feasible sets, we have $\optchoice{\succ^{R'}_\omega}{F} = \optchoice{\succ^R_\omega}{F}$. Thus, the ROL $\succ^R_\omega$ is dominated by $\succ^{R'}_\omega$. In fact, it can be shown that  $\succ^{R'}_\omega$ is an undominated strategy. 
\item The ROL $2 \succ^R_\omega 3 \succ^R_\omega 4 \succ^R_\omega 0$ is dominated by the alternative ROL $2 \succ^{R'}_\omega 3 \succ^{R'}_\omega 4 \succ^{R'}_\omega 5 \succ^{R'}_\omega 0$. For any feasible set $F$ with $5 \in F$ and $F \cap \{2, 3, 4\} = \emptyset$,  we have $\optchoice{\succ^{R'}_\omega}{F} = 5 \succ^Q_\omega 0 = \optchoice{\succ^R_\omega}{F}$. For all other feasible sets, we have $\optchoice{\succ^{R'}_\omega}{F} = \optchoice{\succ^R_\omega}{F}$. 
\end{itemize}
\end{example}

We formally state the assumption considered in this subsection as follows:

\begin{assumption}[undominated strategies]\label{assu:undominance}
The ROL $\succ^R_\omega$ of each student $\omega$ is undominated.
\end{assumption}

Under Assumption \ref{assu:undominance}, the observed ROLs reveal information about students' true preferences. Specifically, observing an ROL $\succ^R_\omega$ restricts the student's true preference $\succ^Q_\omega$ to the set of total orders that rationalize $\succ^R_\omega$ as an undominated strategy. To formally characterize this set of preferences, we construct a partial order $\succ^{P_u}_\omega$ from the student's submitted ROL $\succ^R_\omega$. Compatibility between this constructed partial order and the student's true preference provides a necessary and sufficient condition for the submitted ROL to be undominated. This is formally stated in the following proposition:

\begin{proposition}[Partial order induced by undominated strategies]\label{Prop:undominated_strategy}
	For each student $\omega$, define a partial order $\succ^{P_u}_\omega$ as follows: \emph{(i)} for every pair $(j_1, j_2)$ with $j_1 \succ^{R}_\omega j_2$, set $j_1 \succ^{P_u}_\omega j_2$; and \emph{(ii)} if $|\succ^R_\omega| < K$, for each program $j$ not listed in the domain of $\succ^R_\omega$, set $0 \succ^{P_u}_\omega j$. Then, Assumption \ref{assu:undominance} holds if and only if, for each student $\omega$, the partial order $\succ^{P_u}_\omega$ is compatible with the student's true preference $\succ^Q_\omega$.
\end{proposition}

To illustrate the partial order constructed in Proposition \ref{Prop:undominated_strategy}, we revisit the running example.

\begin{example}[continues=running_example]
Suppose, as before, $\mathcal{J}_0 = \{0, 1, 2, \dots, 6\}$, $K = 5$, student $\omega$'s feasible set at equilibrium cutoff is $F_{\omega}(c) = \{4, 5, 6, 0\}$, and she is matched to program $4$. Additionally, suppose that the ROL she submitted is: $ 2 \succ^R_\omega 3 \succ^R_\omega 4 \succ^R_\omega 5 \succ^R_\omega 0$. Under the assumption of undominated strategies, the partial order $\succ^{P_u}_\omega$ constructed from her ROL coincides directly with $\succ^R_\omega$. Specifically, we have $2 \succ^{P_u}_\omega 3 \succ^{P_u}_\omega 4 \succ^{P_u}_\omega 5 \succ^{P_u}_\omega 0$. as illustrated in Figure \ref{fig:undominance_illustration}.

Alternatively, consider another student $\omega'$ who submits a shorter ROL: $ 2 \succ^R_{\omega'} 3 \succ^R_{\omega'} 4 \succ^R_{\omega'} 0$. Since $|\succ^R_{\omega'}| < K$, the partial order $\succ^{P_u}_{\omega'}$ constructed from this ROL includes relations placing the outside option above the unlisted programs. Thus, the inferred partial order is: $2 \succ^{P_u}_{\omega'} 3 \succ^{P_u}_{\omega'} 4 \succ^{P_u}_{\omega'} 0 \succ^{P_u}_{\omega'} \{1, 5, 6\}$ as depicted in Figure \ref{fig:undominance_illustration_unbinding}.
\end{example}

\begin{figure}[h]
    \centering
    \begin{tikzpicture}[>=latex,line join=bevel,]
  \pgfsetlinewidth{1bp}
\begin{scope}
  \pgfsetstrokecolor{black}
  \definecolor{strokecol}{rgb}{1.0,1.0,1.0};
  \pgfsetstrokecolor{strokecol}
  \definecolor{fillcol}{rgb}{1.0,1.0,1.0};
  \pgfsetfillcolor{fillcol}
  \filldraw (0.0bp,0.0bp) -- (0.0bp,58.59bp) -- (360.0bp,58.59bp) -- (360.0bp,0.0bp) -- cycle;
\end{scope}
  \pgfsetcolor{black}
  \draw [->] (90.141bp,18.0bp) .. controls (92.131bp,18.0bp) and (94.121bp,18.0bp)  .. (107.51bp,18.0bp);
  \draw [->] (144.14bp,18.0bp) .. controls (146.13bp,18.0bp) and (148.12bp,18.0bp)  .. (161.51bp,18.0bp);
  \draw [->] (198.14bp,18.0bp) .. controls (200.13bp,18.0bp) and (202.12bp,18.0bp)  .. (215.51bp,18.0bp);
  \draw [->] (244.46bp,33.125bp) .. controls (250.83bp,40.951bp) and (259.71bp,49.738bp)  .. (270.0bp,54.0bp) .. controls (284.78bp,60.123bp) and (291.22bp,60.123bp)  .. (306.0bp,54.0bp) .. controls (312.75bp,51.203bp) and (318.9bp,46.457bp)  .. (331.54bp,33.125bp);
\begin{scope}
  \definecolor{strokecol}{rgb}{0.0,0.0,0.0};
  \pgfsetstrokecolor{strokecol}
  \draw (18.0bp,18.0bp) ellipse (18.0bp and 18.0bp);
  \draw (18.0bp,18.0bp) node {$1$};
\end{scope}
\begin{scope}
  \definecolor{strokecol}{rgb}{0.0,0.0,0.0};
  \pgfsetstrokecolor{strokecol}
  \draw (72.0bp,18.0bp) ellipse (18.0bp and 18.0bp);
  \draw (72.0bp,18.0bp) node {$2$};
\end{scope}
\begin{scope}
  \definecolor{strokecol}{rgb}{0.0,0.0,0.0};
  \pgfsetstrokecolor{strokecol}
  \draw (126.0bp,18.0bp) ellipse (18.0bp and 18.0bp);
  \draw (126.0bp,18.0bp) node {$3$};
\end{scope}
\begin{scope}
  \definecolor{strokecol}{rgb}{0.0,0.0,0.0};
  \pgfsetstrokecolor{strokecol}
  \draw (180.0bp,18.0bp) ellipse (18.0bp and 18.0bp);
  \draw (180.0bp,18.0bp) node {$4$};
\end{scope}
\begin{scope}
  \definecolor{strokecol}{rgb}{0.0,0.0,0.0};
  \pgfsetstrokecolor{strokecol}
  \draw (234.0bp,18.0bp) ellipse (18.0bp and 18.0bp);
  \draw (234.0bp,18.0bp) node {$5$};
\end{scope}
\begin{scope}
  \definecolor{strokecol}{rgb}{0.0,0.0,0.0};
  \pgfsetstrokecolor{strokecol}
  \draw (288.0bp,18.0bp) ellipse (18.0bp and 18.0bp);
  \draw (288.0bp,18.0bp) node {$6$};
\end{scope}
\begin{scope}
  \definecolor{strokecol}{rgb}{0.0,0.0,0.0};
  \pgfsetstrokecolor{strokecol}
  \draw (342.0bp,18.0bp) ellipse (18.0bp and 18.0bp);
  \draw (342.0bp,18.0bp) node {$0$};
\end{scope}
\end{tikzpicture}

     \vspace{-2em}
    \caption{Illustration of $\succ^{P_u}_\omega$ when $|\succ^R_\omega| = K$}
    \label{fig:undominance_illustration}
\end{figure}

\begin{figure}[h]
    \centering
    \begin{tikzpicture}[>=latex,line join=bevel,]
  \pgfsetlinewidth{1bp}
\begin{scope}
  \pgfsetstrokecolor{black}
  \definecolor{strokecol}{rgb}{1.0,1.0,1.0};
  \pgfsetstrokecolor{strokecol}
  \definecolor{fillcol}{rgb}{1.0,1.0,1.0};
  \pgfsetfillcolor{fillcol}
  \filldraw (0.0bp,0.0bp) -- (0.0bp,96.88bp) -- (360.0bp,96.88bp) -- (360.0bp,0.0bp) -- cycle;
\end{scope}
  \pgfsetcolor{black}
  \draw [->] (90.141bp,49.396bp) .. controls (92.131bp,49.396bp) and (94.121bp,49.396bp)  .. (107.51bp,49.396bp);
  \draw [->] (144.14bp,49.396bp) .. controls (146.13bp,49.396bp) and (148.12bp,49.396bp)  .. (161.51bp,49.396bp);
  \draw [->] (190.46bp,64.521bp) .. controls (196.83bp,72.347bp) and (205.71bp,81.135bp)  .. (216.0bp,85.396bp) .. controls (252.96bp,100.7bp) and (269.04bp,100.7bp)  .. (306.0bp,85.396bp) .. controls (312.75bp,82.6bp) and (318.9bp,77.854bp)  .. (331.54bp,64.521bp);
  \draw [->] (331.54bp,34.272bp) .. controls (325.17bp,26.446bp) and (316.29bp,17.658bp)  .. (306.0bp,13.396bp) .. controls (254.26bp,-8.0339bp) and (108.5bp,0.53026bp)  .. (54.0bp,13.396bp) .. controls (41.487bp,16.35bp) and (28.364bp,15.116bp)  .. (18.0bp,31.396bp);
  \draw [->] (331.54bp,34.272bp) .. controls (325.17bp,26.446bp) and (316.29bp,17.658bp)  .. (306.0bp,13.396bp) .. controls (279.45bp,2.4008bp) and (244.52bp,-0.80895bp)  .. (234.0bp,31.396bp);
  \draw [->] (327.07bp,38.728bp) .. controls (317.38bp,31.801bp) and (304.77bp,23.851bp)  .. (288.0bp,30.228bp);
\begin{scope}
  \definecolor{strokecol}{rgb}{0.0,0.0,0.0};
  \pgfsetstrokecolor{strokecol}
  \draw (18.0bp,49.4bp) ellipse (18.0bp and 18.0bp);
  \draw (18.0bp,49.396bp) node {$1$};
\end{scope}
\begin{scope}
  \definecolor{strokecol}{rgb}{0.0,0.0,0.0};
  \pgfsetstrokecolor{strokecol}
  \draw (72.0bp,49.4bp) ellipse (18.0bp and 18.0bp);
  \draw (72.0bp,49.396bp) node {$2$};
\end{scope}
\begin{scope}
  \definecolor{strokecol}{rgb}{0.0,0.0,0.0};
  \pgfsetstrokecolor{strokecol}
  \draw (126.0bp,49.4bp) ellipse (18.0bp and 18.0bp);
  \draw (126.0bp,49.396bp) node {$3$};
\end{scope}
\begin{scope}
  \definecolor{strokecol}{rgb}{0.0,0.0,0.0};
  \pgfsetstrokecolor{strokecol}
  \draw (180.0bp,49.4bp) ellipse (18.0bp and 18.0bp);
  \draw (180.0bp,49.396bp) node {$4$};
\end{scope}
\begin{scope}
  \definecolor{strokecol}{rgb}{0.0,0.0,0.0};
  \pgfsetstrokecolor{strokecol}
  \draw (234.0bp,49.4bp) ellipse (18.0bp and 18.0bp);
  \draw (234.0bp,49.396bp) node {$5$};
\end{scope}
\begin{scope}
  \definecolor{strokecol}{rgb}{0.0,0.0,0.0};
  \pgfsetstrokecolor{strokecol}
  \draw (342.0bp,49.4bp) ellipse (18.0bp and 18.0bp);
  \draw (342.0bp,49.396bp) node {$0$};
\end{scope}
\begin{scope}
  \definecolor{strokecol}{rgb}{0.0,0.0,0.0};
  \pgfsetstrokecolor{strokecol}
  \draw (288.0bp,49.4bp) ellipse (18.0bp and 18.0bp);
  \draw (288.0bp,49.396bp) node {$6$};
\end{scope}
\end{tikzpicture}

     \vspace{-2em}
     \caption{Illustration of $\succ^{P_u}_{\omega'}$ when $|\succ^R_{\omega'}| < K$}
    \label{fig:undominance_illustration_unbinding}
\end{figure}

The partial order $\succ^{P_u}_\omega$ implied by the assumption of undominated strategies relies exclusively on the submitted ROLs, due to the definition provided in Definition \ref{def:undominance}. In practice, one can construct $\succ^{P_u}_\omega$ for each student as in the above example, provided the ROLs are observed.


\begin{remark}\label{remark:undominance}
	Assumption \ref{assu:undominance} is quite powerful, as the implied partial order $\succ^{P_u}_\omega$ conveys substantial information about their true preferences. Specifically, when the length constraint is not binding (i.e., $|\succ^R_\omega| < K$), the partial order $\succ^{P_u}_\omega$ derived from Assumption \ref{assu:undominance} is always at least as informative as, and typically strictly more informative than, the partial order $\succ^{P_s}_\omega$ induced by stability alone.

However,  Assumption \ref{assu:undominance} can also be too restrictive, as it implicitly assumes that the student believes all feasible sets $F$ are possible, and therefore she must consider every potential feasible set. To illustrate, reconsider student $\omega$ in Example \ref{ex:undominance}. The ROL $2 \succ^R_\omega 3 \succ^R_\omega 4 \succ^R_\omega 0$ is strictly dominated by the alternative ROL $2 \succ^{R'}_\omega 3 \succ^{R'}_\omega 4 \succ^{R'}_\omega 5 \succ^{R'}_\omega 0$ only for feasible sets that include program $5$ but exclude programs $2$, $3$, and $4$. If student $\omega$ believes such feasible sets will not arise at equilibrium, she might reasonably be indifferent between $\succ^R_\omega$ and $\succ^{R'}_\omega$, potentially choosing the shorter ROL. In Section \ref{sec:more_revealed_preference}, we introduce a more robust variant of the undominated strategy assumption to address this concern.

\end{remark}

\subsection{Combine multiple assumptions}
If researchers are willing to impose multiple assumptions simultaneously, our approach can readily combine these assumptions to extract more information about true preferences from the data. In this section, we illustrate this idea by jointly utilizing the two assumptions discussed earlier. Similar techniques can be applied to other assumptions introduced later in Section \ref{sec:more_revealed_preference}.

From Proposition \ref{Prop:PS}, we know that the set of preferences for which a matching outcome $\mu$ is stable coincides exactly with the set of preferences compatible with the partial order $\succ^{P_s}_\omega$. Similarly, Proposition \ref{Prop:undominated_strategy} indicates that the set of preferences rationalizing $\succ^R_\omega$ as an undominated strategy is equal to the set of preferences compatible with the partial order $\succ^{P_u}_\omega$. Therefore, the set of preferences satisfying both the stability and undominated strategy assumptions simultaneously corresponds to the preferences compatible with both $\succ^{P_s}_\omega$ and $\succ^{P_u}_\omega$. According to the following theorem, this set can be characterized by the join of the two partial orders, denoted by $\succ^{P_s}_\omega \vee \succ^{P_u}_\omega$. The formal definition of the join operation is provided in the notation and preliminaries section.

\begin{lemma}\label{thm:join}
Let $\succ^{P}$ and $\succ^{P'}$ be two partial orders. If $\succ^{P}$ and $\succ^{P'}$ are compatible, then a total order $\succ$ is compatible with both partial orders if and only if $\succ$ is compatible with their join $\succ^{P}\vee \succ^{P'}$. If $\succ^{P}$ and $\succ^{P'}$ are not compatible, then there does not exist a total order $\succ$ compatible with both.
\end{lemma}

We formally establish that the partial orders $\succ^{P_s}_\omega$ and $\succ^{P_u}_\omega$ are always compatible with each other in the following proposition. The second part of the proposition follows immediately as a corollary from Lemma \ref{thm:join}, Proposition \ref{Prop:PS}, and Proposition \ref{Prop:undominated_strategy}.

\begin{proposition}\label{Prop:PS_undominated_strategy}
	For each student $\omega \in \Omega$, the partial orders $\succ^{P_s}_\omega$ and $\succ^{P_u}_\omega$ are always compatible. Moreover, Assumptions \ref{assu:stability} and \ref{assu:undominance} hold if and only if, for each student $\omega$, preference $\succ^Q_\omega$ is compatible with the join $\succ^{P_s}_\omega \vee \succ^{P_u}_\omega$.
\end{proposition}

We revisit the running example to illustrate the joined partial order $\succ^{P_s}_\omega \vee \succ^{P_u}_\omega$ established in the above proposition.

\begin{example}[continues=running_example]
Recall that for student $\omega$:
\begin{itemize}
    \item her feasible set at the equilibrium cutoff is $F_{\omega}(c) = \{4, 5, 6, 0\}$,
    \item she is matched to program $4$,
    \item her ROL is $2 \succ^R_\omega 3 \succ^R_\omega 4 \succ^R_\omega 5 \succ^R_\omega 0$.
\end{itemize}
In this scenario, the combined partial order $\succ^{P_s}_\omega \vee \succ^{P_u}_\omega$ is depicted in Figure \ref{fig:undominance_illustration_stability}. Most binary relations in $\succ^{P_s}_\omega \vee \succ^{P_u}_\omega$ are inherited from the partial order $\succ^{P_u}_\omega$. The stability-induced partial order $\succ^{P_s}_\omega$ contributes the additional relation $4 \succ^{P_s}_\omega 6$, because program $6$ is included in the feasible set but not listed in her ROL.
\end{example}

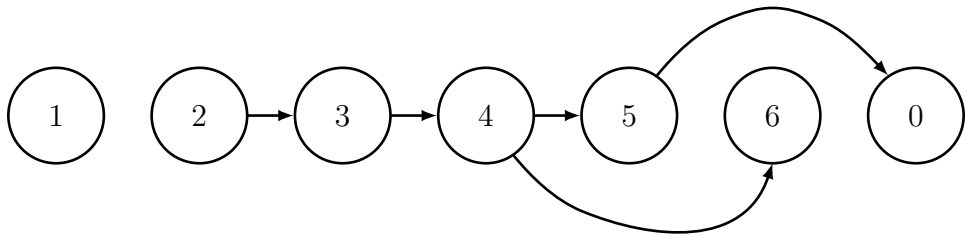
\begin{figure}[h]
    \centering
    \begin{tikzpicture}[>=latex,line join=bevel,]
  \pgfsetlinewidth{1bp}
\begin{scope}
  \pgfsetstrokecolor{black}
  \definecolor{strokecol}{rgb}{1.0,1.0,1.0};
  \pgfsetstrokecolor{strokecol}
  \definecolor{fillcol}{rgb}{1.0,1.0,1.0};
  \pgfsetfillcolor{fillcol}
  \filldraw (0.0bp,0.0bp) -- (0.0bp,85.41bp) -- (360.0bp,85.41bp) -- (360.0bp,0.0bp) -- cycle;
\end{scope}
  \pgfsetcolor{black}
  \draw [->] (90.141bp,44.815bp) .. controls (92.131bp,44.815bp) and (94.121bp,44.815bp)  .. (107.51bp,44.815bp);
  \draw [->] (144.14bp,44.815bp) .. controls (146.13bp,44.815bp) and (148.12bp,44.815bp)  .. (161.51bp,44.815bp);
  \draw [->] (198.14bp,44.815bp) .. controls (200.13bp,44.815bp) and (202.12bp,44.815bp)  .. (215.51bp,44.815bp);
  \draw [->] (190.46bp,29.69bp) .. controls (196.83bp,21.865bp) and (205.71bp,13.077bp)  .. (216.0bp,8.8152bp) .. controls (242.55bp,-2.1804bp) and (277.48bp,-5.3901bp)  .. (288.0bp,26.815bp);
  \draw [->] (244.46bp,59.94bp) .. controls (250.83bp,67.766bp) and (259.71bp,76.553bp)  .. (270.0bp,80.815bp) .. controls (284.78bp,86.938bp) and (291.22bp,86.938bp)  .. (306.0bp,80.815bp) .. controls (312.75bp,78.018bp) and (318.9bp,73.272bp)  .. (331.54bp,59.94bp);
\begin{scope}
  \definecolor{strokecol}{rgb}{0.0,0.0,0.0};
  \pgfsetstrokecolor{strokecol}
  \draw (18.0bp,44.82bp) ellipse (18.0bp and 18.0bp);
  \draw (18.0bp,44.815bp) node {$1$};
\end{scope}
\begin{scope}
  \definecolor{strokecol}{rgb}{0.0,0.0,0.0};
  \pgfsetstrokecolor{strokecol}
  \draw (72.0bp,44.82bp) ellipse (18.0bp and 18.0bp);
  \draw (72.0bp,44.815bp) node {$2$};
\end{scope}
\begin{scope}
  \definecolor{strokecol}{rgb}{0.0,0.0,0.0};
  \pgfsetstrokecolor{strokecol}
  \draw (126.0bp,44.82bp) ellipse (18.0bp and 18.0bp);
  \draw (126.0bp,44.815bp) node {$3$};
\end{scope}
\begin{scope}
  \definecolor{strokecol}{rgb}{0.0,0.0,0.0};
  \pgfsetstrokecolor{strokecol}
  \draw (180.0bp,44.82bp) ellipse (18.0bp and 18.0bp);
  \draw (180.0bp,44.815bp) node {$4$};
\end{scope}
\begin{scope}
  \definecolor{strokecol}{rgb}{0.0,0.0,0.0};
  \pgfsetstrokecolor{strokecol}
  \draw (234.0bp,44.82bp) ellipse (18.0bp and 18.0bp);
  \draw (234.0bp,44.815bp) node {$5$};
\end{scope}
\begin{scope}
  \definecolor{strokecol}{rgb}{0.0,0.0,0.0};
  \pgfsetstrokecolor{strokecol}
  \draw (288.0bp,44.82bp) ellipse (18.0bp and 18.0bp);
  \draw (288.0bp,44.815bp) node {$6$};
\end{scope}
\begin{scope}
  \definecolor{strokecol}{rgb}{0.0,0.0,0.0};
  \pgfsetstrokecolor{strokecol}
  \draw (342.0bp,44.82bp) ellipse (18.0bp and 18.0bp);
  \draw (342.0bp,44.815bp) node {$0$};
\end{scope}
\end{tikzpicture}

     \vspace{-2em}
     \caption{Illustration of $\succ^{P_s}_\omega \vee \succ^{P_u}_\omega$}
    \label{fig:undominance_illustration_stability}
\end{figure}

\section{(Set) Identification of Counterfactual Stable Matchings}\label{sec:counterfactual}
In the previous section, we introduced methods to partially identify students' true preferences by constructing partial orders consistent with observed data. In this section, we explore how these identified preference bounds can be leveraged to derive informative bounds on matching outcomes in various counterfactual scenarios. At the end of the section, we compare our proposed approach to existing methods, highlighting key differences and potential advantages.

Since computing stable matchings in counterfactual scenarios is practically feasible only with a finite number of students, we henceforth focus on the finite-population case, letting $N$ denote the number of students and $\Omega_N = \{1, ..., N\}$ as the set of students. In this setting, it is more natural to represent program capacities as integers. Specifically, we use $q^N_j$ to denote the maximum number of students program $j$ can enroll. As $j=0$ denotes the outside option, $q^N_0 = N$. We let $q^N \coloneqq (q^N_j : j \in \mathcal{J}_0)$. 

The methods introduced in Section \ref{sec:revealed_preference} and later in Section \ref{sec:more_revealed_preference} can all be viewed as procedures that generate, for each student $\omega$, a set of preferences $\mathcal{C}_\omega$ containing the true preference $\succ^Q_\omega$. For instance, under the stability assumption, $\mathcal{C}_\omega$ consists of all total orders compatible with the partial order $\succ^{P_s}_\omega$. To maintain a unified treatment independent of any specific assumption, we treat $\mathcal{C}_\omega$ as given throughout this section without specifying the exact assumptions used to derive it. For convenience, we introduce the shorthand notations $\mathcal{C}^{N} \coloneqq (\mathcal{C}_{\omega} : \omega \in \Omega_N)$.

We focus on matching outcomes under counterfactual changes in programs' capacity constraints and the rules determining priority scores.\footnote{Although perhaps less empirically relevant, our approach can be readily extended to handle counterfactual changes in students' true preferences as well, as long as the mapping from students' preferences to their counterfactual preferences is specified.} We use the notation $\tilde{S}\coloneqq (\tilde{S}_{\omega j}: j\in \mathcal{J}, \omega\in \Omega_N)$ and $\tilde{q}^N \coloneqq (\tilde{q}^N_j: j\in \mathcal{J})$ to denote counterfactual priority scores and program capacities, respectively. This framework encompasses a broad class of counterfactual scenarios, including the ones frequently employed to evaluate affirmative action policies in centralized college admission systems. Notable examples from the existing literature include \cite{barahona2021}, \cite{ngo_preferences_2024}, and \cite{SalesMoses2014}, among many others. To facilitate examining counterfactual matching outcomes across distinct groups of students, we assume the availability of student-level characteristics $X_\omega$ for each student $\omega$.

To formally define the parameters of interest under the counterfactual scenario, we introduce the notion of a \emph{matching matrix} $d = (d_{\omega j}: \omega\in \Omega_N, j\in \mathcal{J}_0)$, where $d_{\omega j} = 1$ indicates that student $\omega$ is matched to program $j$, and $d_{\omega j} = 0$ otherwise. A valid matching matrix must satisfy all the conditions specified in the following definition.

\begin{definition}[Matching Matrix]\label{def:MM}
    An $N \times |\mathcal{J}_0|$ matrix $d$ is called a matching matrix if it satisfies the following conditions:
    \begin{enumerate}
        \item [(1)] $d_{\omega j} \in \{0, 1\}$ for all $\omega \in \Omega_N$ and $j \in \mathcal{J}_0$.
        \item [(2)] Each student is matched to exactly one program in $\mathcal{J}_0$: for each $\omega=1,...,N$, $\sum_{j \in \mathcal{J}_0} d_{\omega j} = 1$.
	\item [(3)] Each program does not exceed its capacity: for each $j\in \mathcal{J}_0$, $\sum_{\omega \in \Omega_N} d_{\omega j} \leq \tilde{q}^N_j$.  
    \end{enumerate}
\end{definition}

We are particularly interested in stable matchings under counterfactual scenarios given the bounds $\mathcal{C}^N$ on students' preferences. The formal definition of the set of counterfactual stable matchings is as follows:

\begin{definition}[Stable Matchings in the Counterfactual]
Let $\mathcal{D}(\mathcal{C}^N, \tilde{q}^N, \tilde{S})$ denote the set of all matching matrices $d$ for which there exist preference $(\succ_\omega: \omega \in \Omega_N)$ satisfying $\succ_\omega \in \mathcal{C}_\omega$ for each student $\omega \in \Omega_N$, such that $d$ is a stable matching given these preferences $(\succ_\omega: \omega \in \Omega_N)$, counterfactual priority scores $\tilde{S}$, and counterfactual capacities $\tilde{q}^N$.
\end{definition}

Given an arbitrary weighting matrix $\gamma = (\gamma_{\omega j}: \omega \in \Omega_N,\, j \in \mathcal{J}_0)$, we define the following optimization problems:
\begin{equation}\label{eq:counterfactual_bound}
	\underline{\theta} \coloneqq \min_{d \in \mathcal{D}(\mathcal{C}^N, \tilde{q}^N, \tilde{S})}
	\sum_{\omega\in \Omega_N} \sum_{j\in \mathcal{J}_0} \gamma_{\omega j} d_{\omega j}
	\quad\text{ and }\quad
	\overline{\theta} \coloneqq \max_{d \in \mathcal{D}(\mathcal{C}^N, \tilde{q}^N, \tilde{S})}
	\sum_{\omega\in \Omega_N} \sum_{j\in \mathcal{J}_0} \gamma_{\omega j} d_{\omega j}.
\end{equation}
Since each student $\omega$'s preferences are bounded within the set $\mathcal{C}_\omega$, the interval $[\underline{\theta}, \overline{\theta}]$ provides valid bounds for the parameter $\theta \coloneqq \sum_{\omega\in \Omega_N} \sum_{j\in \mathcal{J}_0} \gamma_{\omega j} d_{\omega,j}$
across all stable matchings $d$ in the counterfactual scenario. Moreover, if each $\mathcal{C}_\omega$ constitutes a sharp preference bound, such as the one derived in Section \ref{sec:revealed_preference}, then the interval $[\underline{\theta}, \overline{\theta}]$ is a sharp bound for $\theta$.

By choosing $\gamma$ appropriately, we can use $\theta$ and its bounds to represent a variety of parameters of interest in counterfactual analyses. We provide two illustrative examples below:

\begin{itemize}
	\item Let $\mathcal{A}$ be a subset of programs within $\mathcal{J}_0$, and let $\mathcal{B}$ be a subset of students' observable characteristics. For instance, $\mathcal{A}$ could represent all STEM programs, and $\mathcal{B}$ could represent all female students. If we define $\gamma_{\omega j} = \indicator(j\in \mathcal{A},\,X_\omega\in \mathcal{B}) - \indicator(j\in \mathcal{A},\,X_\omega \notin \mathcal{B})$, then
  $\theta$ captures the gender gap in STEM assignment under the counterfactual scenario—specifically, the difference between the number of female and male students allocated to STEM programs.

	\item Suppose we define $\gamma_{\omega j} = \indicator\left(\exists\, \succ_\omega\in \mathcal{C}_\omega:\, j \succ_\omega \mu(\omega)\right)$ where $\mu(\omega)$ is the realized matching in the data. Then, $\overline{\theta}$ provides a valid upper bound for the number of students who prefer their matchings in the counterfactual scenario over their current matchings. Similarly, defining $\gamma_{\omega j} = \indicator\left(\forall\, \succ_\omega\in \mathcal{C}_\omega:\, j \succ_\omega \mu(\omega)\right)$  yields an $\underline{\theta}$ that represents a valid lower bound for the number of students who  prefer their counterfactual assignment, thus providing a lower bound for welfare improvements. Moreover, one can show that they are sharp upper and lower bounds if each $\mathcal{C}_\omega$ is a sharp bound for preferences.
\end{itemize}

\begin{remark}
    Our approach focuses on bounding stable matchings directly in the counterfactual scenario, rather than bounding matchings generated explicitly from students' submitted ROLs under the counterfactual. This choice is justified by the theoretical results in \cite{fack_beyond_2019}, which demonstrate that the matching produced by a DA mechanism converges to a stable matching as the number of students $N$ grows large. 
\end{remark}

\subsection{Characterization of Counterfactual Stable Matchings}
It is computationally infeasible to obtain the bounds defined in equation \eqref{eq:counterfactual_bound} by enumerating all possible elements of the set $\mathcal{D}(\mathcal{C}^N, \tilde{q}^N, \tilde{S})$. To facilitate computation, we first analyze the structure of $\mathcal{D}(\mathcal{C}^N, \tilde{q}^N, \tilde{S})$. In the following theorem, we show that every matching matrix $d \in \mathcal{D}(\mathcal{C}^N, \tilde{q}^N, \tilde{S})$ must satisfy a system of linear inequalities.

\begin{theorem}\label{thm:linear_charact_stable}
    If a matching matrix $d$ is an element of $\mathcal{D}(\mathcal{C}^N, \tilde{q}^N, \tilde{S})$, then it must satisfy the following linear constraints:
    \begin{equation}\label{eq:simple_stable_matching_general}
        \forall \omega \in \Omega_N,\; \forall j\in \mathcal{J}_0,\quad 
        \tilde{q}^N_j d_{\omega j} 
        + \tilde{q}^N_j \sum_{l:\,\exists\,\succ_{\omega}\in\mathcal{C}_{\omega}\text{ s.t. }l\succ_{\omega} j} d_{\omega l} 
        + \sum_{\omega':\,\tilde{S}_{\omega' j}>\tilde{S}_{\omega j}} d_{\omega' j}
        \geq \tilde{q}^N_j,
    \end{equation}
    where the summation in the second term in \eqref{eq:simple_stable_matching_general} is taken over all programs $l$ that could be ranked higher than program $j$ according to at least one preference ordering in $\mathcal{C}_\omega$, and the summation in the third term is taken over all students whose priority score at program $j$ exceeds that of student $\omega$.
\end{theorem}

Theorem \ref{thm:linear_charact_stable} can be viewed as an extension of Lemma 1 in \cite{Baiou2000}. The key distinction is that Lemma 1 in \cite{Baiou2000} characterizes stable matchings under conditions where the true preferences are known, whereas we consider scenarios in which preferences are only set-identified within the sets $\mathcal{C}_\omega$. When every $\mathcal{C}_\omega$ is derived from a partial order\footnote{Without imposing any structural assumptions on the sets $\mathcal{C}_\omega$, obtaining tractable necessary and sufficient conditions for $d \in \mathcal{D}(\mathcal{C}^N, \tilde{q}^N, \tilde{S})$ seems to be challenging.}, such as those discussed in Section \ref{sec:revealed_preference}, we can strengthen Theorem \ref{thm:linear_charact_stable} to obtain a sharp characterization. Specifically, under these conditions, the linear inequalities in \eqref{eq:simple_stable_matching_general} become necessary and sufficient, as stated formally below.

\begin{theorem}\label{thm:iff_stable_matching}
Suppose that, for each student $\omega$, the set $\mathcal{C}_\omega$ consists of all total orders compatible with a partial order $\succ^P_\omega$. Then, a matching matrix $d$ is an element of $\mathcal{D}(\mathcal{C}^N, \tilde{q}^N, \tilde{S})$ if and only if it satisfies the constraints in \eqref{eq:simple_stable_matching_general}. 

Moreover, under these conditions, the constraints in \eqref{eq:simple_stable_matching_general} simplify to:
\begin{equation}\label{eq:simple_stable_matching_partial_order}
    \forall \omega \in \Omega_N,\; \forall j\in \mathcal{J}_0,\quad 
    \tilde{q}^N_j d_{\omega j} + \tilde{q}^N_j \sum_{l:\, l\neq j\text{ and no }j\succ^P_\omega l} d_{\omega l} 
    + \sum_{\omega':\,\tilde{S}_{\omega' j}>\tilde{S}_{\omega j}} d_{\omega' j}
    \geq \tilde{q}^N_j,
\end{equation}
where the summation in the second term of \eqref{eq:simple_stable_matching_partial_order} includes all programs $l$ other than $j$ for which there is no relation $j \succ^P_\omega l$.
\end{theorem}

Combining the definition of a matching matrix (Definition \ref{def:MM}) and the results from Theorem \ref{thm:linear_charact_stable}, we can calculate a lower bound for the parameter $\theta$ by solving the following optimization problem:
\begin{eqnarray}
	\underline{\theta}^\dagger = &\min_{d} & \sum_{\omega\in \Omega_N} \sum_{j\in \mathcal{J}_0} \gamma_{\omega, j} d_{\omega,j} \label{eq:ILP}\\
&\text{s.t.} & \forall \omega\in \Omega_N,\; \forall j\in \mathcal{J}_0,\quad d_{\omega, j} \in \{0, 1\}\nonumber\\
& & \forall \omega\in \Omega_N,\quad \sum_{j\in \mathcal{J}_0} d_{\omega,j} = 1.\nonumber\\
&& \forall j\in \mathcal{J}_0,\quad \sum_{\omega\in \Omega_N}q_{\omega, j} < \tilde{q}_j,\nonumber\\
&&\forall \omega \in \Omega_N,\; \forall j\in \mathcal{J}_0,\quad 
        \tilde{q}^N_j d_{\omega j} + \tilde{q}^N_j \sum_{l:\,\exists\,\succ_{\omega}\in\mathcal{C}_{\omega}\text{ s.t. }l\succ_{\omega} j} d_{\omega l} + \sum_{\omega':\,\tilde{S}_{\omega' j}>\tilde{S}_{\omega j}} d_{\omega' j}
        \geq \tilde{q}^N_j.\nonumber
\end{eqnarray}
Here, the first three constraints directly follow from the definition of a matching matrix, while the final constraint is the same as equation \eqref{eq:simple_stable_matching_general}. By Theorem \ref{thm:linear_charact_stable}, we immediately have that $\underline{\theta}^{\dagger} \le \underline{\theta}$. Furthermore, given the additional structures specified in Theorem \ref{thm:iff_stable_matching}, the equality $\underline{\theta}^{\dagger} = \underline{\theta}$ holds, ensuring that the bound is sharp. The upper bound $\overline{\theta}$ can be computed similarly.

We now discuss the computational complexity of solving the optimization problem given by \eqref{eq:ILP}. The first challenge arises because the problem is an integer linear programming (ILP) problem,
\footnote{\cite{Baiou2000} develops a sophisticated linear programming (LP) formulation whose optimal solution always coincides with that of the corresponding ILP when the exact preference order $\succ^Q_\omega$ is known. Unfortunately, their formulation significantly increases the number of linear constraints, which we find it computationally impractical to solve even with known preferences. Extending their results to our set-identified scenario is also nontrivial and left for future research.}
 which is inherently more computationally intensive than linear programming (LP) problems. A common approach to mitigate complexity is to relax the integer constraints $d_{\omega, j}\in \{0, 1\}$ to continuous constraints $d_{\omega, j}\in [0, 1]$. Solving this relaxed LP formulation yields a conservative lower bound for $\underline{\theta}^\dagger$. However, modern integer programming solvers employ advanced techniques, such as cutting-plane methods, branch-and-bound, and branch-and-cut algorithms, to generate significantly tighter bounds for integer solutions. Therefore, we recommend leveraging these specialized integer-programming algorithms instead of solely working with the relaxed LP problem.

The second issue pertains to the sheer scale of the optimization problem presented in \eqref{eq:ILP}. It involves $N + |\mathcal{J}_0| + N \times |\mathcal{J}_0|$ linear constraints and $N \times |\mathcal{J}_0|$ decision variables. When the matching system has a moderate number of programs and students, for instance, $|\mathcal{J}_0| \approx 20$ and $N \approx 1,000$, this scale remains tractable using modern ILP or LP solvers. However, for larger-scale matching markets, the optimization problem quickly becomes very high-dimensional. To illustrate, consider a scenario with $N = 10,000$ students and $|\mathcal{J}_0| = 1,000$ programs. In this case, the decision matrix $d$ has 10 million dimensions, and the constraints involve approximately 10 million inequalities. Simply storing all the nonzero elements of these constraints would require more than 410 GB of memory. Without additional strategies to simplify or reduce the dimensionality of this problem, directly solving it would be computationally infeasible, even with state-of-the-art ILP and LP solvers. These considerations motivate the dimension-reduction technique we develop in the next section, which significantly reduces the active dimensions without losing any relevant information.

\subsection{Dimension reduction and generalized Gale-Shapley DA algorithm} \label{section: Dimension-Red}
The key to reducing the dimensionality in equation \eqref{eq:ILP} lies in the following observation: there exist certain student-program pairs that are never matched in any stable matching in the counterfactual scenario. If we can identify these student-program pairs, we can directly set the corresponding $d_{\omega j}$ to zero, thereby effectively reducing the dimensionality of the optimization problem. Similarly, certain student-program pairs are always matched across all stable matchings. Identifying them allows us to fix their corresponding decision variables at one.

Before illustrating how this dimension-reduction approach applies to set-identified preferences, it is instructive to first consider the simpler scenario in which students' preferences are exactly known. Note that the standard Gale-Shapley DA algorithm can be used to construct bounds for equilibrium cutoffs associated with stable matchings. Specifically, the DA algorithm with programs proposing to students yields the highest equilibrium cutoffs (upper bounds), whereas the DA algorithm with students proposing to programs generates the lowest equilibrium cutoffs (lower bounds). 

The DA algorithm with programs proposing can be described in Algorithm~\ref{alg:DA} using the cutoff characterization. In words, the algorithm involves the following main steps:

\begin{algorithm}[h!]
\caption{Deferred Acceptance with Program-Proposing to Students}\label{alg:DA}
\begin{algorithmic}
\Require Preferences $(\succ^Q_\omega: \omega\in \Omega_N)$, scores $(S_\omega:\omega\in \Omega_N)$ and capacities $(q^N_j: j\in \mathcal{J}_0)$.
\Ensure Generates a sequence $c^{(0)}, c^{(1)}, \ldots$ such that $\lim_{k \to \infty} c^{(k)}$ converges to the cutoff scores in a stable matching.\vspace{0.1em}\\

\State \textbf{Initialize:} 
\begin{itemize}
    \item Set initial cutoffs: $c^{(0)} \gets (1, 1, \ldots, 1)\in \real^{|\mathcal{J}_0|}$
    \item Set iteration counter: $k \gets 0$
\end{itemize}

\Repeat
\State For each student $\omega$ and each program $j \in \mathcal{J}_0$, define 
\begin{equation}\label{eq:DA_d}
d^{(k)}_{\omega j} = \indicator\left[\optchoice{\succ^Q_\omega}{F(S_\omega, c^{(k)})} = j\right].
\end{equation}
\ForAll{$j \in \mathcal{J}_0$}
    \If{ $\sum_\omega \indicator[S_{\omega j} \geq c^{(k)}_{j},\ d^{(k)}_{\omega j} = 1] + \sum_\omega \indicator[S_{\omega j} < c^{(k)}_{j}] < q^N_j$}
    \State Let $c^{(k+1)}_j \gets 0$.
    \Else
    \State Let $\tau = q^{N}_j - \sum_\omega \indicator(S_{\omega j} \geq c^{(k)}_{j},\ d^{(k)}_{\omega j} = 1)$.
	    \If {$\tau > 0$}
	    \State Let cutoff $c^{(k+1)}_j$ be the $\tau$-th highest value in $\{S_{\omega j}: \omega \in \Omega_N,\ S_{\omega j} < c_j^{(k)}\}$.
	    \Else
	    \State Let $c_j^{(k + 1)} \gets c_j^{(k)}$.
	    \EndIf
    \EndIf
\EndFor
\State Increment iteration: $k \gets k + 1$
\Until{The sequence $c^{(k)}$ converges.}
\end{algorithmic}
\end{algorithm}

\begin{enumerate}
	\item At each iteration $k$, every program $j$ maintains a cutoff score denoted by $c^{(k)}_j$. Initially, no student receives an offer, and the cutoffs are initialized at their maximum possible value, $c_j^{(0)} = 1$ for all programs. (Recall that priority scores are normalized to be distributed within the interval $[0,1]$.)

    \item Given the cutoff vector $c^{(k)}$, each student $\omega$ tentatively accepts the offer from her most preferred feasible program according to her preference $\succ^Q_\omega$. Here, $d^{(k)}_{\omega j}$ in Algorithm \ref{alg:DA} will equal $1$ if student $\omega$ accepts the offer tentatively. 

    \item Each program $j$ updates its cutoff to reflect the tentative acceptances. Specifically, 
	    \begin{itemize}
	    \item Compute the number of students who tentatively accepted an offer from program $j$: $\sum_\omega \indicator[S_{\omega j} \geq c^{(k)}_{j},\ d^{(k)}_{\omega j} = 1]$.
	    \item Determine the number of students to whom program $j$ has not yet made an offer: $\sum_\omega \indicator[S_{\omega j} < c^{(k)}_{j}]$.
	    \end{itemize}
	    If the sum of these two numbers is less than program $j$'s capacity $q_j^N$, we know that program $j$'s capacity constraint is not binding. In this case, set the cutoff $c^{(k+1)}_j = 0$.

	     Otherwise, program $j$ needs to make additional offers to fill its capacity. Specfically, the number of new offers is $\tau = q^{N}_j - \sum_\omega \indicator(S_{\omega j} \geq c^{(k)}_{j},\ d^{(k)}_{\omega j} = 1)$. Program $j$ then makes offers to the $\tau$ students with the highest priority scores among those not yet receiving an offer from program $j$, setting the updated cutoff $c^{(k+1)}_j$ to the priority score of the $\tau$-th highest-ranked student in this group.
\end{enumerate}

The following theorem summarizes two key properties of the cutoff sequence $(c^{(k)})_k$ generated by Algorithm~\ref{alg:DA}. Specifically, this sequence is nonincreasing and converges to an upper bound on equilibrium cutoffs for all stable matchings.\footnote{We have not found this exact theorem in the existing literature, but we doubt it is a novel result. Nonetheless, for completeness, we provide its proof in Appendix \ref{proof:DA_program_proposing}.}

\begin{theorem}\label{thm:DA_convergence_properties}
Given any preferences $(\succ^Q_\omega: \omega\in \Omega_N)$, any priority scores $(S_\omega: \omega\in \Omega_N)$ and any program capacities $(q^N_j: j\in \mathcal{J}_0)$, the sequence of cutoffs $(c^{(k)})_k$ produced by Algorithm~\ref{alg:DA} satisfies the following properties:
\begin{enumerate}
    \item For each program $j\in \mathcal{J}_0$, the sequence $(c^{(k)}_j)_k$ is nonincreasing in $k$ and converges to a limit cutoff $\overline{c}_j$ after a finite number of iterations.
    \item For any stable matching with equilibrium cutoff vector $c$, $\overline{c}_j \ge c_j$ for all $j\in \mathcal{J}_0$.
\end{enumerate}
\end{theorem}

Let $\overline{c} \coloneqq (\overline{c}_j : j\in \mathcal{J}_0)$ denote the upper bound on equilibrium cutoffs computed by Algorithm~\ref{alg:DA}, given the true preference $(\succ^Q_\omega: \omega\in \Omega_N)$, priority scores $(S_\omega: \omega\in \Omega_N)$ and program capacities $(q^N_j: j\in \mathcal{J}_0)$. Similarly, we can obtain the lower bound $\underline{c}$ from the DA algorithm with students proposing, given the same counterfactual priority scores and capacities.

Equipped with these cutoff bounds, we are now prepared to discuss how to effectively utilize them for dimension reduction. Consider an arbitrary stable matching $d$ and its corresponding cutoff vector $c$. By the cutoff characterization established in \cite{azevedo_supply_2016}, we have
\[
    d_{\omega j} = \indicator\left[\optchoice{\succ^Q_\omega}{F_\omega(S_\omega, c)} = j\right],
\]
indicating that each student $\omega$ is matched to the most preferred program within her feasible set. Since a program $j$ must belong to the student's feasible set to be matched, the matching variable $d_{\omega j}$ is weakly decreasing in the cutoff $c_j$. Moreover, as other programs $j' \neq j$ compete with program $j$ for the student's selection, $d_{\omega j}$ is weakly increasing in  cutoff $c_{j'}$ for all $j' \neq j$. 

Given that each program's cutoff satisfies $c_j \in [\underline{c}_j, \overline{c}_j]$, the above monotonicity relationships imply the following bounds:
\begin{equation}\label{eq:d_bound}
\underline{d}_{\omega j} \leq d_{\omega j} \leq \overline{d}_{\omega j},
\end{equation}
where the lower and upper bounds are defined as:
\begin{equation}\label{eq:d_bound_def}
\begin{aligned}
	\underline{d}_{\omega j} &= \indicator\left[ \optchoice{\succ^Q_\omega}{F(S_\omega, (\underline{c}_0,\dots,\underline{c}_{j-1}, \overline{c}_j, \underline{c}_{j+1}, \dots,\underline{c}_J))} = j\right], \\[5pt]
    \overline{d}_{\omega j} &= \indicator\left[ \optchoice{\succ^Q_\omega}{F(S_\omega, (\overline{c}_0,\dots,\overline{c}_{j-1}, \underline{c}_j, \overline{c}_{j+1}, \dots,\overline{c}_J))} = j\right].
\end{aligned}
\end{equation}
Suppose that for a given student-program pair $(\omega,j)$ we have $\underline{d}_{\omega j} = \overline{d}_{\omega j} = 0$. In this case, student $\omega$ and program $j$ will never be matched to each other in a stable matching. Conversely, if $\underline{d}_{\omega j} = \overline{d}_{\omega j} = 1$, then this student-program pair must always be matched together. Since the cutoff bounds $\underline{c}$ and $\overline{c}$ can be efficiently computed via the deferred acceptance algorithm, this approach provides a computationally feasible way to reduce the dimensionality of the optimization problem without explicitly solving the integer linear program in equation \eqref{eq:ILP}.

The same logic extends naturally to the scenario in which preferences are set-identified. However, the primary challenge is that the standard DA algorithm is not applicable when preferences are only partially known (set-identified). To address this issue, we generalize the standard DA algorithm to handle set-identified preferences, introducing what we term the \emph{generalized deferred acceptance} algorithm. The program-proposing version of this generalized algorithm, which generates the upper bounds on equilibrium cutoffs, is presented in Algorithm~\ref{alg:M-DA}. 

\begin{algorithm}[h!]
\caption{Generalized Deferred Acceptance (Upper bounds)}\label{alg:M-DA}
\begin{algorithmic}
\Require Set of Preferences $(\mathcal{C}_\omega: \omega\in \Omega_N)$, scores $(S_\omega:\omega\in \Omega_N)$ and capacities $(q^N_j: j\in \mathcal{J}_0)$.
\Ensure Generates a sequence $\overline{c}^{(0)}, \overline{c}^{(1)}, \ldots$ such that $\lim_{k \to \infty} \overline{c}^{(k)}$ converges to an upper bound.

\State \textbf{Initialize:} 
\begin{itemize}
    \item Set initial cutoffs: $\overline{c}^{(0)} \gets (1, 1, \ldots, 1) \in \mathbb{R}^J$
    \item Set iteration counter: $k \gets 0$
\end{itemize}

\Repeat
\State For each student $\omega$ and each program $j \in \mathcal{J}$, define:
\begin{equation}
	d^{(k)}_{\omega j} = \indicator\left[ \exists \succ^Q \in \mathcal{C}_\omega \text{ such that } \optchoice{\succ^Q}{\mathcal{F}_\omega(\overline{c}^{(k)})} = j \right].
\end{equation}

\ForAll{$j \in \mathcal{J}_0$}
\If{ $\sum_\omega \indicator[S_{\omega j} \geq \overline{c}^{(k)}_{j},\ d^{(k)}_{\omega j} = 1] + \sum_\omega \indicator[S_{\omega j} < \overline{c}^{(k)}_{j}] < q^N_j$}
\State Let $\overline{c}^{(k+1)}_j \gets 0$.
    \Else
    \State Let $\tau = q^{N}_j - \sum_\omega \indicator(S_{\omega j} \geq \overline{c}^{(k)}_{j},\ d^{(k)}_{\omega j} = 1)$.
	    \If {$\tau > 0$}
	    \State Let cutoff $\overline{c}^{(k+1)}_j$ be the $\tau$-th highest value in $\{S_{\omega j}: \omega \in \Omega_N,\ S_{\omega j} < \overline{c}_j^{(k)}\}$.
	    \Else
	    \State Let $\overline{c}_j^{(k + 1)} \gets \overline{c}_j^{(k)}$.
	    \EndIf
    \EndIf
\EndFor

\State Increment iteration: $k \gets k + 1$
\Until{The sequence $\overline{c}^{(k)}$ converges.}
\end{algorithmic}
\end{algorithm}

The primary difference between Algorithms~\ref{alg:DA} and \ref{alg:M-DA} lies in the definition of $d^{(k)}_{\omega j}$. Specifically, the definition of $d^{(k)}_{\omega j}$ in Algorithm~\ref{alg:DA} (given by equation \eqref{eq:DA_d}) utilizes the known preferences $\succ^Q_\omega$. In contrast, Algorithm~\ref{alg:M-DA} searches across all possible preferences within the preference bounds $\mathcal{C}_\omega$ to determine whether student $\omega$ might accept a proposed offer. Despite this difference, the next theorem demonstrates that Algorithm~\ref{alg:M-DA} retains the essential properties established in Theorem~\ref{thm:DA_convergence_properties} for Algorithm~\ref{alg:DA}.

\begin{theorem}\label{thm:upper_bound_simple}
Given arbitrary preference bounds $(\mathcal{C}_\omega: \omega\in \Omega_N)$, priority scores $(S_\omega: \omega\in \Omega_N)$, and program capacities $(q^N_j: j\in \mathcal{J}_0)$, the sequence of cutoffs $(\overline{c}^{(k)})_k$ produced by Algorithm~\ref{alg:M-DA} satisfies the following properties:
\begin{enumerate}
    \item For each program $j\in \mathcal{J}_0$, the sequence $(\overline{c}^{(k)}_j)_k$ is nonincreasing in $k$ and converges to a limit cutoff $\overline{c}_j$ in a finite number of iterations.
    \item For any matching $d\in \mathcal{D}(\mathcal{C}^N, q^N, S)$ with associated equilibrium cutoff vector $c$, we have $\overline{c}_j \ge c_j$ for all $j\in \mathcal{J}_0$.
\end{enumerate}
Moreover, if each set $\mathcal{C}_\omega$ is a singleton (i.e., preferences are exactly known for all $\omega$), Algorithm~\ref{alg:M-DA} reduces precisely to Algorithm~\ref{alg:DA}.
\end{theorem}

Using Theorem~\ref{thm:upper_bound_simple}, we obtain an upper bound $\overline{c}$ on the equilibrium cutoffs associated with all stable matchings in the counterfactual scenario. A corresponding lower bound $\underline{c}$ can similarly be derived by adapting the deferred acceptance algorithm with students proposing.  Equipped with these equilibrium cutoff bounds, we can compute $\underline{d}_{\omega j}$ and $\overline{d}_{\omega j}$ in equation \eqref{eq:d_bound_def} for set-identified preferences. Finally, we apply the inequality from equation \eqref{eq:d_bound} to bound each dimension of $d\in\mathcal{D}(\mathcal{C}^N, \tilde{q}^N, \tilde{S})$, achieving dimensionality reduction with set-identified preferences.


\subsection{Comparing to an existing approach}\label{subsec:para}
In this subsection, we outline the approach proposed by \cite{fack_beyond_2019}, which has subsequently been adopted in numerous papers, including \cite{ngo_preferences_2024} and \cite{barahona2021}, among others. We begin by restating two critical assumptions from \cite{fack_beyond_2019}.

\begin{assumption}[Parametric Utility]\label{Ass:Para}
Each student's true preference ordering $\succ^Q_\omega$ can be represented by a vector of real-valued utility functions $U_\omega \equiv (U_{\omega j}: j \in \mathcal{J})$, where
\[
U_{\omega j} = h(X_\omega; \gamma) + \epsilon_{\omega j},
\]
$h(\cdot)$ is a known function, and the conditional distribution of the preference shocks $(\epsilon_{\omega j}: j \in \mathcal{J})$ given $X_{\omega}$ is known up to a finite-dimensional parameter.
\end{assumption}

\begin{assumption}[Exogeneity of Priority Scores]\label{Ass:Exo}
Conditional on student characteristics $X_\omega$, the vector of priority scores $S_\omega$ is independent of the preference shocks. Formally, 
\[
S_\omega \independent \epsilon_{\omega j}\mid X_\omega.
\]
\end{assumption}

The primary role of Assumptions \ref{Ass:Para} and \ref{Ass:Exo} is to fully characterize the joint distribution of preferences and priority scores. Under the stability assumption (Assumption \ref{assu:stability}), this joint distribution provides the essential information required for point identification. Most empirical studies conducting counterfactual analyses in Deferred Acceptance (DA)-type models follow the procedure outlined below:

\begin{enumerate}
    \item[(Step 1)] {\bf Estimate Preferences ($\succ^Q_\omega$):}
     Under Assumptions \ref{assu:stability}, \ref{Ass:Para}, and \ref{Ass:Exo}, researchers estimate the parametric utilities $ \hat{U}_\omega $ and thus recover the joint distribution of priority scores and true preference orderings $ \succ^Q_\omega $ for each student. Detailed discussions of this step can be found in \cite{fack_beyond_2019} and \cite{ngo_preferences_2024}. 

    \item[(Step 2)] {\bf Simulate the DA Algorithm for the Counterfactual Economy:}
    Using the estimated preferences $ \hat{U}_\omega $, counterfactual priority scores $ \tilde{S}_\omega $, and capacities $ \tilde{q}^N $, researchers simulate the DA mechanism under the counterfactual setting and solve the stable matching and its associated cutoffs, recovering counterfactual outcomes of interest.
\end{enumerate}

In Step 2, researchers very often utilize the estimated preferences and typically assume that the matching outcomes in the counterfactual scenario are determined directly by students' true preferences rather than their submitted ROLs. See \cite{Artemov2023}, section V.B for a more detailed discussion. Our approach adopts this assumption as well.

The validity of the approach described above relies on consistently recovering students' true preference $\succ^Q_\omega$. However, as discussed and illustrated in Section \ref{sec:stable}, the stability assumption alone offers limited nonparametric identification power regarding students' true preferences. Indeed, the partial orders induced solely by stability are often insufficiently informative. This is also noted by \cite{kapor2024aftermarket} in practice. Consequently, the validity of the existing approach crucially depends on Assumptions \ref{Ass:Para} and \ref{Ass:Exo}.

However, these two critical assumptions could be restrictive and lack robustness in certain empirical contexts. It is well known that ad-hoc parametric approaches such as Assumption \ref{Ass:Para} can introduce substantial misspecification bias. Avoiding overly strong parameterizations enhances the credibility of the results.  
Assumption \ref{Ass:Exo}, raises more fundamental concerns. For example, in the Roy model—a commonly used framework in labor economics—skills and preferences are closely intertwined. Assumption \ref{Ass:Exo} presumes that cognitive and non-cognitive skills that are fundamental determinants of labour markets outcomes are fully captured by priority scores and covariates. 
For example, \cite{fack_beyond_2019} includes existing exam scores as covariates, which are intended to proxy students' underlying abilities.
However, previous research (e.g. \cite{heckmankautz2012}, \cite{borghans2008}, \cite{almlund2011}) demonstrates that test scores often do not account for soft skills and personality traits that determine the outcomes of the labor market and are then correlated with preferences. This underscores the presence of unobserved skills  that influence preferences, further challenging the validity of the approach.  

The key distinction between our proposed methodology and this existing approach is that we do not impose Assumptions~\ref{Ass:Para} and~\ref{Ass:Exo}, resulting in preferences that are set-identified rather than point-identified. To nevertheless obtain informative results without relying on Assumptions~\ref{Ass:Para} and~\ref{Ass:Exo}, we develop a flexible analytical framework capable of incorporating alternative assumptions, such as the undominated strategies (Assumption \ref{assu:undominance}) introduced in Section~\ref{sec:undominance} and its variants discussed later in Section~\ref{sec:more_revealed_preference}. These alternative assumptions are grounded in economic theory and enable us to leverage the rich information embedded in students' submitted ROLs. Although our relaxation of assumptions makes the counterfactual analysis more challenging, it remains computationally feasible by utilizing the ILP representation and dimension-reduction techniques developed earlier in this section.

\section{Additional methods for revealing preferences}\label{sec:more_revealed_preference}

In this section, we introduce two additional methods for inferring student preferences from observed data. The first method, which we call the \emph{robust undominated strategy}, addresses robustness concerns related to the standard undominated strategy assumption, as previously discussed in Remark~\ref{remark:undominance}. Compared to the standard undominated strategy assumption, this robust approach yields preference bounds that are less informative but more reliable. The second method aims to infer the selectivity ordering among programs and leverages this additional information to further sharpen our inference about student preferences. 

\subsection{Robust Undominated Strategy Assumptions}\label{sec:robust_undominated}
As discussed in Remark~\ref{remark:undominance}, a primary concern with the standard undominated strategy assumption is that it implicitly assumes student $\omega$ considers all possible feasible sets when deciding which ROLs to submit. To address this concern, we introduce the notion of a \emph{consideration set} of feasible sets. This set, denoted by $\mathcal{F}^*_\omega$, includes only those feasible sets that student $\omega$ believes could occur with positive probability. Given this concept, it is natural to assume students submit ROLs that are undominated with respect to their consideration sets $\mathcal{F}^*_\omega$. We refer to this as the \emph{robust undominated strategy}, formally defined as follows:

\begin{definition}[Undominated strategies given consideration set]\label{def:robust_undominance}
Given student $\omega$'s true preference $\succ^Q_\omega$ and her consideration set $\mathcal{F}^*_\omega$, a submitted ROL $\succ^R_\omega$ is said to be \emph{dominated given $\mathcal{F}^*_\omega$} by another ROL $\succ^{R'}_\omega$ if both of the following conditions hold:
\begin{enumerate}
    \item For every feasible set $F \in \mathcal{F}^*_\omega$, either 
    $ \optchoice{\succ^{R'}_\omega}{F} = \optchoice{\succ^{R}_\omega}{F}$ or $\optchoice{\succ^{R'}_\omega}{F} \succ^Q_\omega \optchoice{\succ^{R}_\omega}{F}$.
    \item There exists at least one feasible set $F \in \mathcal{F}^*_\omega$  for which 
    \[
        \optchoice{\succ^{R'}_\omega}{F} \succ^Q_\omega \optchoice{\succ^{R}_\omega}{F}.
    \]
\end{enumerate}
A submitted ROL $\succ^R_\omega$ is said to be \emph{undominated given $\mathcal{F}^*_\omega$} if there is no alternative ROL that dominates it given the consideration set $\mathcal{F}^*_\omega$.
\end{definition}

In general, undominated strategies given a consideration set $\mathcal{F}^*_\omega$, as defined in Definition~\ref{def:robust_undominance}, differ from the undominated strategies presented in Definition~\ref{def:undominance}. Specifically, Definition~\ref{def:robust_undominance} features a weaker Condition (i) and a stronger Condition (ii) compared to those in Definition~\ref{def:undominance}. Consequently, neither definition implies the other in general. However, the two definitions coincide in the special case where the consideration set $\mathcal{F}^*_\omega$ includes every subset of $\mathcal{J}_0$ that contains the outside option.  Thus, the undominated strategy introduced in Definition~\ref{def:undominance} provides the same empirical content as that the Definition~\ref{def:robust_undominance} whenever $\mathcal{F}^*_\omega$ includes every feasible subset in $\mathcal{J}_0$.

Another interesting extreme case to consider is when student $\omega$ has no uncertainty about the equilibrium cutoff $c$ and, hence, her feasible set $F_\omega(c)$. In this case, her consideration set is $\mathcal{F}^*_\omega = \{F_\omega(c)\}$. One can show that $\succ^R_\omega$ is undominated given this $\mathcal{F}^*_\omega$ if and only if this $\succ^R_\omega$ leads her to be matched with her most preferred feasible program, i.e., $\optchoice{\succ^Q_\omega}{F_\omega(c)} = \optchoice{\succ^R_\omega}{F_\omega(c)}$. This coincides with equation \eqref{eq:stability_immediate}, the implication of the stability assumption. In other words, the stability assumption provides the same empirical content on preference revealing as the undominance in Definition~\ref{def:robust_undominance} given $\mathcal{F}^*_\omega = \{F_\omega(c)\}$.

Because the consideration set $\mathcal{F}^*_\omega$ is unobservable, we do not attempt to infer student $\omega$'s true preferences directly under the assumption that $\succ^R_\omega$ is undominated given a known $\mathcal{F}^*_\omega$. Instead, we first construct a set $\mathcal{F}_\omega$ from observable data, assuming that $\mathcal{F}_\omega \subseteq \mathcal{F}^*_\omega$. We then assume that $\succ^R_\omega$ is undominated given $\mathcal{F}^*_\omega$. These assumptions are formalized in the following assumption.

\begin{assumption}[Robust undominated strategy]\label{assu:robust_undominance}
For each student $\omega$, $\mathcal{F}_\omega \subseteq \mathcal{F}^*_\omega$ and $\succ^R_\omega$ is an undominated strategy given $\mathcal{F}^*_\omega$.
\end{assumption}

Assumption \ref{assu:robust_undominance} is a weaker assumption than the undominated stratgy assumption (Assumption \ref{assu:undominance}), as it allows for all consideration set $\mathcal{F}^*_\omega$ that includes $\mathcal{F}_\omega$. In fact, Assumption \ref{assu:undominance} is can be viewed as a special case of Assumption \ref{assu:robust_undominance} when the consideration set includes all possible feasible sets, i.e., $\mathcal{F}_\omega = \{F \subseteq \mathcal{J}_0: 0 \in F\}$.

We defer the discussion of methods for constructing $\mathcal{F}_\omega$ based on historical cutoff data to the end of this subsection. Before that, we examine the implications of Assumption~\ref{assu:robust_undominance} on students' revealed preferences. As before, we construct partial orders compatible with students' true preferences. However, unlike previous instances, we need to define two partial orders. 

\begin{itemize}
\item Define the binary relation $\succ^{u1}_\omega$ where, for each $B \in \mathcal{F}_\omega$ and each $j \in B \setminus {\optchoice{\succ^R_\omega}{B}}$, let $\optchoice{\succ^R_\omega}{B} \succ^{u1}_\omega j$.

\item Define the binary relation $\succ^{u2}_\omega$ as follows:
\begin{enumerate}
	\item For any $j \in \text{domain}(\succ^R_\omega)\setminus \{0\}$, let $j \succ^{u2}_\omega 0$.

    \item For each $B \in \mathcal{F}_\omega$ and each $j \in B \cap \text{domain}(\succ^R_\omega) \setminus \{\optchoice{\succ^R_\omega}{B}\}$, let $\optchoice{\succ^R_\omega}{B} \succ^{u2}_\omega j$.
\end{enumerate}

\item Define $\succ^{P_{u1}}_\omega$ and $\succ^{P_{u2}}_\omega$ to be the transitive closure of $\succ^{u1}_\omega$ and $\succ^{u2}_\omega$ respectively.
\end{itemize}
Here, binary relations $\succ^{u1}_\omega$ and $\succ^{u2}_\omega$ are not partial orders because they are not necessarily transitive. Their transitive closures $\succ^{P_{u1}}_\omega$ and $\succ^{P_{u2}}_\omega$ are partial orders as shown in Lemma \ref{lem:partial_order_u1_u2} in Appendix \ref{sec:supp_robust_undominated}. The transitive closure is defined in the notation and preliminaries section. Let us illustrate these concepts using the running example.

\begin{example}[continues=running_example]
Recall that in the running example, we have $\mathcal{J}_0 = \{0, 1, 2, \dots, 6\}$ and $K = 5$. As before, we focus on student $\omega$, whose ROL is  $ 2 \succ^R_\omega 3 \succ^R_\omega 4 \succ^R_\omega 5 \succ^R_\omega 0$ and feasible set is $\{4, 5, 6, 0\}$.

Suppose we construct $\mathcal{F}_\omega$ as $\mathcal{F}_{\omega} = \{\{6, 0\}, \{5, 6, 0\}, \{4, 5, 6, 0\}, \{3, 4, 5, 6, 0\} \}$. To construct $\succ^{u1}_\omega$, we enumerate every $B$ in $\mathcal{F}_{\omega}$ and let $\max(\succ^R_\omega, B)$ dominate the other elements in $B$. For instance, from $B = \{6, 0\}\in \mathcal{F}_\omega$, we get $\optchoice{\succ^R_\omega}{B} = 0 \succ^{u1}_\omega 6$; from $B = \{5, 6, 0\}$, we get $5\succ^{u1}_\omega 0$ and $5\succ^{u1}_\omega 6$. The resulted $\succ^{P_{u1}}_\omega$  is $3 \succ^{P_{u1}}_\omega 4 \succ^{P_{u1}}_\omega 5 \succ^{P_{u1}}_\omega 0 \succ^{P_{u1}}_\omega 6$, as illustrated by Figure \ref{fig:robust_undominated_illustration_u1}. 

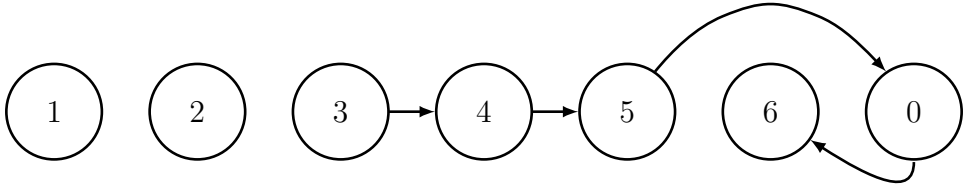
\begin{figure}[h]
    \centering
    \begin{tikzpicture}[>=latex,line join=bevel,]
  \pgfsetlinewidth{1bp}
\begin{scope}
  \pgfsetstrokecolor{black}
  \definecolor{strokecol}{rgb}{1.0,1.0,1.0};
  \pgfsetstrokecolor{strokecol}
  \definecolor{fillcol}{rgb}{1.0,1.0,1.0};
  \pgfsetfillcolor{fillcol}
  \filldraw (0.0bp,0.0bp) -- (0.0bp,67.53bp) -- (360.0bp,67.53bp) -- (360.0bp,0.0bp) -- cycle;
\end{scope}
  \pgfsetcolor{black}
  \draw [->] (144.14bp,26.936bp) .. controls (146.13bp,26.936bp) and (148.12bp,26.936bp)  .. (161.51bp,26.936bp);
  \draw [->] (198.14bp,26.936bp) .. controls (200.13bp,26.936bp) and (202.12bp,26.936bp)  .. (215.51bp,26.936bp);
  \draw [->] (244.46bp,42.061bp) .. controls (250.83bp,49.887bp) and (259.71bp,58.675bp)  .. (270.0bp,62.936bp) .. controls (284.78bp,69.059bp) and (291.22bp,69.059bp)  .. (306.0bp,62.936bp) .. controls (312.75bp,60.14bp) and (318.9bp,55.394bp)  .. (331.54bp,42.061bp);
  \draw [->] (342.0bp,7.9365bp) .. controls (342.0bp,-5.5649bp) and (326.57bp,1.0342bp)  .. (302.93bp,16.319bp);
\begin{scope}
  \definecolor{strokecol}{rgb}{0.0,0.0,0.0};
  \pgfsetstrokecolor{strokecol}
  \draw (18.0bp,26.94bp) ellipse (18.0bp and 18.0bp);
  \draw (18.0bp,26.936bp) node {$1$};
\end{scope}
\begin{scope}
  \definecolor{strokecol}{rgb}{0.0,0.0,0.0};
  \pgfsetstrokecolor{strokecol}
  \draw (72.0bp,26.94bp) ellipse (18.0bp and 18.0bp);
  \draw (72.0bp,26.936bp) node {$2$};
\end{scope}
\begin{scope}
  \definecolor{strokecol}{rgb}{0.0,0.0,0.0};
  \pgfsetstrokecolor{strokecol}
  \draw (126.0bp,26.94bp) ellipse (18.0bp and 18.0bp);
  \draw (126.0bp,26.936bp) node {$3$};
\end{scope}
\begin{scope}
  \definecolor{strokecol}{rgb}{0.0,0.0,0.0};
  \pgfsetstrokecolor{strokecol}
  \draw (180.0bp,26.94bp) ellipse (18.0bp and 18.0bp);
  \draw (180.0bp,26.936bp) node {$4$};
\end{scope}
\begin{scope}
  \definecolor{strokecol}{rgb}{0.0,0.0,0.0};
  \pgfsetstrokecolor{strokecol}
  \draw (234.0bp,26.94bp) ellipse (18.0bp and 18.0bp);
  \draw (234.0bp,26.936bp) node {$5$};
\end{scope}
\begin{scope}
  \definecolor{strokecol}{rgb}{0.0,0.0,0.0};
  \pgfsetstrokecolor{strokecol}
  \draw (288.0bp,26.94bp) ellipse (18.0bp and 18.0bp);
  \draw (288.0bp,26.936bp) node {$6$};
\end{scope}
\begin{scope}
  \definecolor{strokecol}{rgb}{0.0,0.0,0.0};
  \pgfsetstrokecolor{strokecol}
  \draw (342.0bp,26.94bp) ellipse (18.0bp and 18.0bp);
  \draw (342.0bp,26.936bp) node {$0$};
\end{scope}
\end{tikzpicture}

    \vspace{-3em}
    \caption{Illustration of $\succ^{P_{u1}}_\omega$}
    \label{fig:robust_undominated_illustration_u1}
\end{figure}

To construct $\succ^{u2}_\omega$, we first ensure that the programs listed in $\succ^R_\omega$ dominate the outside option,  so $2$, $3$, $4$, $5$ should $\succ^{u2}_\omega 0$. Next, for each $B$ in $\mathcal{F}_\omega$, we let $\optchoice{\succ^R_\omega}{B}$ dominate all the programs in both $B$ and the domain of $\succ^R_\omega$. For example, from $B = \{3, 4, 5, 6, 0\}\in \mathcal{F}_\omega$, we get $\optchoice{\succ^R_\omega}{B} = 3$ so that $3 \succ^{u2}_\omega 4$, $3 \succ^{u2}_\omega 5$ and $3 \succ^{u2}_\omega 0$. The resulted $\succ^{P_{u2}}_\omega$ is $2\succ^{P_{u2}}_\omega 0$ and $3 \succ^{P_{u2}}_\omega 4 \succ^{P_{u2}}_\omega 5 \succ^{P_{u2}}_\omega 0$, as illustrated in Figure \ref{fig:robust_undominated_illustration_u2}. 

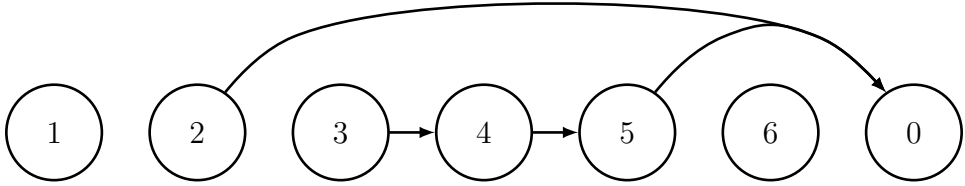
\begin{figure}[h]
    \centering
    \begin{tikzpicture}[>=latex,line join=bevel,]
  \pgfsetlinewidth{1bp}
\begin{scope}
  \pgfsetstrokecolor{black}
  \definecolor{strokecol}{rgb}{1.0,1.0,1.0};
  \pgfsetstrokecolor{strokecol}
  \definecolor{fillcol}{rgb}{1.0,1.0,1.0};
  \pgfsetfillcolor{fillcol}
  \filldraw (0.0bp,0.0bp) -- (0.0bp,66.63bp) -- (360.0bp,66.63bp) -- (360.0bp,0.0bp) -- cycle;
\end{scope}
  \pgfsetcolor{black}
  \draw [->] (82.462bp,33.125bp) .. controls (88.829bp,40.951bp) and (97.711bp,49.738bp)  .. (108.0bp,54.0bp) .. controls (148.65bp,70.838bp) and (265.35bp,70.838bp)  .. (306.0bp,54.0bp) .. controls (312.75bp,51.203bp) and (318.9bp,46.457bp)  .. (331.54bp,33.125bp);
  \draw [->] (144.14bp,18.0bp) .. controls (146.13bp,18.0bp) and (148.12bp,18.0bp)  .. (161.51bp,18.0bp);
  \draw [->] (198.14bp,18.0bp) .. controls (200.13bp,18.0bp) and (202.12bp,18.0bp)  .. (215.51bp,18.0bp);
  \draw [->] (244.46bp,33.125bp) .. controls (250.83bp,40.951bp) and (259.71bp,49.738bp)  .. (270.0bp,54.0bp) .. controls (284.78bp,60.123bp) and (291.22bp,60.123bp)  .. (306.0bp,54.0bp) .. controls (312.75bp,51.203bp) and (318.9bp,46.457bp)  .. (331.54bp,33.125bp);
\begin{scope}
  \definecolor{strokecol}{rgb}{0.0,0.0,0.0};
  \pgfsetstrokecolor{strokecol}
  \draw (18.0bp,18.0bp) ellipse (18.0bp and 18.0bp);
  \draw (18.0bp,18.0bp) node {$1$};
\end{scope}
\begin{scope}
  \definecolor{strokecol}{rgb}{0.0,0.0,0.0};
  \pgfsetstrokecolor{strokecol}
  \draw (72.0bp,18.0bp) ellipse (18.0bp and 18.0bp);
  \draw (72.0bp,18.0bp) node {$2$};
\end{scope}
\begin{scope}
  \definecolor{strokecol}{rgb}{0.0,0.0,0.0};
  \pgfsetstrokecolor{strokecol}
  \draw (126.0bp,18.0bp) ellipse (18.0bp and 18.0bp);
  \draw (126.0bp,18.0bp) node {$3$};
\end{scope}
\begin{scope}
  \definecolor{strokecol}{rgb}{0.0,0.0,0.0};
  \pgfsetstrokecolor{strokecol}
  \draw (342.0bp,18.0bp) ellipse (18.0bp and 18.0bp);
  \draw (342.0bp,18.0bp) node {$0$};
\end{scope}
\begin{scope}
  \definecolor{strokecol}{rgb}{0.0,0.0,0.0};
  \pgfsetstrokecolor{strokecol}
  \draw (180.0bp,18.0bp) ellipse (18.0bp and 18.0bp);
  \draw (180.0bp,18.0bp) node {$4$};
\end{scope}
\begin{scope}
  \definecolor{strokecol}{rgb}{0.0,0.0,0.0};
  \pgfsetstrokecolor{strokecol}
  \draw (234.0bp,18.0bp) ellipse (18.0bp and 18.0bp);
  \draw (234.0bp,18.0bp) node {$5$};
\end{scope}
\begin{scope}
  \definecolor{strokecol}{rgb}{0.0,0.0,0.0};
  \pgfsetstrokecolor{strokecol}
  \draw (288.0bp,18.0bp) ellipse (18.0bp and 18.0bp);
  \draw (288.0bp,18.0bp) node {$6$};
\end{scope}
\end{tikzpicture}

    \vspace{-2em}
    \caption{Illustration of $\succ^{P_{u2}}_\omega$}
    \label{fig:robust_undominated_illustration_u2}
\end{figure}
\end{example}

We are now able to characterize the implications of Assumption ~\ref{assu:robust_undominance} for student preferences using $\succ^{P_{u1}}_\omega$ and $\succ^{P_{u2}}_\omega$.

\begin{proposition}\label{prop:robust_undominance}
There exists $(\mathcal{F}^*_\omega: \omega \in \Omega)$ satisfying Assumption~\ref{assu:robust_undominance} if and only if, the following two conditions are satisfied:
\begin{enumerate}
\item for student $\omega$ with $|\succ^R_\omega| < K$, $\succ^Q_\omega$ is compatible with $\succ^{P_{u1}}_\omega$.
\item for student $\omega$ with $|\succ^R_\omega| = K$, $\succ^Q_\omega$ is compatible with at least one of $\succ^{P_{u1}}_\omega$ and $\succ^{P_{u2}}_\omega$. 
\end{enumerate}
\end{proposition}

The empirical content of the partial orders $\succ^{P_{u1}}_\omega$ and $\succ^{P_{u2}}_\omega$ derived from Assumption~\ref{assu:robust_undominance}  depends on the richness of the consideration set $\mathcal{F}_{\omega}$. If we assume students are highly uncertain about their feasible sets and set $\mathcal{F}_{\omega} = \{B\subseteq \mathcal{J}_0: 0\in B\}$, then the partial orders $\succ^{P_{u1}}_\omega$ and $\succ^{P_{u2}}_\omega$ become as informative as $\succ^{P_u}_\omega$, and Proposition~\ref{prop:robust_undominance} simplifies to Proposition~\ref{Prop:undominated_strategy}. Specifically, under this maximal uncertainty scenario, $\succ^{P_{u1}}_\omega$ coincides with $\succ^{P_u}_\omega$ when $|\succ^R_\omega| < K$, and $\succ^{P_{u2}}_\omega$ coincides with $\succ^{P_u}_\omega$ when $|\succ^R_\omega| = K$. Furthermore, in this case, $\succ^{P_{u2}}_\omega$ is always weaker than $\succ^{P_{u1}}_\omega$.

On the other hand, if students are very confident about their feasible sets, we may choose the minimal consideration set containing only the realized feasible set, $\mathcal{F}_{\omega} = \{F_{\omega}\}$. In this scenario, $\succ^{P_{u1}}_\omega$ reduces to the partial order $\succ^{P_s}_\omega$ constructed under the stability assumption in Proposition~\ref{Prop:PS}. Consequently, in this special case, Assumption~\ref{assu:robust_undominance} provides weakly less empirical content than Assumption~\ref{assu:stability}.

This brings us to the practical question of how to construct the consideration set $\mathcal{F}_{\omega}$. Suppose that we have observed the cutoff scores for all programs over multiple periods. One natural approach to constructing $\mathcal{F}_{\omega}$ is to assume that students believe that the historical cutoffs might recur in the current period. Formally, let $c_{j,t}$ denote the cutoff for program $j$ in period $t$, and let $S_{j,t}(\omega)$ represent student $\omega$'s priority score at program $j$ in period $t$. Then, the set $F_{t,t', \omega} \coloneqq \{j \in \mathcal{J}: S_{j,t}(\omega) \ge c_{j,t'}\}$ represents student $\omega$'s feasible set in period $t$ when using the cutoff from period $t'$ as a reference. By defining the consideration set for student $\omega$ in period $t$ as $\mathcal{F}_{\omega} = \{F_{t,t',\omega} : t' \le t\}$, we explicitly assume that students account for historical cutoffs when they decide their reported ROL $\succ^R_\omega$.

Alternatively, we can construct a richer $\mathcal{F}_{\omega}$ by allowing perturbations around historical cutoffs. Given a slack parameter $\delta\ge 0$, define
\[
\underline{F}_{\omega, t,t'}=\{\,j\in\mathcal{J}_0:\; S_{j,t}(\omega)\ge c_{j,t'}+\delta\,\},\quad
\overline{F}_{\omega,t,t'}=\{\,j\in\mathcal{J}_0:\; S_{j,t}(\omega)\ge c_{j,t'}-\delta\,\}.
\]
For fixed period $t$, let $\underline{F}_{\omega,t}\coloneqq \bigcap_{t'\le t}\underline{F}_{\omega,t,t'}$ be the programs that are always feasible to $\omega$ even under the least favorable perturbation, and let $\overline{F}_{\omega,t}\coloneqq \bigcup_{t'\le t}\overline{F}_{\omega, t,t'}$ be the programs that may be feasible in some years under favorable perturbations. A sensible construction is then
\begin{equation}\label{eq:perturbation_construct_F}
\mathcal{F}_{\omega}=\{\,F\in\mathscr{F}_0:\; \underline{F}_{\omega,t}\subseteq F \subseteq \overline{F}_{\omega,t}\,\},\text{ where }\mathscr{F}_0=\{F\subseteq \mathcal{J}_0: 0\in F\}.
\end{equation}
This construction applies to both multi-period and single-period data. 




\subsection{Selectivity ranking}\label{sec:selective}
In practice, certain programs are often perceived as more prestigious or selective than others. In this section, we formalize this notion of selectivity ranking and demonstrate how incorporating it can yield additional information about students' true preferences from their submitted ROLs. Specifically, we explore how integrating selectivity information can enrich the information obtained from the robust undominated strategy outlined in Assumption \ref{assu:robust_undominance}. 

Following \cite{bertanha_causal_2024}, we model the selectivity ranking as a binary relation derived from student's consideration sets.

\begin{definition}[Selectivity ranking]\label{def:selectiveness}
We say program $j$ is \emph{more selective} than program $j'$, denoted as $j \gg^* j'$, if, for any student $\omega$ and any feasible set $B$ within her consideration set $\mathcal{F}^*_\omega$, the inclusion of $j$ implies the inclusion of $j'$, i.e., $j\in B$ always implies $j'\in B$.
\end{definition}

According to Definition~\ref{def:selectiveness}, if program $j$ is more selective than program $j'$, any student qualified for program $j$ is necessarily qualified for program $j'$. Equivalently, it implies that admission into program $j$ is always weakly more difficult than admission into program $j'$. By construction, the selectivity ranking $\gg^*$ is reflexive and transitive. In fact, this $\gg^*$ is a nonstrict partial order.

Because students' consideration sets are not observable, the selectivity ranking $\gg^*$ is also not observable. Similar as the approach adopted in the previous subsection, instead of working directly on $\gg^*$, we would first construct a binary relation $\gg$ from the data and assume $\gg$ is consistent with the true selectivity ranking induced from the true consideration sets.

\begin{assumption}[Selectiveness]\label{assu:selectiveness}
	Assume $\gg$ is included within the $\gg^*$ induced from $(\mathcal{F}^*_\omega:\omega\in \Omega)$. That is, for any $j,j'\in \mathcal{J}_0$,  $j \gg j'$ implies that for any $\omega\in \Omega$, there does not exist some $B\in \mathcal{F}^*_\omega$ such that $j\in B$ and $j'\notin B$.
\end{assumption}

Assumption~\ref{assu:selectiveness} imposes additional structure on students' consideration sets, enabling further empirical implications to be derived from submitted ROLs beyond those implied by Assumption~\ref{assu:robust_undominance}. As before, we defer the discussion on how to construct the selectivity ranking $\gg$ to the end of this section. For now, we only require the constructed $\gg$ and $(\mathcal{F}_\omega: \omega\in \Omega)$ are compatible in the sense that for any $j,j'$ with $j \gg j'$, there cannot exist some $\omega$ and $B\in \mathcal{F}_\omega$ such that $j\in B$ but $j'\notin B$. Otherwise, there cannot exist $(\mathcal{F}^*_\omega:\omega\in \Omega)$ that satisfies both Assumption \ref{assu:robust_undominance} and \ref{assu:selectiveness}. Moreover, because the outside option is always within each feasible set, we know $j\gg 0$ for each $j\in \mathcal{J}$.

To discuss the joint implication of Assumptions \ref{assu:robust_undominance} and \ref{assu:selectiveness} for preference revealing given a constructed selectivity ranking $\gg$, we need to define $\succ^{P_{sel}}_\omega$ as follows:

\begin{itemize}
\item Define a binary relation $\succ^{sel}_\omega$ based on student $\omega$'s ROL $\succ^R_\omega$ and $\gg$:
\begin{enumerate}
	\item[(i)] For any $j \in \text{domain}(\succ^R_\omega)\setminus\{0\}$and any $j' \in \mathcal{J}_0\setminus\{j\}$ satisfying $j \gg j'$ and no $j'\succ^R_\omega j$,\footnote{Note that, because $\succ^R_\omega$ is a partial order, no $j'\succ^R_\omega j$ does not mean $j \succ^R_\omega j'$. For example, if $\mathcal{J}_0 = \{1, 2, 3, 0\}$ and $\succ^R_\omega$ is defined as $1 \succ^R_\omega 2 \succ^R_\omega 0$, then there is no $1 \succ^R_\omega 3$ and no $3 \succ^R_\omega 1$. Of course, there is also no $2 \succ^R_\omega 1$ in this example.} let $j \succ^{sel}_\omega j'$.
\item[(ii)] For each feasible set $B \in \mathcal{F}_\omega$, and for each program $j \in B \cap \text{domain}(\succ^{R}_\omega)$ that is different from $\optchoice{\succ^{R}_\omega}{B} $, let $\optchoice{\succ^{R}_\omega}{B} \succ^{sel}_\omega j$.
\end{enumerate}
\item Define $\succ^{P_{sel}}_\omega$ as the transitive closure of the binary relation $\succ^{sel}_\omega$ constructed above.
\end{itemize}
Comparing the definition of $\succ^{sel}_\omega$ and $\succ^{u2}_\omega$, we can see that $\succ^{sel}_\omega$ has a stronger Condition (i) and the same Condition (ii) in the definition. Lemma \ref{lem:select} in Appendix \ref{sec:lemma_sel} shows that $\succ^{P_{sel}}_\omega$ is a partial order as long as $\succ^R_\omega$ and $\gg$ are compatible with each other. Let us illustrate $\succ^{sel}_\omega$ and $\succ^{P_{sel}}_\omega$ using the running example.

\begin{example}[continues=running_example]
Recall the example used to illustrate $\succ^{P_{u1}}_\omega $ and $\succ^{P_{u2}}_\omega $. Suppose the $\gg$ is constructed as $1 \gg 2 \gg 3 \gg 0$. 

To work out the $\succ^{sel}_\omega$, note that Condition (i) in its definition implies $2 \succ^{sel}_\omega 3 \succ^{sel}_\omega 0$. Condition (ii) in its definition is the same as the Condition (ii) in the definition of $\succ^{u2}_\omega$, which implies $3 \succ^{sel}_\omega 4 \succ^{sel}_\omega 5 \succ^{sel}_\omega 0$. Then, its transitive closure $\succ^{P_{sel}}_\omega$ is $2\succ^{P_{sel}}_\omega 3 \succ^{P_{sel}}_\omega 4 \succ^{P_{sel}}_\omega 5 \succ^{P_{sel}}_\omega 0$, as illustrated in Figure \ref{fig:selectiveness_illustration_u3}. Compared to $\succ^{P_{u2}}_\omega$ in Figure \ref{fig:robust_undominated_illustration_u2}, $\succ^{P_{sel}}_\omega$ includes extra relations between $2$ and $\{3, 4, 5, 6\}$. Note that $\gg$ and $\succ^{P_{sel}}_\omega$ are generally different from each other. For example, there is no $1 \succ^{P_{sel}}_\omega 2$ even though there is $1 \gg 2$ in this case.
\end{example}

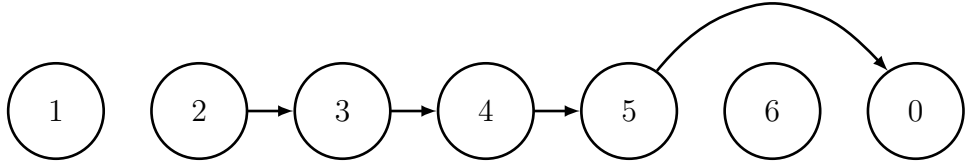
\begin{figure}[h]
    \centering
    \begin{tikzpicture}[>=latex,line join=bevel,]
  \pgfsetlinewidth{1bp}
\begin{scope}
  \pgfsetstrokecolor{black}
  \definecolor{strokecol}{rgb}{1.0,1.0,1.0};
  \pgfsetstrokecolor{strokecol}
  \definecolor{fillcol}{rgb}{1.0,1.0,1.0};
  \pgfsetfillcolor{fillcol}
  \filldraw (0.0bp,0.0bp) -- (0.0bp,58.59bp) -- (360.0bp,58.59bp) -- (360.0bp,0.0bp) -- cycle;
\end{scope}
  \pgfsetcolor{black}
  \draw [->] (90.141bp,18.0bp) .. controls (92.131bp,18.0bp) and (94.121bp,18.0bp)  .. (107.51bp,18.0bp);
  \draw [->] (144.14bp,18.0bp) .. controls (146.13bp,18.0bp) and (148.12bp,18.0bp)  .. (161.51bp,18.0bp);
  \draw [->] (198.14bp,18.0bp) .. controls (200.13bp,18.0bp) and (202.12bp,18.0bp)  .. (215.51bp,18.0bp);
  \draw [->] (244.46bp,33.125bp) .. controls (250.83bp,40.951bp) and (259.71bp,49.738bp)  .. (270.0bp,54.0bp) .. controls (284.78bp,60.123bp) and (291.22bp,60.123bp)  .. (306.0bp,54.0bp) .. controls (312.75bp,51.203bp) and (318.9bp,46.457bp)  .. (331.54bp,33.125bp);
\begin{scope}
  \definecolor{strokecol}{rgb}{0.0,0.0,0.0};
  \pgfsetstrokecolor{strokecol}
  \draw (18.0bp,18.0bp) ellipse (18.0bp and 18.0bp);
  \draw (18.0bp,18.0bp) node {$1$};
\end{scope}
\begin{scope}
  \definecolor{strokecol}{rgb}{0.0,0.0,0.0};
  \pgfsetstrokecolor{strokecol}
  \draw (72.0bp,18.0bp) ellipse (18.0bp and 18.0bp);
  \draw (72.0bp,18.0bp) node {$2$};
\end{scope}
\begin{scope}
  \definecolor{strokecol}{rgb}{0.0,0.0,0.0};
  \pgfsetstrokecolor{strokecol}
  \draw (126.0bp,18.0bp) ellipse (18.0bp and 18.0bp);
  \draw (126.0bp,18.0bp) node {$3$};
\end{scope}
\begin{scope}
  \definecolor{strokecol}{rgb}{0.0,0.0,0.0};
  \pgfsetstrokecolor{strokecol}
  \draw (180.0bp,18.0bp) ellipse (18.0bp and 18.0bp);
  \draw (180.0bp,18.0bp) node {$4$};
\end{scope}
\begin{scope}
  \definecolor{strokecol}{rgb}{0.0,0.0,0.0};
  \pgfsetstrokecolor{strokecol}
  \draw (234.0bp,18.0bp) ellipse (18.0bp and 18.0bp);
  \draw (234.0bp,18.0bp) node {$5$};
\end{scope}
\begin{scope}
  \definecolor{strokecol}{rgb}{0.0,0.0,0.0};
  \pgfsetstrokecolor{strokecol}
  \draw (288.0bp,18.0bp) ellipse (18.0bp and 18.0bp);
  \draw (288.0bp,18.0bp) node {$6$};
\end{scope}
\begin{scope}
  \definecolor{strokecol}{rgb}{0.0,0.0,0.0};
  \pgfsetstrokecolor{strokecol}
  \draw (342.0bp,18.0bp) ellipse (18.0bp and 18.0bp);
  \draw (342.0bp,18.0bp) node {$0$};
\end{scope}
\end{tikzpicture}
    \caption{Illustration of $\succ^{P_{sel}}_{\omega}$} \label{fig:selectiveness_illustration_u3}
\end{figure}

We are now ready to state the formal result. 

\begin{proposition}\label{prop:selectiveness}
	Suppose $\gg$ and $(\mathcal{F}_\omega:\omega\in \Omega)$ are compatible with each other. There exists $(\mathcal{F}^*_\omega: \omega\in \Omega)$ satisfying Assumptions \ref{assu:robust_undominance} and \ref{assu:selectiveness} if and only if, the following conditions are satisfied:
\begin{enumerate}
\item for student $\omega$ with $|\succ^R_\omega| < K$ or $\succ^R_\omega$ not compatile with $\gg$, $\succ^Q_\omega$ is compatible with $\succ^{P_{u1}}_\omega$
\item for student $\omega$ with $|\succ^R_\omega| = K$ and $\succ^R_\omega$ compatile with $\gg$, $\succ^Q_\omega$ is compatible with at least one of  $\succ^{P_{u1}}_\omega$ and $\succ^{P_{sel}}_\omega$.
\end{enumerate}
\end{proposition}

We now describe how to construct $\gg$. Suppose we observe priority scores and cutoffs over multiple periods. Let $c_{j,t}$ denote the cutoff for program $j$ in period $t$, and let $S_{j,t}(\omega)$ be student $\omega$'s priority score at program $j$ in period $t$. One way to construct $\gg$ for period $t$ is to define it as follows:
\begin{equation}\label{eq:selectiveness_history}
    j \gg j' \quad \text{iff} \quad \forall\, t' \le t:\ 
\{\omega : S_{j,t'}(\omega) \le c_{j,t'}\} \subseteq \{\omega : S_{j',t'}(\omega) \le c_{j',t'}\}.
\end{equation}
That is, $j \gg j'$ if in every observed period up to $t$, the set of students with scores below the cutoff for school $j$ is a subset of those with scores below the cutoff for school $j'$.

Note that this construction does not require programs $j$ and $j'$ to use the same priority score (i.e., it does not require $S_{j,t}(\omega)=S_{j',t}(\omega)$). Under this construction, Assumption~\ref{assu:selectiveness} admits a behavioral interpretation: students regard $j$ as more selective than $j'$ when $j$ has consistently been harder to enter than $j'$ across all observed periods.


\section{Empirical Analysis}\label{Sec:App}

The choice of college major has long been recognized as a key determinant of labor market outcomes. Since \cite{james1989college} highlighted that major choice is more consequential than college choice, a growing literature has examined both the determinants of major choice and its implications for labor market success.  Within this literature, STEM fields have received particular attention. Numerous studies document that STEM programs yield high labor-market returns, foster innovation and technological progress, and are associated with a smaller gender wage gap (e.g., \cite{black2008college,blau2017gender,dahl2023gender,beede2011women}). However, women remain persistently underrepresented in STEM programs and careers, a pattern thought to contribute meaningfully to the overall gender wage gap (see \cite{daymont1984}; \cite{zafar2013}). Reducing barriers to women's participation in STEM has therefore become a policy priority for governments and international organizations, as emphasized by UNESCO (2017). 

The literature points to two broad explanations for this persistent underrepresentation. First, gendered preferences and expectations influence students' program choices (see \cite{kahn_women_2017}). Second, institutional features of admissions systems, such as how priority scores are constructed and how assignment mechanisms are designed, may contribute to gender imbalances in access to STEM fields (see \cite{MontolioTaberner2021}).

This section examines how admission structures shape women's enrollment in STEM. We focus on Chile, where university admissions rely on a composite priority score that combines standardized high-stakes exam scores and high-school GPA. The exams include subject-specific assessments in Mathematics, Science, Language, and History and are typically administered on a single day, whereas GPA aggregates performance across multiple years of secondary education. Using the tools developed earlier, we evaluate policies that modify how the priority score is constructed, either by reweighting the GPA and exam components or by rescaling raw exam scores. We then quantify the resulting changes in the gender gap in STEM enrollment.

We document two facts in Subsection \ref{sec:descriptive_stat}. First, although each academic program sets its own weights, nearly all programs assign most of the weight to standardized exams, typically around 60-80\% (Figure \ref{fig:GPA_weight}). Second, as shown in Table \ref{tab:summary_gender} and Figure \ref{fig:scores_distribution}, male applicants tend to score higher on high-stakes exams, whereas female applicants tend to have higher GPAs. Similar patterns appear in other competitive educational contexts (see \cite{ArenasCalsamiglia2025} and references therein). 

These facts raise a concern: heavier weighting of exams may systematically favor male applicants and reinforce imbalances in selective, male-dominated fields. Because emphasizing high-stakes exams over GPA can have meaningful distributional consequences, it is empirically important to quantify how these weights shape gender gaps in STEM.\footnote{A related question concerns allocative efficiency: do these weights improve match quality and subsequent educational outcomes? The relative predictive power of GPA versus high-stakes exams for such outcomes remains poorly understood in the literature. Table \ref{tab:glm_results} provides suggestive evidence that GPA appears to be a stronger predictor of successful graduation than high-stakes exams.}

To shed light on this issue, we consider a series of counterfactual scenarios that progressively increase the weights of the GPA. We quantify how much these policies reduce the gender gap in STEM enrollment and identify the minimum weight of GPA required to produce a substantively meaningful reduction in current disparities. The results of this counterfactual analysis are presented in Subsection \ref{sec:more_weights_GPA}.

Alternatively, we consider a counterfactual policy that standardizes high-stakes exam scores within gender. Under this policy, each student's exam performance is evaluated relative to same-gender peers rather than all students who take the exam, removing, by construction, any between-gender differences in the distribution of exam scores. This design is motivated by evidence of gender differences in high-stakes exam performance across countries \citep{JurajdaMunich2011,Saygin2019,MontolioTaberner2021}. In particular, \citet{IriberriReyBiel2019} document heightened anxiety and cortisol responses among women under time pressure and in competitive settings, suggesting that high-stakes exams may impose disproportionate physiological and psychological costs on women and thereby compromise the fairness of admissions decisions, especially in competitive STEM programs. This gender-based standardization aims to neutralize differential stress effects on measured performance and, in turn, reduce gender gaps in access to selective programs, particularly in STEM. In Subsection \ref{sec:gender_standardization}, we assess whether such a policy could reduce observed gender disparities in STEM admissions, and by how much.

\subsection{Data}
We use publicly available data on Chile's centralized college application and assignment system from year 2005 to year 2010. The institutional setting has been described in detail by \citet{hastingsneilsonzimmerman2013} and \citet{larroucaurios2020,larroucaurios2021}, among others.  

College choice in Chile is organized as a semi-centralized system: a subset of universities participate in a centralized market in which a clearinghouse collects applicants' rank-order lists and determines assignments using a variant of the DA algorithm. Applicants may submit rank-order lists of up to eight major–university pairs (``programs'') out of more than 1,000.\footnote{The cap was increased to ten in 2011.} Program-specific priorities are determined by a weighted average of scores on the national standardized test (the PSU, \textit{prueba de selecci\'on universitaria}) and high-school GPA.\footnote{In 2013, students' relative rank within their high school was added as one of the “primary'' scores used to construct priorities.} 

We analyze gender gaps in both STEM and medical programs for two reasons. First, medical programs attract many of the highest-scoring applicants. Second, while women's underrepresentation in STEM receives the most attention, there is a growing call in Latin America to address gender gaps in medicine as well \cite{Negrotto2025GenderGap}. Thus, we study STEM and medical programs jointly in the empirical analysis.

We classify a program as ``STEM \& Med'' if it belongs to Engineering and Manufacturing, Information and Communication Technologies, Natural Sciences and Mathematics, or Medical Sciences. Operationally, we map all Chilean programs to 2020 CIP codes.\footnote{CIP codes are six-digit identifiers in the Classification of Instructional Programs, a national standard maintained by the U.S. Department of Education's National Center for Education Statistics (NCES) for classifying fields of study and reporting program completions.} We label a program as STEM if its CIP code appears on the \emph{U.S. Department of Homeland Security STEM-Designated Degree Program List} (as of July 12, 2023),\footnote{This list enumerates CIP codes that qualify for the STEM OPT extension in the United States.} and as Med if its CIP code begins with 51 (“Health Professions and Related Programs”), 60 (“Health Professions Residency”), or 61 (“Medical Residency”). We refer to the union of these categories as STEM \& Med in what follows.

\subsection{Descriptive statistics}\label{sec:descriptive_stat}
\begin{figure}
    \centering
    \includegraphics[width=0.9\linewidth]{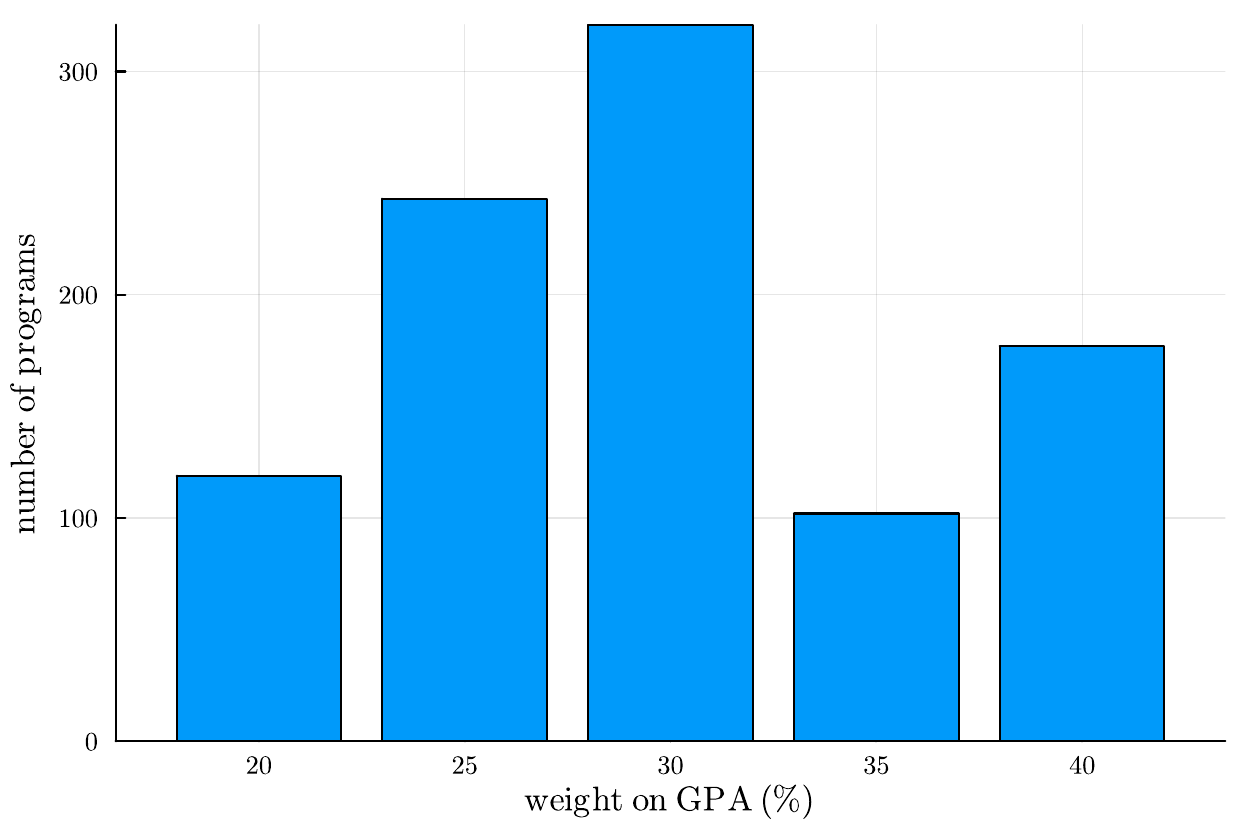}
    \caption{Weight distribution of High School GPA in the calculation of admission scores, 2010.}
    \label{fig:GPA_weight}
\end{figure}

We first document the weights that programs assign to GPA and high-stakes exams in their admissions priority scores. Figure \ref{fig:GPA_weight} plots the distribution of program-level GPA weights in year 2010. All programs allocate $20$--$40\%$ of the composite priority score to GPA, implying $60$--$80\%$ to high-stakes exams. Of all the $962$ programs, $683$ ($\approx 71\%$) programs assign at most $30\%$ to GPA. Overall, programs in Chile place the majority of the composite weight on high-stakes exams rather than GPA, and similar patterns were observed in other years.

\begin{table}[htbp]\centering
\caption{Academic, Choice, and Demographic Summary Statistics by Gender}
\label{tab:summary_gender}
\begin{tabular}{lccc}
\hline\hline
 & Male (1) & Female (2) & Difference (3) \\
\hline\hline

\\[-0.9em]
\multicolumn{4}{l}{\textbf{Academic achievement}} \\ \\

Language (High-Stakes Exam) & 578.394 & 576.114 & 2.281 \\
 & (0.389) & (0.381) & (0.545) \\

Math (High-Stakes Exam) & 596.972 & 562.302 & 34.671 \\
 & (0.415) & (0.384) & (0.566) \\

Science (High-Stakes Exam) & 571.148 & 549.714 & 21.434 \\
 & (0.495) & (0.496) & (0.701) \\

History (High-Stakes Exam) & 590.067 & 560.823 & 29.244 \\
 & (0.550) & (0.512) & (0.752) \\

High School GPA & 573.993 & 602.181 & -28.188 \\
 & (0.455) & (0.429) & (0.625) \\

Admission Score (avg. weights) & 562.936 & 555.957 & 6.979 \\
 & (0.348) & (0.332) & (0.480) \\

Admission Score (STEM \& Med, among science takers) & 592.471 & 584.338 & 8.134 \\
 & (0.408) & (0.402) & (0.573) \\

Admission Score (non-STEM \& Med, among history takers) & 542.224 & 533.033 & 9.191 \\
 & (0.497) & (0.451) & (0.671) \\

\hline \\[-0.9em]
\multicolumn{4}{l}{\textbf{School choices}} \\ \\

Total number of choices listed & 4.700 & 4.663 & 0.037 \\
 & (0.010) & (0.010) & (0.014) \\

ROL strictly shorter than permitted & 0.817 & 0.827 & -0.010 \\
 & (0.002) & (0.002) & (0.002) \\

Average weights of high-stakes exams in the ROL & 71.351 & 71.370 & -0.019 \\
 & (0.021) & (0.021) & (0.030) \\

Percent of choices that are STEM \& Med & 0.621 & 0.487 & 0.135 \\
 & (0.002) & (0.002) & (0.003) \\

\hline \\[-0.9em]
\multicolumn{4}{l}{\textbf{Program type conditional on assignment}} \\ \\

Rank of assigned program (conditional) & 1.331 & 1.280 & 0.052 \\
 & (0.008) & (0.008) & (0.011) \\

STEM \& Med assigned program & 0.359 & 0.228 & 0.131 \\
 & (0.002) & (0.002) & (0.003) \\

Number of observations & 47,644 & 49,494 & -1,850 \\

\hline\hline
\end{tabular}
\vspace{0.3em}
\begin{minipage}{0.9\textwidth}\footnotesize
\textit{Notes}: Calculations are for all students in the analysis sample for year 2010. 
Standard errors are in parentheses. 
\end{minipage}
\end{table}

\begin{figure}[htbp]
\centering 
\captionsetup[subfigure]{justification=centering}

\begin{subfigure}[t]{0.45\textwidth}
    \centering
    \includegraphics[width=\linewidth]{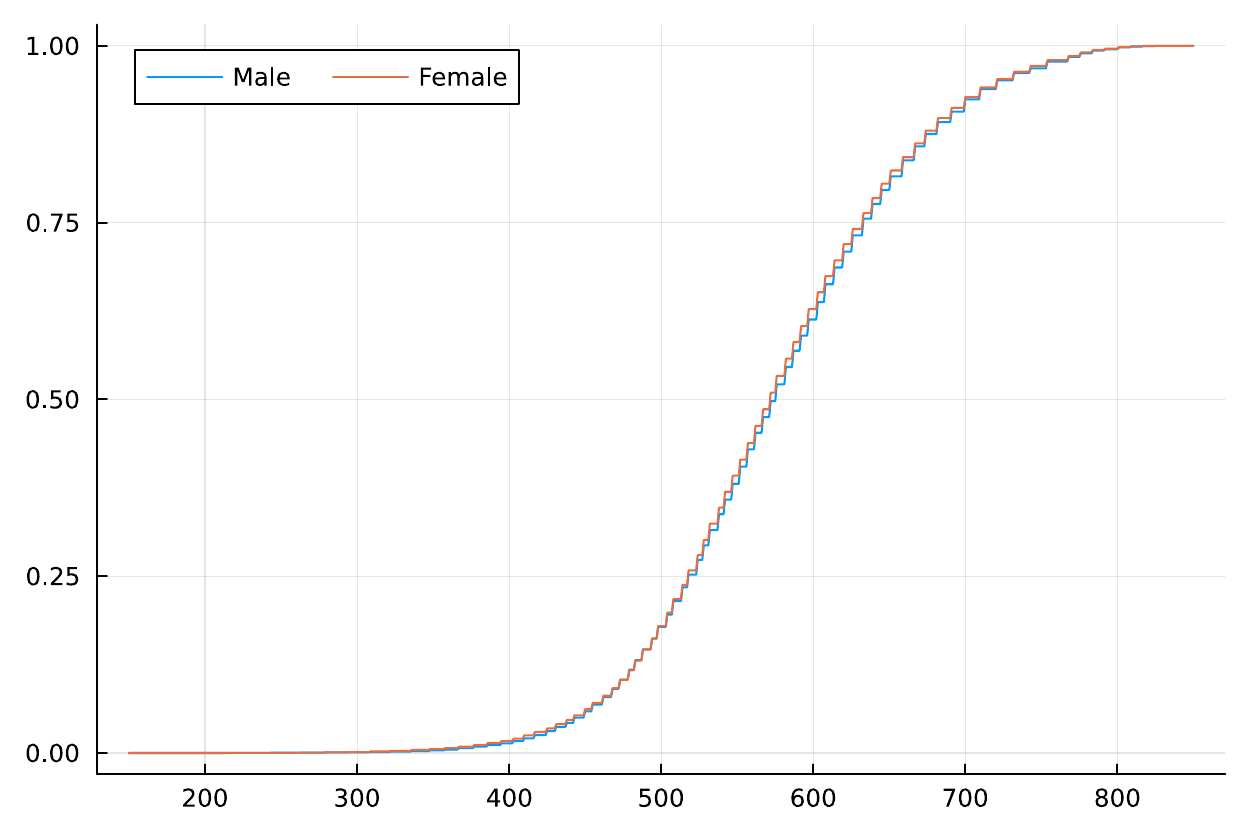}
    \caption{Language}
    \label{fig:sub1}
\end{subfigure}
\hfill
\begin{subfigure}[t]{0.45\textwidth}
    \centering
    \includegraphics[width=\linewidth]{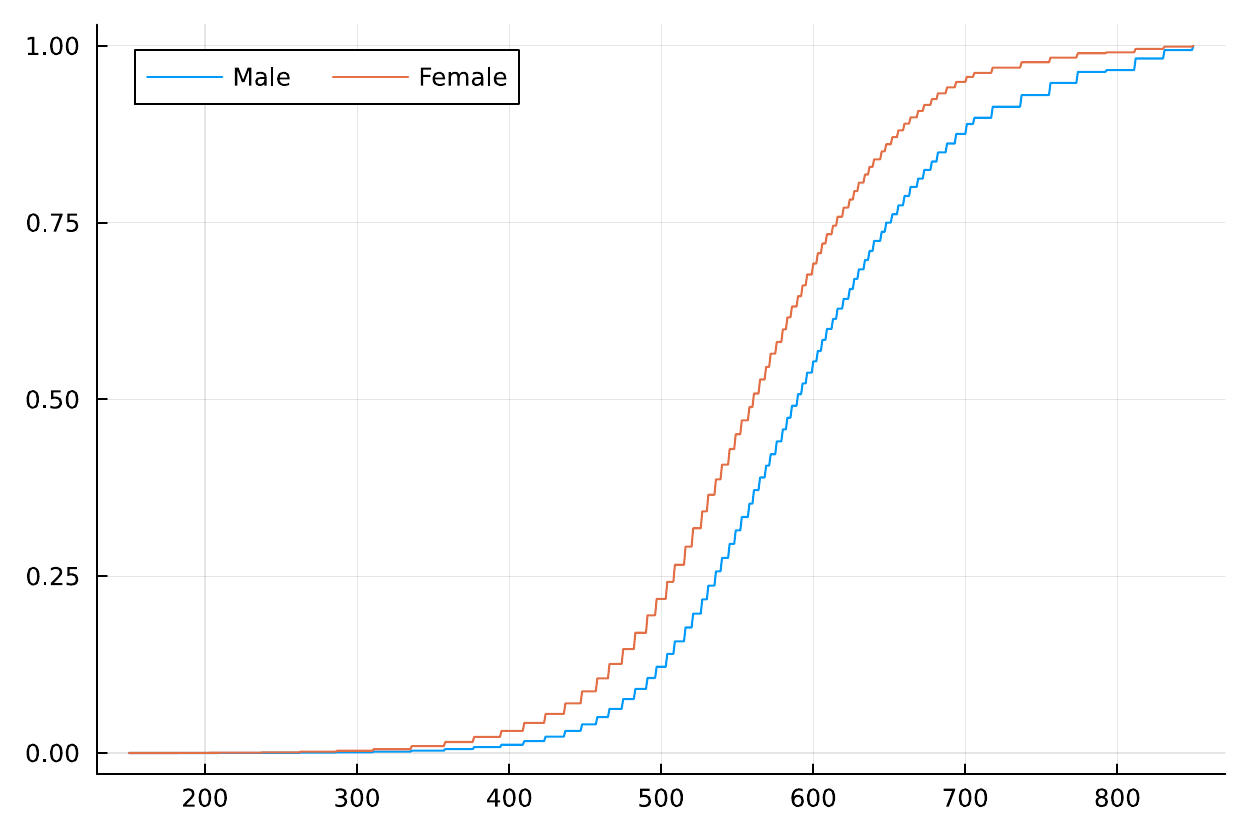}
    \caption{Math}
    \label{fig:sub2}
\end{subfigure}

\vspace{1em} 

\begin{subfigure}[t]{0.45\textwidth}
    \centering
    \includegraphics[width=\linewidth]{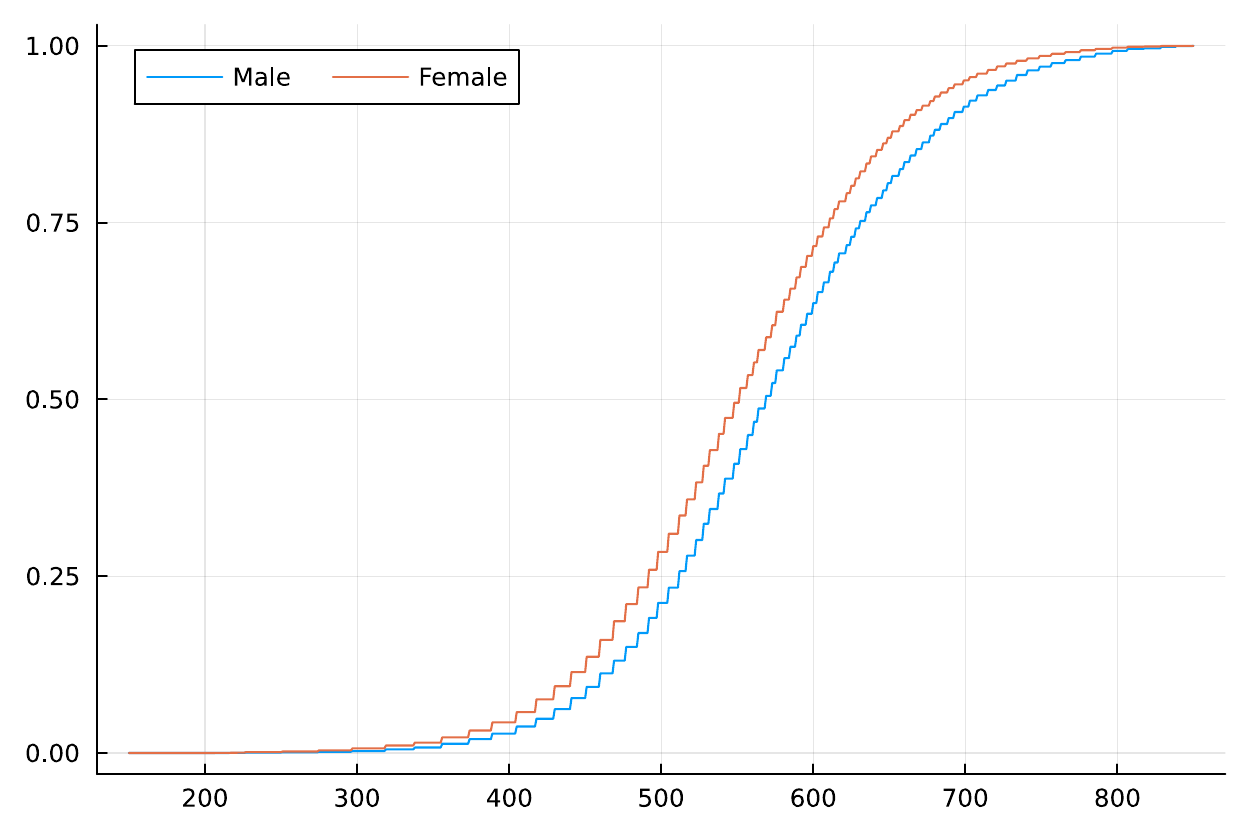}
    \caption{Science}
    \label{fig:sub3}
\end{subfigure}
\hfill
\begin{subfigure}[t]{0.45\textwidth}
    \centering
    \includegraphics[width=\linewidth]{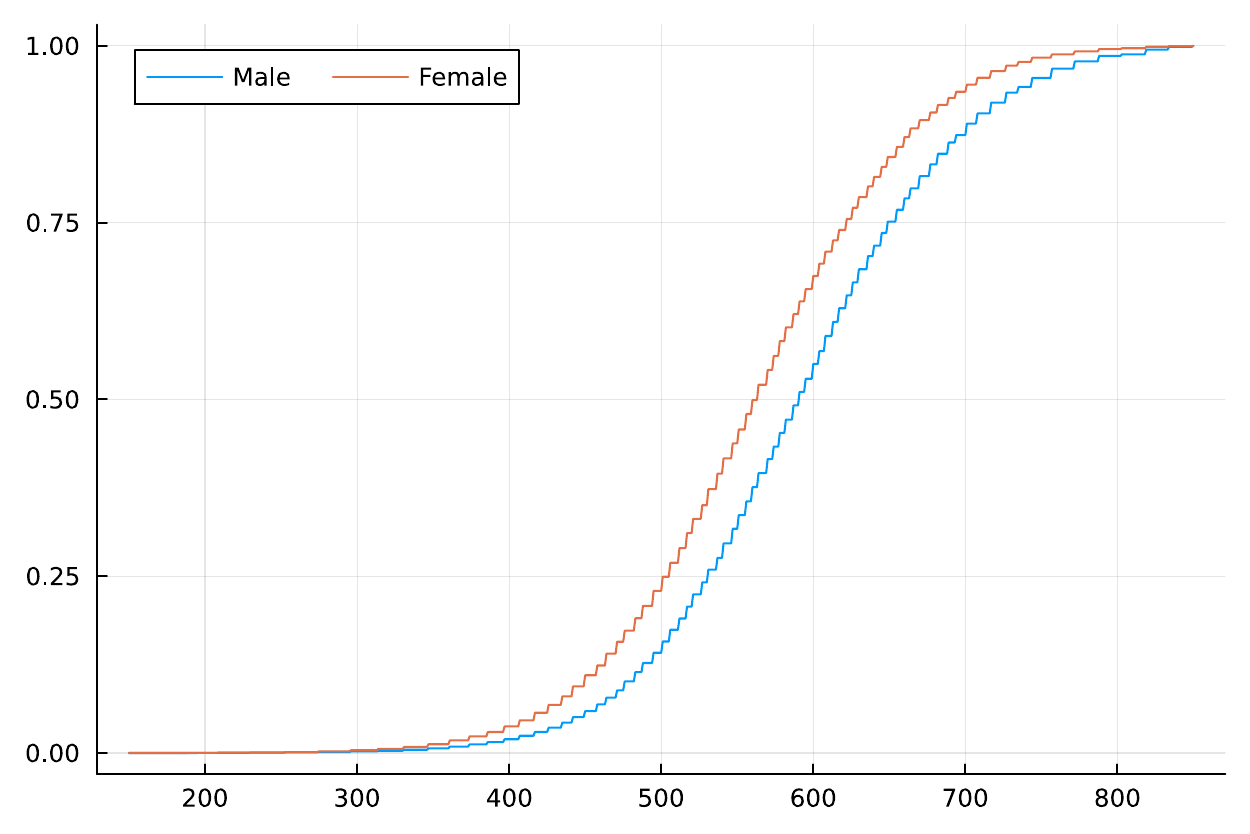}
    \caption{History}
    \label{fig:sub4}
\end{subfigure}

\vspace{1em} 


\begin{subfigure}[t]{0.45\textwidth}
    \centering
    \includegraphics[width=\linewidth]{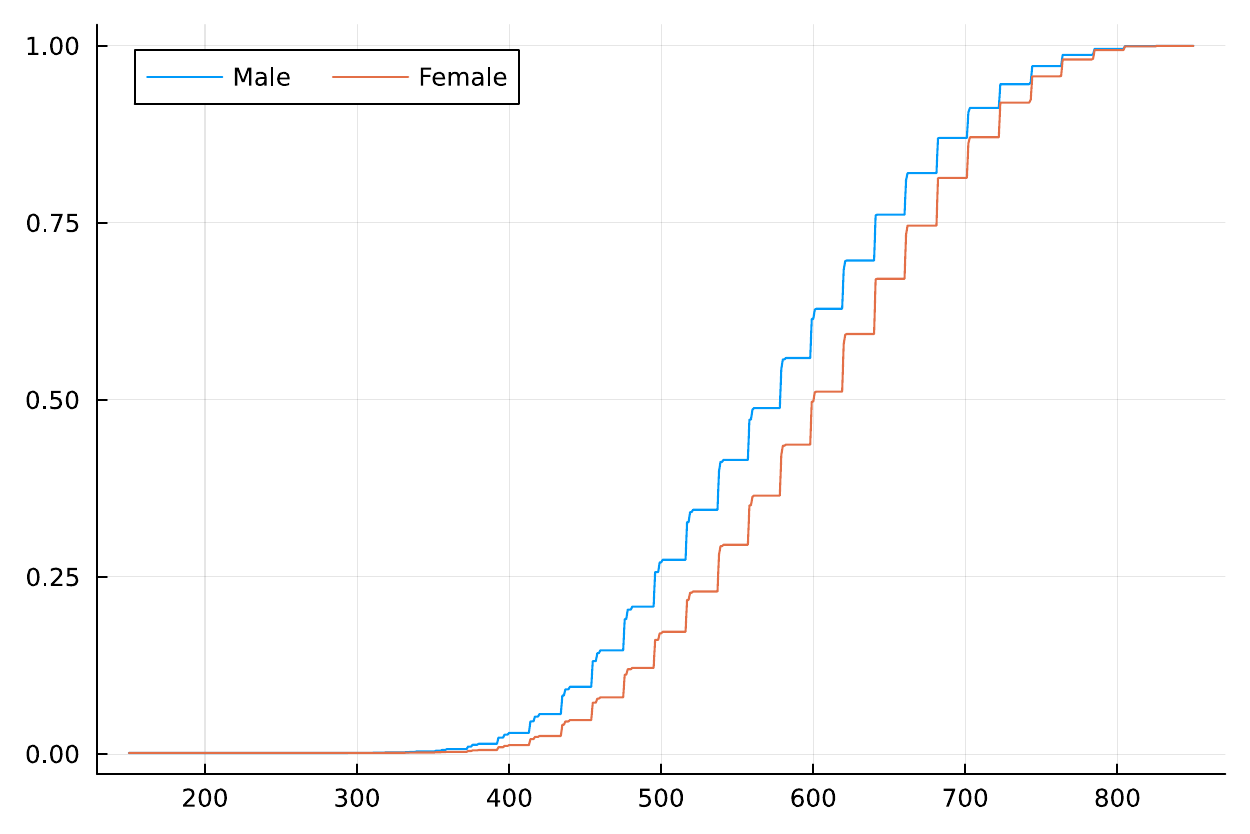}
    \caption{GPA}
    \label{fig:sub5}
\end{subfigure}
\hfill
\begin{subfigure}[t]{0.45\textwidth}
    \centering
    \includegraphics[width=\linewidth]{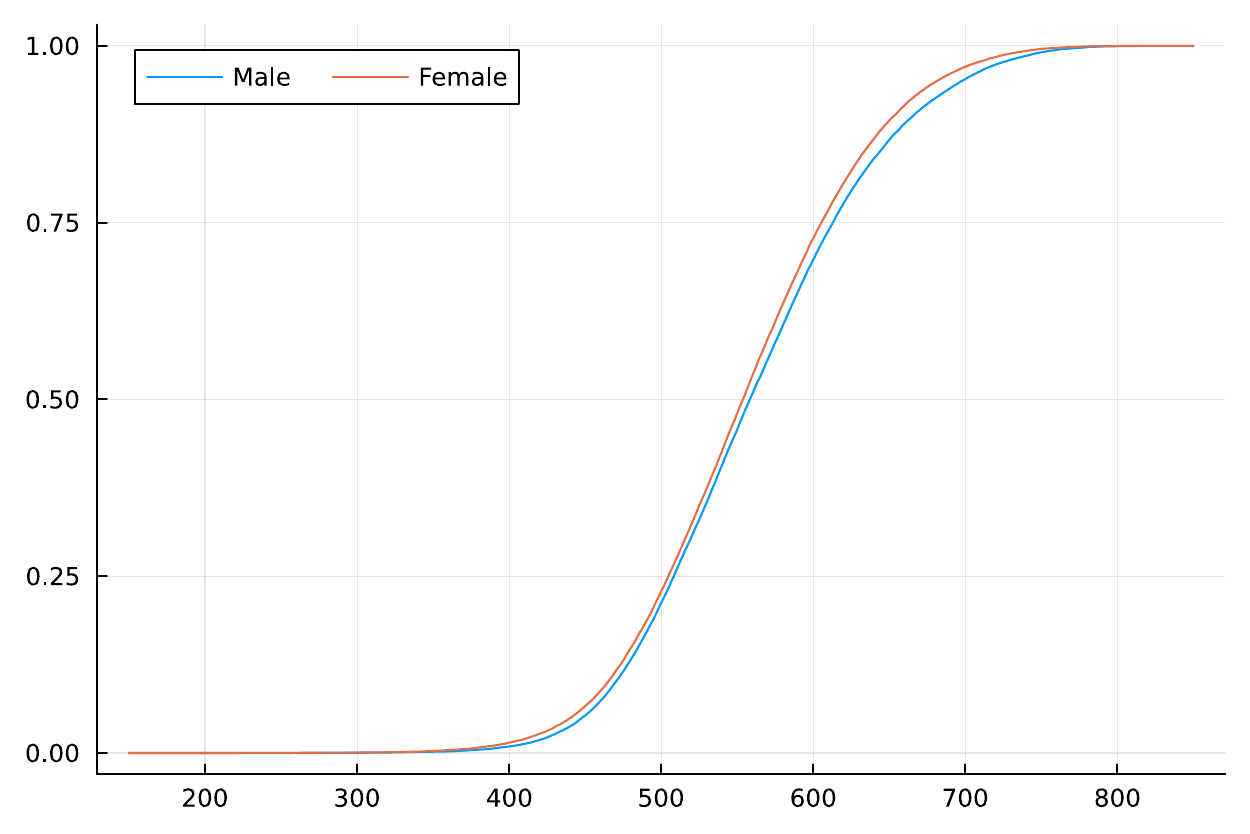}
    \caption{Average Admission Score}
    \label{fig:sub6}
\end{subfigure}


\caption{Cumulative distribution functions by subject area. Panels (a)–(d) display distributions for Language, Math, Science, and History, respectively. Panel (e) presents the GPA distribution, and Panel (f) presents the admission score using average program weights.}
\label{fig:scores_distribution}
\end{figure}

We next document gender differences in exam scores, GPA, ROLs, and match outcomes. Table \ref{tab:summary_gender} reports summary statistics for the 2010 cohort. Across subjects, male and female applicants display distinct patterns: in History, Mathematics, and Science, males score on average 29.24, 34.67, and 21.43 points higher than females, respectively, whereas Language scores differ only slightly but are still statistically significant. By contrast, high-school GPA reverses the pattern: females average a 28.19-point advantage over males. Given that exam components typically account for $60$--$80\%$ of the composite priority score, these performance gaps imply that men will, on average, have higher composite priority scores than women. Indeed, as shown in the sixth row of Table \ref{tab:summary_gender}, applying the average program weights yields priority scores that are 6.97 points higher for men than for women on average. A similar male advantage persists when we use weights only from STEM \& Med programs and restrict the sample to students who took the Science exam, and when we use weights only from non-STEM \& Med programs and restrict the sample to ones who took the History exam.

To see these differences in more detail, Figure~\ref{fig:scores_distribution} compares the cumulative distribution functions (CDFs) of subject-specific scores by gender and of the composite priority score (constructed using average program weights). Panels~(b)--(d) show that, in Mathematics, Science, and History, the distribution of male scores first-order stochastically dominates that of female scores. Panel~(a) shows a similar pattern for Language, though the gap is observationally small. In contrast, the GPA distribution in Panel (e) exhibits the opposite pattern: female scores first-order stochastically dominate male scores. Panel~(f) shows that the composite priority score likewise favors men, with the male CDF lying below the female CDF throughout. Although gaps are present across the distribution, they are largest in the upper tail: at the 90th percentile, males score about 12.46 points more than females.



Turning to program choices (middle panel of Table \ref{tab:summary_gender}), male applicants list slightly more programs on their rank-order lists (mean length: $4.7$ vs.~$4.6$). They also exhibit a stronger preference for STEM \& Med programs, which comprise $62.1\%$ of their ranked choices versus $48.7\%$ for female applicants. The average exam weight of the programs they rank is nearly identical across genders ($71.35$ vs.~$71.37$), suggesting that the weighting scheme does not differentially induce strategic program selection by gender. 

\begin{figure}[h!]
    \centering
    \includegraphics[width=0.9\linewidth]{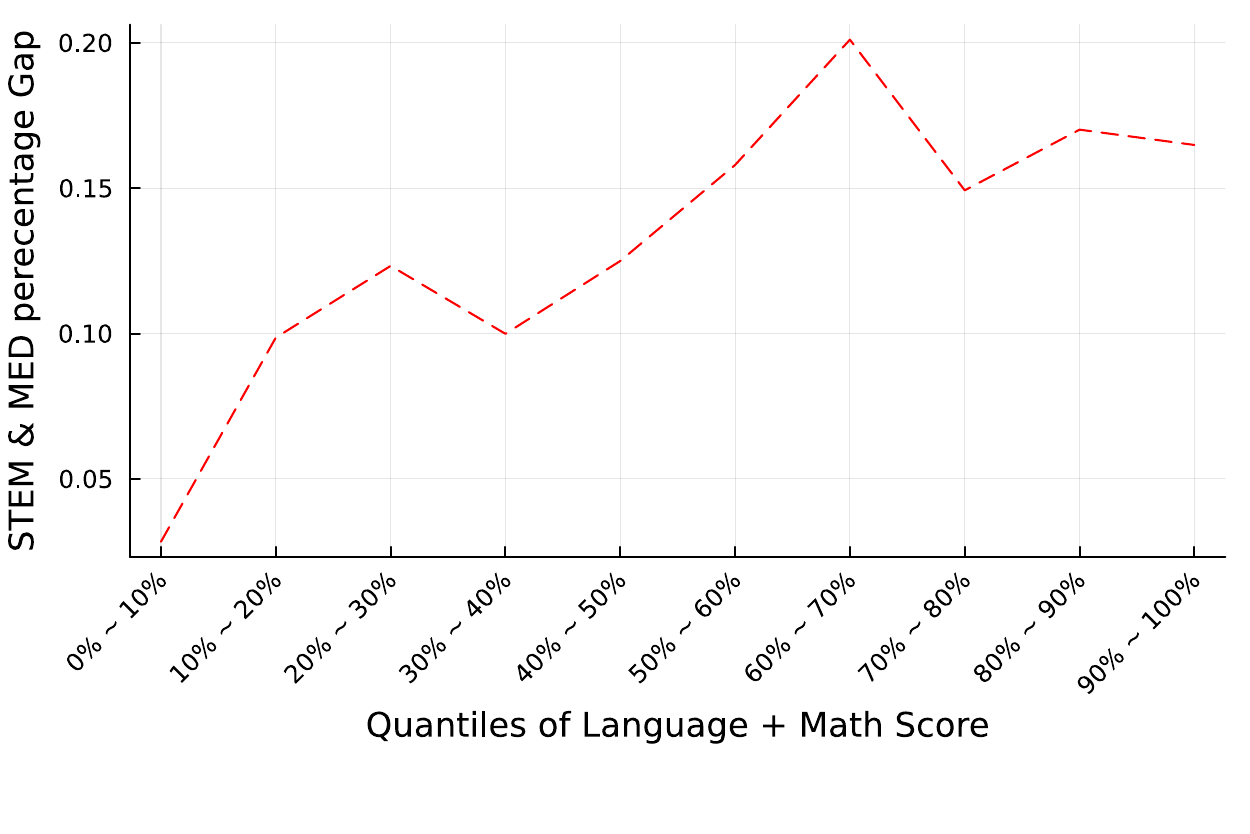}
    \caption{Percentage Gender gap in STEM \& Med enrollment across percentiles of Math and Language high-stakes exam scores, 2010.}
    \label{fig:quant_STEM}
\end{figure}

These differences in ROL composition, together with the score gaps documented above, translate into substantial disparities in final assignments, as shown in the last panel of Table \ref{tab:summary_gender}: $39.5\%$ of male applicants are assigned to a STEM \& Med program, compared with $22.8\%$ of female applicants, roughly half as many.
Figure~\ref{fig:quant_STEM} provides a more detailed view of the gender gap in STEM and Medicine enrollment (in percentage points) across percentiles of the Mathematics and Language exam score distributions. 
The gap is largest among high-performing students—peaking at about 20 percentage points between the 60th and 70th percentiles—and narrows considerably among lower-performing ones. 
This pattern suggests that gender disparities in STEM enrollment are particularly pronounced among students with strong academic preparation.

Finally, we estimate a logistic regression of student's graduation outcome on high-stakes exam scores (weighting subject-specific high-stakes exam scores with the admitting program's weights) and GPA scores. The estimates in Table \ref{tab:glm_results} indicate that both coefficients are statistically significant. Moreover, the GPA coefficient is roughly twice the magnitude of the exam coefficient, indicating a stronger association with graduation. Thus, while both components matter, sustained academic performance---as captured by GPA---appears to be the more powerful predictor. These findings raise questions about placing greater weights on high-stakes exams in admissions, particularly if weighting choices also shape gender gaps in STEM outcomes. We examine this in the next section.

\begin{table}[htbp]
\centering
\caption{Logit Model of Graduation on High-stakes exam and GPA Scores}
\label{tab:glm_results}
\begin{tabular}{lcccccc}
\toprule
 & \textbf{Coef.} & \textbf{Std. Error} & \textbf{z} & \textbf{Pr(>|z|)} & \textbf{Lower 95\%} & \textbf{Upper 95\%} \\
\midrule
Intercept        & -1.56836  & 0.09117  & -17.20 & $<$1e--65 & -1.74705 & -1.38967 \\
High-stakes exam &  0.00271  & 0.00013  &  21.30 & $<$1e--99 &  0.00246 &  0.00296 \\
GPA score        &  0.00563  & 0.00022  &  25.58 & $<$1e--99 &  0.00520 &  0.00606 \\
\bottomrule
\end{tabular}
\end{table}

\subsection{Counterfactuals}

We consider two counterfactual policies that could reduce the STEM \& Med gender gap: (i) increasing the weight on GPA relative to high-stakes exams in the priority score, and (ii) standardizing exam scores within gender before constructing priority scores. 

Our main goal is to evaluate the extent to which these policies may affect the STEM \& Med gender gap. Formally, our parameter of interest is defined as
\[
\theta \coloneqq \sum_{\omega\in \Omega_N} \sum_{j\in \mathcal{J}_0} \gamma_{\omega j} d_{\omega j},
\]
across all stable matchings $d$ in the counterfactual scenario, where
\[
\gamma_{\omega j} \;=\; \frac{1}{N_{\text{men}}}\,\indicator\!\big(j\in \text{STEM \& Med},\, \omega \in \text{Men}\big)
  \;-\; \frac{1}{N_{\text{women}}}\,\indicator\!\big(j\in \text{STEM \& Med},\, \omega\in \text{Women}\big).
\]
Here, $N_{\text{men}}$ ($N_{\text{women}}$) denotes the number of male (female) students, $\omega\in\text{Men}$ ($\omega\in\text{Women}$) indicates that $\omega$ is male (female). And $j\in\text{STEM \& Med}$ indicates that program $j$ is a STEM \& Med program.


\subsubsection{Revealed preferences under various assumptions}
As discussed earlier, our framework allows us to impose different combinations of assumptions. In what follows, we focus on Assumption~\ref{ass:stable}, Assumption~\ref{assu:undominance}, Assumption~\ref{assu:robust_undominance}, and Assumption~\ref{assu:selectiveness}. For the 2010 cohort, we construct $\mathcal{F}_\omega$ for Assumption~\ref{assu:undominance} using \eqref{eq:perturbation_construct_F} with $\delta=20$ and data from 2008--2010. We construct $\gg$ for Assumption~\ref{assu:selectiveness} using \eqref{eq:selectiveness_history} with the same data. \footnote{When imposing Assumptions~\ref{assu:robust_undominance} and \ref{assu:selectiveness} jointly, we replace $\mathscr{F}_0$ in \eqref{eq:perturbation_construct_F} with $\mathscr{F}_0^* \coloneqq \{F\in \mathscr{F}_0 : \nexists (j,j') \text{ with } j\gg j',\, j\in F,\, j'\notin F\}$ to ensure compatibility between $\gg$ and $\mathcal{F}_\omega$.}

To illustrate the identification power of each assumption set, Table \ref{tab:Iden-power} reports the number of pairwise orderings (binary relations) implied by each set, by gender and at selected percentiles of the Mathematics and Language high-stakes exam score distributions. When an assumption set implies a single partial order, we count the number of binary relations within that partial order. When it implies two partial orders, we report the number of binary relations that exist in both partial orders. More binary relations indicate greater identification power.

\begin{table}[htbp]
\centering
\caption{Identification power of different sets of assumptions by Gender and Percentile Group}
\label{tab:Iden-power}
\resizebox{\textwidth}{!}{
\begin{tabular}{lccccccccccc}
\toprule
 & \multicolumn{2}{c}{\textbf{S + US + Sel}} & \multicolumn{2}{c}{\textbf{S + US}} & 
   \multicolumn{2}{c}{\textbf{S + RUS + Sel}} & \multicolumn{2}{c}{\textbf{S + RUS}} & 
   \multicolumn{2}{c}{\textbf{S}} \\
\cmidrule(lr){2-3} \cmidrule(lr){4-5} \cmidrule(lr){6-7} \cmidrule(lr){8-9} \cmidrule(lr){10-11}
\textbf{Percentile Group} & Male & Female & Male & Female & Male & Female & Male & Female & Male & Female \\
\midrule
\textbf{Average (all)} & 5334.85 & 5284.51 & 4185.54 & 4231.49 & 1323.13 & 1225.91 & 1169.23 & 1103.48 & 395.71 & 374.35 \\
\midrule
0--10\%   & 4819.74 & 4763.96 & 3577.00 & 3496.61 & 108.91 & 31.33 & 94.22 & 27.88 & 25.66 & 8.84 \\
10--20\%  & 5241.00 & 5105.69 & 3829.80 & 3891.73 & 392.31 & 345.31 & 327.87 & 303.33 & 86.96 & 87.91 \\
20--30\%  & 5412.40 & 5374.88 & 4004.55 & 4071.24 & 666.29 & 563.64 & 567.57 & 480.44 & 146.81 & 138.14 \\
30--40\%  & 5456.75 & 5473.44 & 4127.79 & 4113.42 & 975.16 & 817.56 & 831.95 & 700.55 & 224.13 & 205.87 \\
40--50\%  & 5552.34 & 5529.71 & 4266.91 & 4287.26 & 1252.59 & 1095.43 & 1076.33 & 949.06 & 315.68 & 292.16 \\
50--60\%  & 5566.97 & 5523.79 & 4325.59 & 4389.32 & 1521.88 & 1344.73 & 1306.57 & 1185.85 & 411.39 & 387.41 \\
60--70\%  & 5572.55 & 5519.89 & 4466.32 & 4560.55 & 1792.98 & 1621.71 & 1574.15 & 1457.23 & 517.47 & 487.86 \\
70--80\%  & 5595.28 & 5450.95 & 4546.15 & 4598.19 & 2127.22 & 1965.68 & 1862.78 & 1762.35 & 623.74 & 590.65 \\
80--90\%  & 5345.86 & 5261.76 & 4474.70 & 4528.59 & 2298.47 & 2277.07 & 2067.49 & 2090.19 & 737.40 & 704.89 \\
90--100\% & 4786.05 & 4834.58 & 4230.29 & 4366.98 & 2076.16 & 2168.80 & 1965.06 & 2052.36 & 859.64 & 831.00 \\
\bottomrule
\end{tabular}
}
\vspace{0.3em}
\begin{minipage}{0.9\textwidth}\footnotesize
\textit{Notes}: ``US'' denotes an \textit{Undominated Strategy}, while ``RUS'' refers to a \textit{Robust Undominated Strategy}. 
\end{minipage}
\end{table}

\noindent
First, we observe that the \textit{Stability} assumption on its own provides limited identifying power. However, when combined with the \textit{Undominated Strategy} assumption, the number of revealed binary relationships increases substantially---from 395 to 4,185 for males and from 374 to 4,185 for females. This highlights the strong identification role played by the \textit{Undominated Strategy}. 

\medskip

The \textit{Robust Undominated Strategy} assumption also exhibits considerable identifying power relative to \textit{Stability}, though its contribution remains modest compared to the \textit{Undominated Strategy}. Specifically, the number of revealed binary relationships increases from 395 to 1,169 for males and from 374 to 1,103 for females when moving from \textit{Stability} alone to \textit{Stability + Robust Undominated Strategy}. Enhancing its usefulness would require a richer specification of the sets $\mathcal{F}_\omega$.

\medskip

Finally, the marginal contribution of the \textit{Selectiveness} assumption depends critically on the set of accompanying assumptions. When combined with \textit{Stability} and the \textit{Undominated Strategy}, it provides substantially more identifying information than when paired with \textit{Stability} and \textit{Robust Undominated Strategy}. 

\medskip

Regarding gender differences, under the \textit{Stability} assumption, males reveal on average 5.7\% more information than females. However, this gap narrows considerably when \textit{Stability} is combined with the \textit{Undominated Strategy}, the gap falls to about 1\% in favor of females. This reduction may stem from the fact that males and females report a similar average number of programs on their rank-ordered lists (4.7 for males versus 4.6 for females).


In what follows, we implement the counterfactual analysis under the most informative set of assumptions, namely \textit{Stability + Undominated Strategy + Selectiveness} (Assumptions~\ref{ass:stable}, \ref{assu:undominance}, and~\ref{assu:selectiveness}). To conduct this exercise, we apply the dimension-reduction approach introduced in Section~\ref{section: Dimension-Red}. The counterfactual results presented below are computed using a representative 10\% subsample of the student population.

\subsubsection{More Weight on GPA}\label{sec:more_weights_GPA}

The admission score is computed as a weighted average of the high-stakes exam scores and the GPA, as follows\footnote{Very rarely, some programs compute $S_j$ as 
$S_j = \alpha_{1j} S_{j}^{\text{Lang}} + \alpha_{2j} S_{j}^{\text{Math}} + \alpha_{3j} \max(S_{j}^{\text{Sci}}, S_{j}^{\text{Hist}}) + \alpha_{5j} S_{j}^{\text{GPA}}.$
}:
\[
S_j = \alpha_{1j} S_{j}^{\text{Lang}} 
    + \alpha_{2j} S_{j}^{\text{Math}} 
    + \alpha_{3j} S_{j}^{\text{Sci}} 
    + \alpha_{4j} S_{j}^{\text{Hist}} 
    + \alpha_{5j} S_{j}^{\text{GPA}}.
\]

In the counterfactual exercise, we retain the original scores but modify the weighting scheme to place greater emphasis on the GPA component. Specifically, we define:
\[
\tilde{S}_j = \tilde{\alpha}_{1j} S_{j}^{\text{Lang}} 
            + \tilde{\alpha}_{2j} S_{j}^{\text{Math}} 
            + \tilde{\alpha}_{3j} S_{j}^{\text{Sci}} 
            + \tilde{\alpha}_{4j} S_{j}^{\text{Hist}} 
            + \tilde{\alpha}_{5j} S_{j}^{\text{GPA}},
\]
where, for each program $j$, the recalibrated weights $\{\tilde\alpha_{kj}\}_{k=1}^5$ allocate at least $X\%$ to GPA, with $X\in\{30,40,50,60,70,80\}$. That is, $\tilde{\alpha}_{5j} = \max(\alpha_{5j}, X\%)$. The relative weights across the four exam subjects are preserved and the weights are renormalized to keep the total unchanged in the counterfactual. This counterfactual design allows us to examine how increasing the relative importance of GPA—often viewed as a measure of sustained academic effort—affects the gender gap in the percentage of students admitted to STEM and Medicine programs.

Figure~\ref{fig:more_weights} summarizes the counterfactual results. The dashed red line marks the observed gender gap in STEM \& Med admissions, $0.132$ (i.e., the 13.2 percentage points). The vertical blue segments plot the bounds on the counterfactual gap under policies that require each program to allocate at least $X\%$ of the total weight to GPA, with $X\in\{30,40,50,60,70,80\}$. The upper bound falls below the observed gap once programs allocate at least $70\%$ to GPA. At $80\%$, the upper bound declines to $0.121$, indicating that placing more weight on sustained academic performance (GPA) can narrow the observed STEM \& Med gender gap.

Figure~\ref{fig:quantiles-GPA} extends the analysis by plotting bounds on the counterfactual gender gap by percentile of the sum of Mathematics and Language exam-score distributions, for $X\in\{50,60,70,80\}$. Panel~(a) ($X=50$) shows that, while the aggregate gap remains large, the gender gap narrows among the top decile. As the lower bound for GPA weights rises (Panels~(b)–(d), $X=60,70,80$), the reduction extends progressively down the distribution---reaching roughly the 70th percentile---indicating that higher GPA weights reduces the gender gap in a broader set of high-performing students.

This first set of counterfactuals yields two insights. First, the weighting placed on GPA versus exams has a measurable impact on the gender gap in STEM \& Med admissions. Second, raising the GPA weight produces heterogeneous effects across achievement levels, with especially pronounced reductions among higher-achieving applicants. Taken together, these findings indicate that emphasizing GPA in admissions has meaningful redistributive effects and can promote greater gender equity in access to STEM \& Med programs, primarily by shifting outcomes in the upper tail of the achievement distribution.

\begin{figure}[h!]
   \centering
   \includegraphics[width=0.9\linewidth]{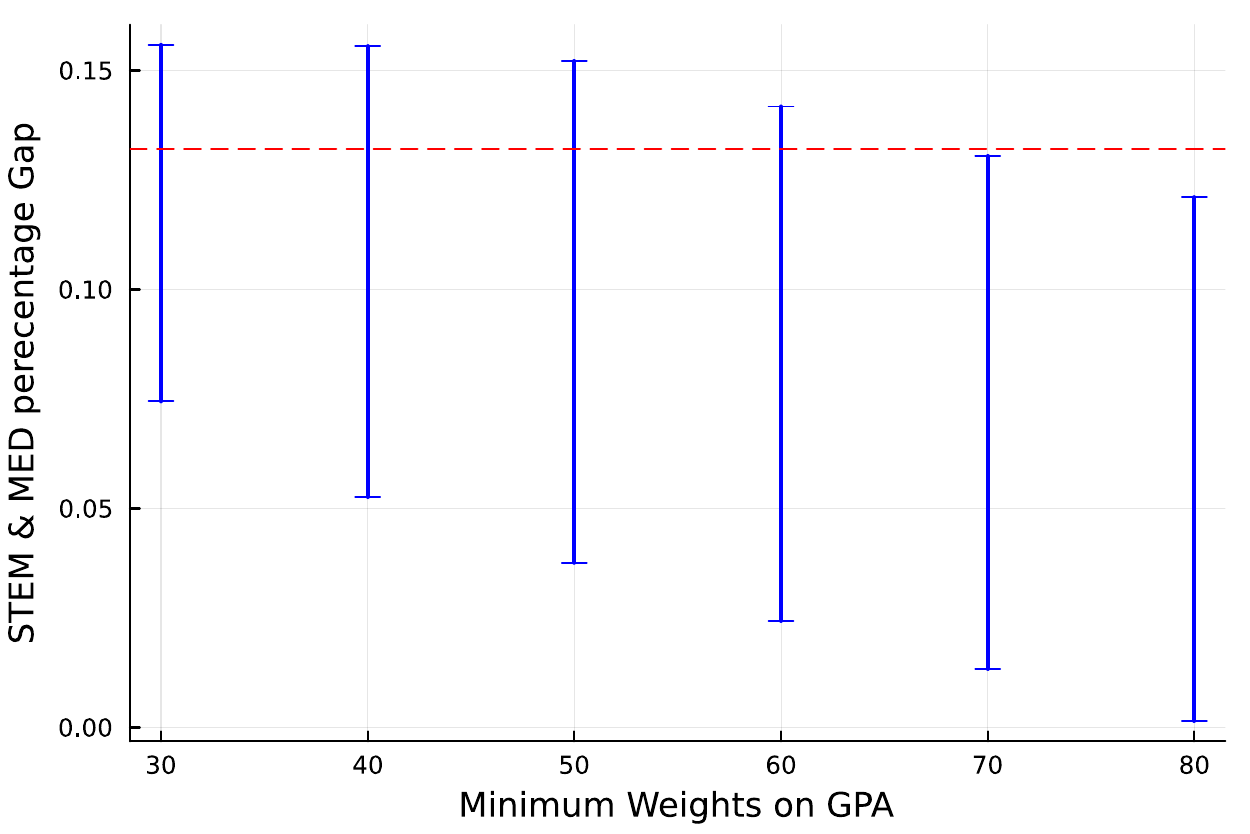}
   \caption{Counterfactuals: Percentage gender gap in STEM \& Medicine admissions, 2010. The red dashed line denotes the observed gap ($0.132$), while the blue lines indicate bounds under alternative weighting schemes with increasing GPA importance.}
   \label{fig:more_weights}
\end{figure}

\begin{figure}[htbp]
\centering
\captionsetup[subfigure]{justification=centering}

\begin{subfigure}[t]{0.45\textwidth}
    \centering
    \includegraphics[width=\linewidth]{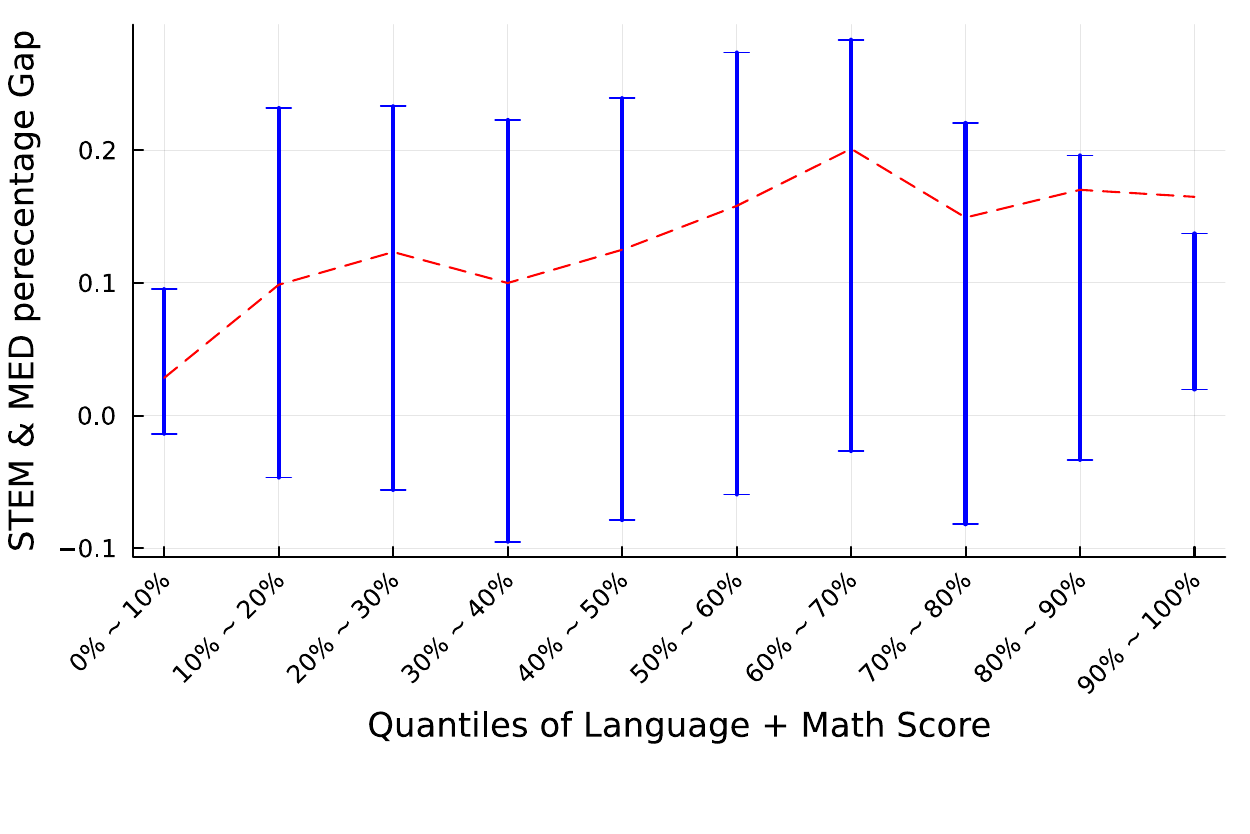}
    \caption{Panel A: At least 50\% GPA weight}
    \label{fig:subA}
\end{subfigure}
\hfill
\begin{subfigure}[t]{0.45\textwidth}
    \centering
    \includegraphics[width=\linewidth]{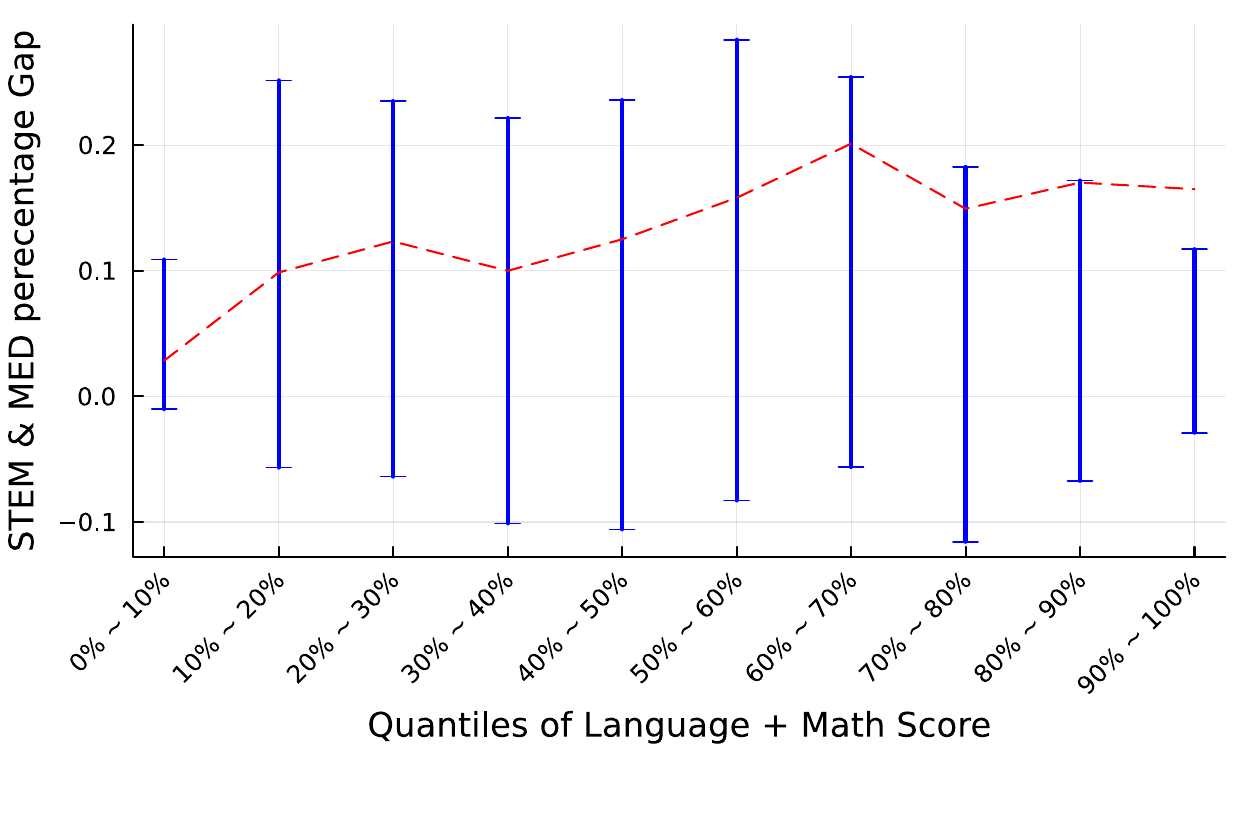}
    \caption{Panel B: At least 60\% GPA weight}
    \label{fig:subB}
\end{subfigure}

\vspace{1em} 

\begin{subfigure}[t]{0.45\textwidth}
    \centering
    \includegraphics[width=\linewidth]{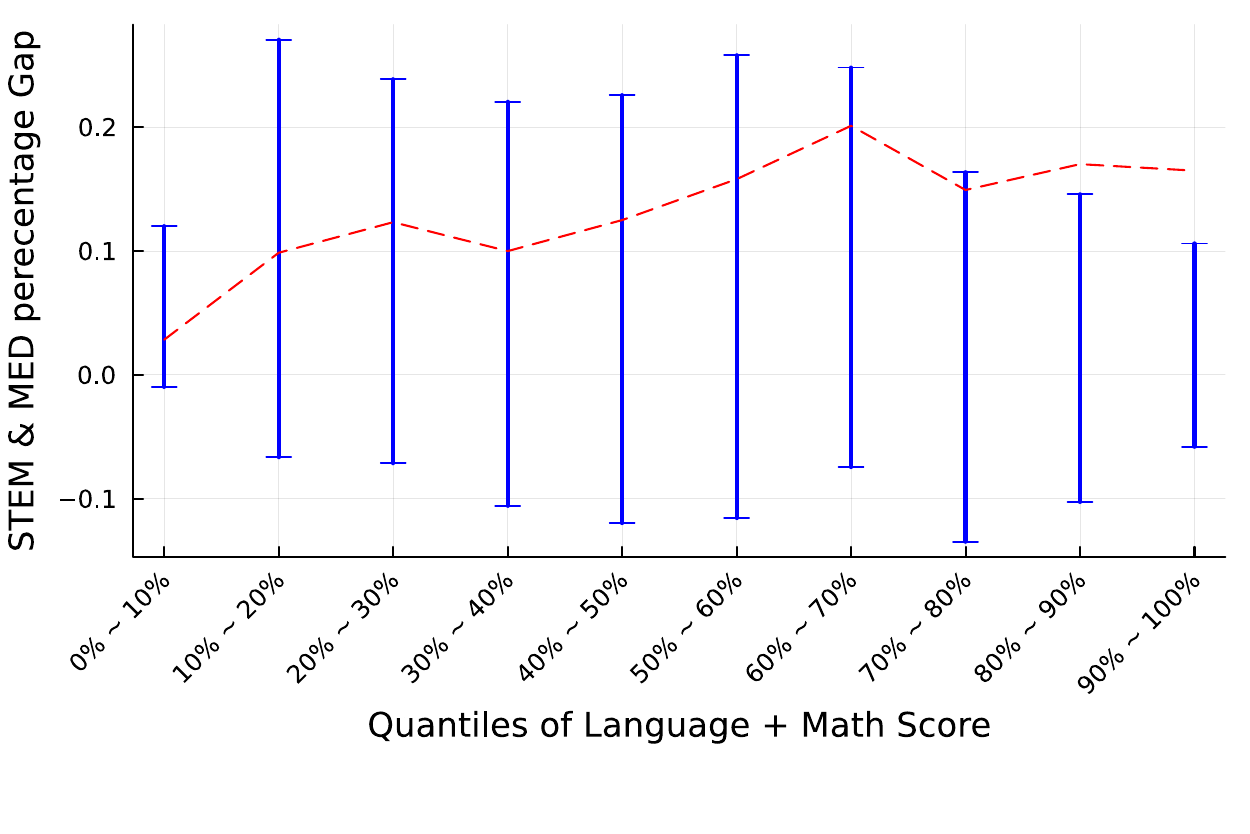}
    \caption{Panel C: At least 70\% GPA weight}
    \label{fig:subC}
\end{subfigure}
\hfill
\begin{subfigure}[t]{0.45\textwidth}
    \centering
    \includegraphics[width=\linewidth]{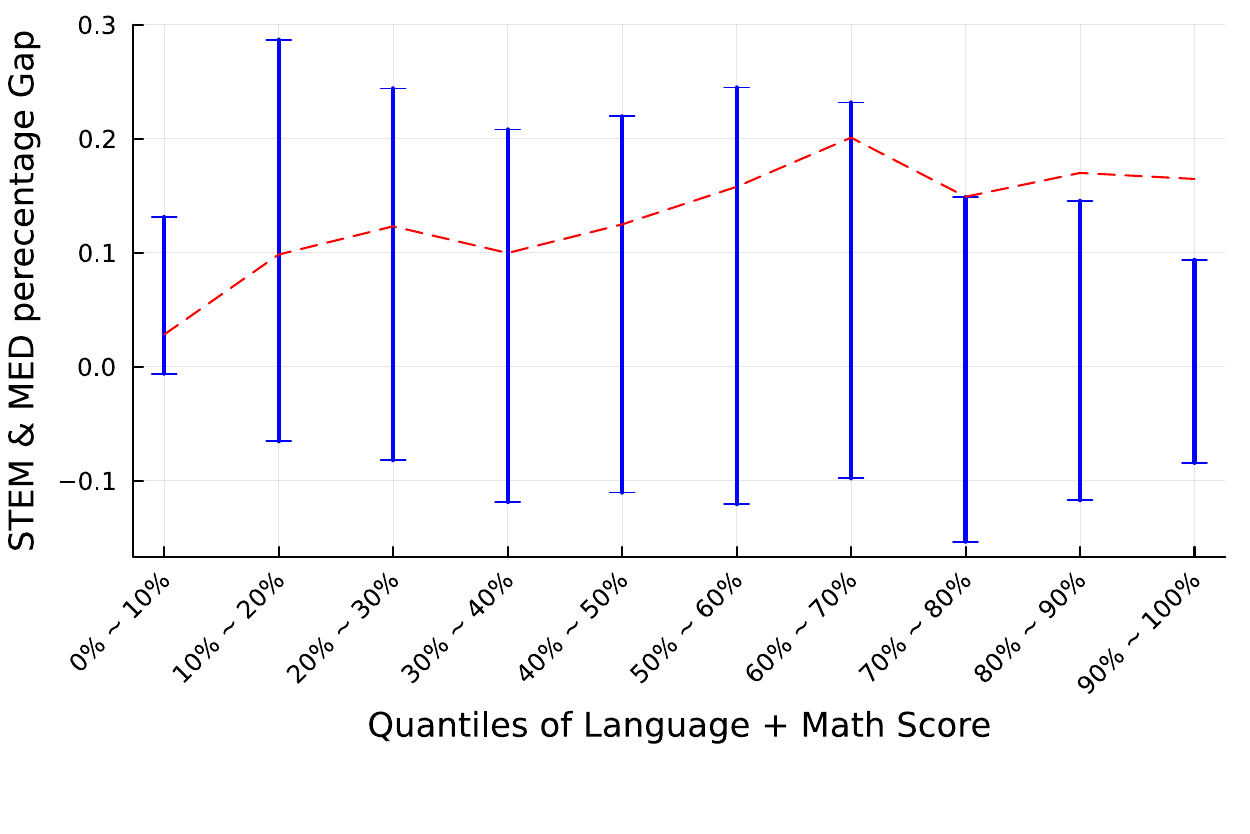}
    \caption{Panel D: At least 80\% GPA weight}
    \label{fig:subD}
\end{subfigure}

\caption{Counterfactual gender gap bounds across Math and Language percentiles for different GPA weighting thresholds.}
\label{fig:quantiles-GPA}
\end{figure}

\subsubsection{Gender-based standardization.} \label{sec:gender_standardization}

In this counterfactual exercise, we keep the weighting scheme currently used by each academic program but replace raw exam scores with within-gender standardized scores. Specifically,
\[
\tilde{S}_j = \alpha_{1j} \tilde{S}_{j}^{\text{Lang}} 
            + \alpha_{2j} \tilde{S}_{j}^{\text{Math}} 
            + \alpha_{3j} \tilde{S}_{j}^{\text{Sci}} 
            + \alpha_{4j} \tilde{S}_{j}^{\text{Hist}} 
            + \alpha_{5j} S_{j}^{\text{GPA}},
\]
where $\tilde{S}_{j}^{X}$ denotes the gender-normalized score for exam $X$, obtained by standardizing each student's score relative to the mean and standard deviation within their gender group.
Figure~\ref{fig:Stand} displays the estimated bounds on the STEM and Medicine gender gaps under this gender-based standardization scenario. The results show a substantial reduction in the upper bound of the gap—from 0.132 in the observed data to 0.098 in the counterfactual. For comparison, when implementing an alternative policy that assigns at least 80\% weight to the GPA, the corresponding upper bound is 0.121. Thus, gender-based standardization appears more effective at narrowing the STEM and Medicine gender gaps than increasing the relative weight on the GPA.

Figure~\ref{fig:quantile_Stand} examines how within-gender standardization affects the gender gap across the exam-performance distribution. At the 90th percentile of Mathematics and Language scores, the upper bound of the gap declines from 0.164 in the observed data to 0.079 under standardization; at the 80th percentile, it falls from 0.170 to 0.118. These patterns indicate that gender-based standardization of high-stakes exam scores can meaningfully reduce disparities in access to STEM \& Med programs, with the largest effects among high-achieving students. This echoes our first set of counterfactuals: policies that target the admissions scoring structure have their most pronounced impact in the upper tail.

\begin{figure}[h!]
   \centering
   \includegraphics[width=0.9\linewidth]{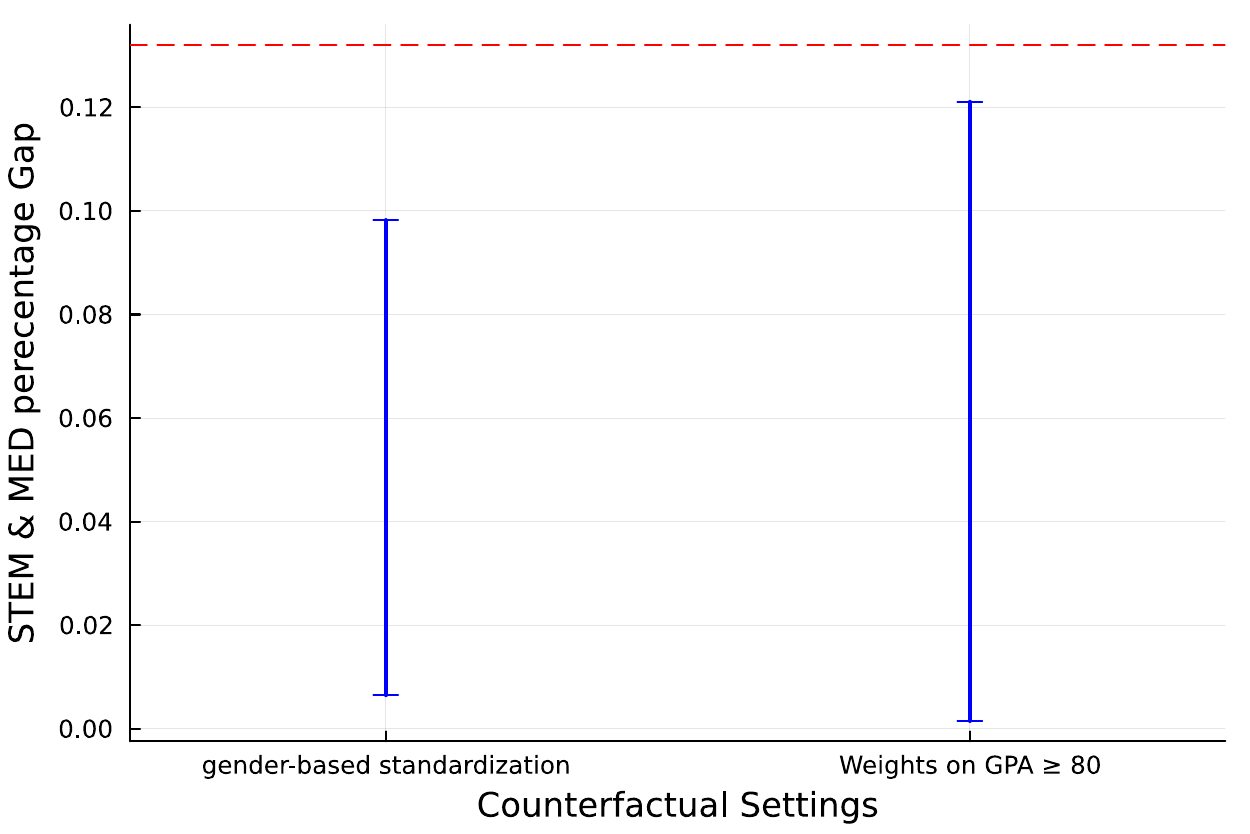}
   \caption{Counterfactual bounds on the STEM \& Med gender gaps under gender-based score standardization and under a GPA-weighting policy. The figure compares the observed bounds with those obtained when (i) exam scores are standardized within gender groups, and (ii) at least 80\% of the admission score weight is assigned to GPA.}
   \label{fig:Stand}
\end{figure}

\begin{figure}[h!]
   \centering
   \includegraphics[width=0.9\linewidth]{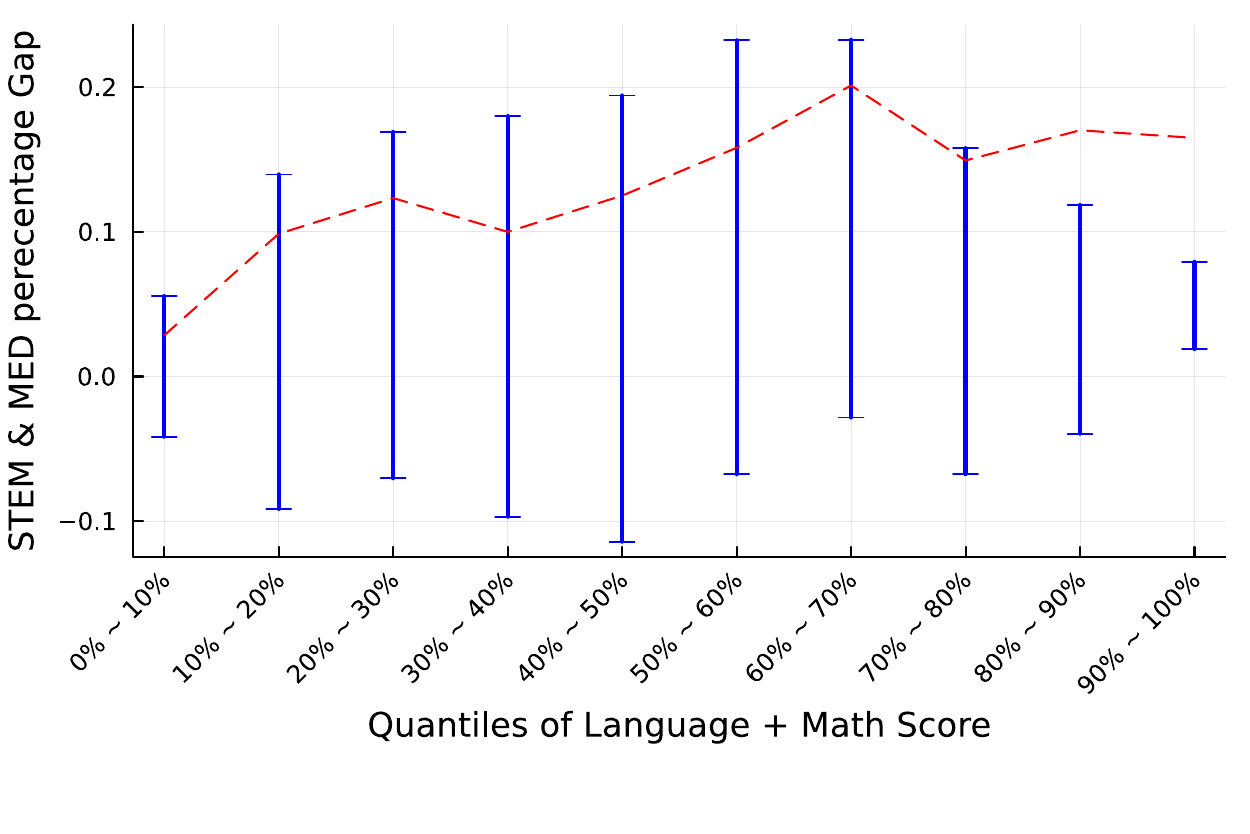}
   \caption{Counterfactual bounds on the STEM and Medicine gender gaps across the distribution of high-stakes exam performance. Gender-based standardization yields substantial gap reductions, especially among students in the top percentiles of Math and Language performance.}
   \label{fig:quantile_Stand}
\end{figure}

\section{Conclusion}

\noindent
Counterfactual analysis plays a central role in education market design and serves as a foundation for credible policy recommendations. 
However, empirical researchers often rely on highly parameterized models to ensure analytical tractability and deliver precise counterfactual predictions. 
This pursuit of tractability and precision frequently comes at the cost of credibility, as policy conclusions derived from overly restrictive or ad hoc assumptions may lack robustness. 
Hence, there exists a fundamental tension between the strength of the assumptions underlying counterfactual exercises and the credibility of the policy insights that follow from them.

\medskip
\noindent
This paper proposes a novel methodology for conducting counterfactual analysis within Gale--Shapley DA assignment mechanisms under weaker identifying assumptions and without imposing any ad hoc parameterization. 
The relaxation of standard modeling assumptions leads to an incomplete model, which makes the counterfactual analysis more challenging. 
Nevertheless, we demonstrate that computing such counterfactuals remains computationally feasible by leveraging the ILP representation of the DA mechanism combined with dimension-reduction techniques.
We apply our approach to evaluate two policy interventions designed to increase female enrollment in STEM fields in Chile: 
(i) reallocating admission weights from standardized exams toward GPA, and 
(ii) implementing within-gender standardization of exam scores. 
Despite the absence of any parametric structure, the resulting sharp bounds on the counterfactual outcomes are sufficiently informative to assess the potential impact of these policies on gender disparities in STEM and Medicine programs.

\medskip
\noindent
Our findings suggest that both policies could reduce gender gaps in admissions to STEM programs. These results indicate that adjustments to priority weights and within-gender normalization of exam scores can meaningfully mitigate the influence of gender-based performance differences and contribute to narrowing the gender gap in STEM enrollment.

A natural next question is: for a given policy objective, what choice of weights is optimal? Our analysis can be extended to address this question provided that reweighting does not alter students' effort. If changes in weights do affect effort, one must explicitly model endogenous effort (and the resulting equilibrium responses). We leave this extension to future research.

\clearpage
\bibliography{biblio}

@article{fack_beyond_2019,
  author =        {Fack, Gabrielle and Grenet, Julien and He, Yinghua},
  journal =       {American Economic Review},
  month =         apr,
  number =        {4},
  pages =         {1486--1529},
  title =         {Beyond {Truth}-{Telling}: {Preference} {Estimation}
                   with {Centralized} {School} {Choice} and {College}
                   {Admissions}},
  volume =        {109},
  year =          {2019},
  language =      {en},
}

@article{barahona2021,
  title={An Economic Perspective on Regulatory Capture and Its Implications},
  author={Barahona, Rodrigo and Dobbin, Frank and Otero, Carlos},
  journal={Journal of Economic Perspectives},
  volume={35},
  number={4},
  pages={75--94},
  year={2021},
  publisher={American Economic Association},
  doi={10.1257/jep.35.4.75}
}

@article{daymont1984,
  author    = {Thomas N. Daymont and Paul G. Andrisani},
  title     = {Job Preferences, College Major, and the Gender Gap in Earnings},
  journal   = {Journal of Human Resources},
  volume    = {19},
  number    = {3},
  pages     = {408--428},
  year      = {1984}
}

@article{bello2020stem,
  author    = {Bello, Amparo},
  title     = {Breaking the STEM Glass Ceiling in Latin America},
  journal   = {Development Policy Review},
  year      = {2020},
  volume    = {38},
  number    = {S1},
  pages     = {O1--O22},
  publisher = {Wiley},
}

@techreport{uribe2021gender,
  author      = {Uribe, Camila},
  title       = {Gender Gaps in STEM in Latin America: Evidence and Policy Perspectives},
  institution = {Inter-American Development Bank},
  year        = {2021},
  type        = {Policy Brief},
  number      = {IDB-PB-347},
}

@article{zafar2013,
  author    = {Basit Zafar},
  title     = {College Major Choice and the Gender Gap},
  journal   = {Journal of Human Resources},
  volume    = {48},
  number    = {3},
  pages     = {545--595},
  year      = {2013}
}

@article{kapor2024aftermarket,
  title={Aftermarket Frictions and the Cost of Off-Platform Options in Centralized Assignment Mechanisms},
  author={Kapor, Adam and Karnani, Mohit and Neilson, Christopher},
  journal={Journal of Political Economy},
  volume={132},
  number={7},
  pages={2346--2395},
  year={2024},
  publisher={University of Chicago Press},
  doi={10.1086/729068}
}

@article{almlund2011,
	author = {Mathilde Almlund and Angela Duckworth and Jim Heckman and Tim Kautz},
	journal = {Handbook of the Economics of Education},
	pages = {1--188},
	title = {Personality Psychology and Economics},
	volume = {4},
	year = {2011}}

@article{borghans2008,
	author = {Lex Borghans and Angela Duckworth and Jim Heckman and Bas ter Weel},
	journal = {Journal of Human Resources},
	pages = {972--1059},
	title = {The Economics and Psychology of Personality Traits},
	volume = {43},
	year = {2008}}

@article{heckmankautz2012,
	author = {Jim Heckman and Tim Kautz},
	journal = {Labor Economics},
	pages = {451--464},
	title = {Hard Evidence on Soft Skills},
	volume = {19},
	year = {2012}}

@incollection{SalesMoses2014,
  author       = {Sales, S. and Moses, M. S.},
  title        = {Brazil: Enhancing opportunity and justice through new affirmative action policies for black and mixed-race students},
  booktitle    = {Affirmative Action Matters: Creating Opportunities for Students Around the World},
  year         = {2014},
  publisher    = {Routledge},
  address      = {New York, NY},
}

@article{azevedo_supply_2016,
  author =        {Azevedo, Eduardo M. and Leshno, Jacob D.},
  journal =       {Journal of Political Economy},
  month =         oct,
  number =        {5},
  pages =         {1235--1268},
  title =         {A {Supply} and {Demand} {Framework} for {Two}-{Sided}
                   {Matching} {Markets}},
  volume =        {124},
  year =          {2016},
  language =      {en},
}

@article{JurajdaMunich2011,
  author       = {Štěpán Jurajda and Daniel Münich},
  title        = {Gender Gap in Performance under Competitive Pressure: Admissions to Czech Universities},
  journal      = {American Economic Review},
  volume       = {101},
  number       = {3},
  pages        = {514--518},
  year         = {2011},
  month        = may,
  doi          = {10.1257/aer.101.3.514}
}

@article{Saygin2019,
  author       = {Perihan O. Saygın},
  title        = {Gender Bias in Standardized Tests: Evidence from a Centralized College Admissions System},
  journal      = {Empirical Economics},
  volume       = {59},
  number       = {2},
  pages        = {1037--1065},
  year         = {2019},
  month        = mar,
  doi          = {10.1007/s00181-019-01662-z}
}

@article{MontolioTaberner2021,
  author       = {Daniel Montolio and Pere A. Taberner},
  title        = {Gender Differences under Test Pressure and Their Impact on Academic Performance: A Quasi‑Experimental Design},
  journal      = {Journal of Economic Behavior \& Organization},
  volume       = {191},
  pages        = {1065--1090},
  year         = {2021}
}

@article{IriberriReyBiel2019,
  author       = {Nagore Iriberri and Pedro Rey‑Biel},
  title        = {Competitive Pressure Widens the Gender Gap in Performance: Evidence from a Two‑Stage Competition in Mathematics},
  journal      = {The Economic Journal},
  volume       = {129},
  number       = {620},
  pages        = {1863--1893},
  year         = {2019},
  doi          = {10.1111/ecoj.12617}
}

@article{ArenasCalsamiglia2025,
  author       = {Andreu Arenas and Caterina Calsamiglia},
  title        = {Gender Differences in High‑Stakes Performance and College Admission Policies},
  journal      = {Management Science},
  year         = {2025},
  doi          = {10.1287/mnsc.2023.02979},
  issn         = {0025-1909},
}

@article{Artemov2023,
  title={Stable Matching with Mistaken Agents},
  author={Georgy Artemov and Yeon-Koo Che and YingHua He},
  journal={Journal of Political Economy Microeconomics},
  volume={1},
  number={2},
  pages={270--320},
  year={2023},
  publisher={University of Chicago Press},
  doi={10.1086/722978},
  url={https://www.journals.uchicago.edu/doi/10.1086/722978}
}

@article{black2008college,
  title     = {The College Type: Returns to College for Marginal Students},
  author    = {Black, Dan A. and Daniel, Kermit G. and Smith, Jeffrey A.},
  journal   = {Journal of Labor Economics},
  year      = {2008},
  volume    = {23},
  number    = {3},
  pages     = {629--699},
  doi       = {10.1086/430285}
}

@article{blau2017gender,
  title     = {The Gender Wage Gap: Extent, Trends, and Explanations},
  author    = {Blau, Francine D. and Kahn, Lawrence M.},
  journal   = {Journal of Economic Literature},
  year      = {2017},
  volume    = {55},
  number    = {3},
  pages     = {789--865},
  doi       = {10.1257/jel.20160995}
}

@article{dahl2023gender,
  title     = {The Gender Gap in Earnings at Career Entry},
  author    = {Dahl, Gordon B. and L{\o}ken, Katrine V. and Mogstad, Magne and Salvanes, Kari Vea},
  journal   = {American Economic Review: Insights},
  year      = {2023},
  volume    = {5},
  number    = {2},
  pages     = {123--140},
  doi       = {10.1257/aeri.20210624}
}

@article{Negrotto2025GenderGap,
  author       = {Laura Negrotto and Teresita Corona V{\'a}zquez and Carolina Iba{\~n}ez and Paola Montenegro and Priscilla Monterrey and Carmen Pupareli and Erika Ruiz‐Garcia and Clara In{\'e}s Saldarriaga Giraldo and Marise Samama and Tannia Soria Samaniego and Taysser Sowley and Mar{\'\i}a C{\'e}lica Ysrraelit and Ivonne Andrea D{\'\i}az},
  title        = {Gender gap in medicine: a call to action for Latin America},
  journal      = {Human Resources for Health},
  volume       = {23},
  number       = {50},
  year         = {2025},
  publisher    = {BioMed Central},
  doi          = {10.1186/s12960-025-00998-1},
  url          = {https://human-resources-health.biomedcentral.com/articles/10.1186/s12960-025-00998-1},
}

@techreport{beede2011women,
  title     = {Women in STEM: A Gender Gap to Innovation},
  author    = {Beede, David N. and Julian, Tiffany A. and Langdon, David and McKittrick, George and Khan, Beethika and Doms, Mark E.},
  institution = {U.S. Department of Commerce, Economics and Statistics Administration},
  year      = {2011},
  number    = {ESA Issue Brief \#04-11},
  url       = {https://files.eric.ed.gov/fulltext/ED523766.pdf}
}

@article{james1989college,
  title     = {College Quality and Future Earnings: Where Should You Send Your Child to College?},
  author    = {James, Estelle and Alsalam, Nabeel and Conaty, Joseph C. and To, Duc-Le},
  journal   = {American Economic Review},
  year      = {1989},
  volume    = {79},
  number    = {2},
  pages     = {247--252}
}

@article{dubins_machiavelli_1981,
  author =        {Dubins, L. E. and Freedman, D. A.},
  journal =       {The American Mathematical Monthly},
  month =         aug,
  number =        {7},
  pages =         {485--494},
  title =         {Machiavelli and the {Gale}-{Shapley} {Algorithm}},
  volume =        {88},
  year =          {1981},
  language =      {en},
}

@article{roth_economics_1982,
  author =        {Roth, Alvin E.},
  journal =       {Mathematics of Operations Research},
  month =         nov,
  number =        {4},
  pages =         {617--628},
  title =         {The {Economics} of {Matching}: {Stability} and
                   {Incentives}},
  volume =        {7},
  year =          {1982},
  language =      {en},
}

@article{azevedo_strategy-proofness_2018,
  author =        {Azevedo, Eduardo M and Budish, Eric},
  journal =       {The Review of Economic Studies},
  month =         aug,
  title =         {Strategy-proofness in the {Large}},
  year =          {2018},
  language =      {en},
}

@article{haeringer_constrained_2009,
  author =        {Haeringer, Guillaume and Klijn, Flip},
  journal =       {Journal of Economic Theory},
  month =         sep,
  number =        {5},
  pages =         {1921--1947},
  title =         {Constrained school choice},
  volume =        {144},
  year =          {2009},
  language =      {en},
}

@article{kojima_incentives_2009,
  author =        {Kojima, Fuhito and Pathak, Parag A},
  journal =       {American Economic Review},
  month =         may,
  number =        {3},
  pages =         {608--627},
  title =         {Incentives and {Stability} in {Large} {Two}-{Sided}
                   {Matching} {Markets}},
  volume =        {99},
  year =          {2009},
  language =      {en},
}

@misc{bertanha_causal_2024,
  author =        {Bertanha, Marinho and Luflade, Margaux and
                   Mourifié, Ismael},
  month =         may,
  title =         {Causal {Effects} in {Matching} {Mechanisms} with
                   {Strategically} {Reported} {Preferences}},
  year =          {2024},
}

@phdthesis{ben-akiva1973,
  author       = {Ben-Akiva, Moshe E.},
  title        = {Structure of Passenger Travel Demand Models},
  school       = {Massachusetts Institute of Technology},
  year         = {1973},
  type         = {Ph.D. dissertation}
}

@article{barseghyan2021,
  author       = {Barseghyan, Levon and Molinari, Francesca and Thirkettle, William},
  title        = {Discrete Choice under Risk with Limited Consideration},
  journal      = {Review of Economic Studies},
  year         = {2021},
  volume       = {88},
  number       = {6},
  pages        = {2806--2842}
}

@article{cattaneo2020,
  author       = {Cattaneo, Matias D. and Ma, Xinwei and Masatlioglu, Yusuf and Suleymanov, Efe},
  title        = {A Random Attention Model},
  journal      = {Econometrica},
  year         = {2020},
  volume       = {88},
  number       = {6},
  pages        = {1993--2018}
}

@techreport{kahn_women_2017,
  address =       {Cambridge, MA},
  author =        {Kahn, Shulamit and Ginther, Donna},
  institution =   {National Bureau of Economic Research},
  month =         jun,
  number =        {w23525},
  pages =         {w23525},
  title =         {Women and {STEM}},
  year =          {2017},
  language =      {en},
}

@article{ngo_preferences_2024,
  author =        {Ngo, Diana and Dustan, Andrew},
  journal =       {American Economic Journal: Applied Economics},
  title =         {Preferences, access, and the {STEM} gender gap in
                   centralized high school assignment},
  year =          {2024},
}

@article{agarwal_demand_2018,
  author =        {Agarwal, Nikhil and Somaini, Paulo},
  journal =       {Econometrica},
  number =        {2},
  pages =         {391--444},
  title =         {Demand {Analysis} {Using} {Strategic} {Reports}: {An}
                   {Application} to a {School} {Choice} {Mechanism}},
  volume =        {86},
  year =          {2018},
  language =      {en},
}

@article{Rothblum1992,
  author    = {Uriel G. Rothblum},
  title     = {Characterization of stable matchings as extreme points of a polytope},
  journal   = {Mathematical Programming},
  year      = {1992},
  volume    = {54},
  number    = {1-3},
  pages     = {57--67},
  doi       = {10.1007/BF01586044}
}

@article{Roth1993,
  author    = {Alvin E. Roth and Uriel G. Rothblum and John H. Vande Vate},
  title     = {Stable Matchings, Optimal Assignments, and Linear Programming},
  journal   = {Mathematics of Operations Research},
  year      = {1993},
  volume    = {18},
  number    = {4},
  pages     = {803--828},
  doi       = {10.1287/moor.18.4.803}
}

@article{Teo1998,
  author    = {Chung-Piaw Teo and Jay Sethuraman},
  title     = {The geometry of fractional stable matchings and its applications},
  journal   = {Mathematics of Operations Research},
  year      = {1998},
  volume    = {23},
  number    = {4},
  pages     = {874--891},
  doi       = {10.1287/moor.23.4.874}
}

@article{Baiou2000,
  author    = {Meriem Bai'ou and Michel L. Balinski},
  title     = {Many-to-many matching: stable polyandrous polygamy (or stable marriage of the Gulf States)},
  journal   = {Mathematical Programming},
  year      = {2000},
  volume    = {91},
  number    = {3},
  pages     = {1--12},
  doi       = {10.1007/s101070000192}
}

@article{Akyol2016,
  title={Preferences, Selection, and Value Added: A Structural Approach},
  author={Akyol, Şaziye Pelin and Krishna, Kala},
  journal={European Economic Review},
  volume={84},
  pages={123--141},
  year={2016},
  doi={10.1016/j.euroecorev.2016.01.004}
}

@techreport{Bucarey2018,
  author = {Bucarey, A.},
  title = {Who Pays for Free College? Crowding Out on Campus},
  year = {2018},
  institution = {Massachusetts Institute of Technology},
  type = {Job Market Paper},
  address = {Cambridge, MA}
}

@article{Agarwal2020,
  title={Revealed Preference Analysis of School Choice Models},
  author={Nikhil Agarwal and Paulo Somaini},
  journal={Annual Review of Economics},
  volume={12},
  number={1},
  pages={471--501},
  year={2020},
  doi={10.1146/annurev-economics-082019-112339}
}

@article{hastingsneilsonzimmerman2013,
  title={Are Some Degrees Worth More Than Others? Evidence From College Admissions Cutoffs in Chile},
  author={Hastings, Justine and Neilson, Christopher and Zimmerman, Seth},
  journal={NBER Working paper No. 19241},
  year={2013}
}

@article{larroucaurios2020,
author={Larroucau, Tomas and Rios, Ignacio},
title = {Do ``Short-List'' Students Report Truthfully? Strategic Behavior in the Chilean College Admissions Problem},
journal = {Working paper},
year={2020}
}

@article{larroucaurios2021,
author={Larroucau, Tomas and Rios, Ignacio},
title = {Dynamic College Admissions},
journal = {Working paper},
year={2021}
}

\begin{appendices}

\section{Proofs of Results in Section \ref{sec:revealed_preference}}

\subsection{Proof of Proposition \ref{Prop:PS}}\label{proof:stability}
We use $F_\omega$ as a shorthand notation for the feasible set of student $\omega$. Moreover, the following proof, we utilize the following equality implied by the DA mechanism:
\begin{equation*}
	\mu(\omega) = \optchoice{\succ^R_\omega}{F_\omega}.
\end{equation*}

\paragraph{if part}
Suppose, for any student $\omega$, $\succ^Q_\omega$ is compatible with $\succ^{P_s}_\omega$. This means, for any $j\in F_\omega \setminus \{\mu(\omega)\}$, we have $\mu(\omega) \succ^Q_j j$. As a result, there cannot exist a $(\omega, j)$ pair such that $j\in F_\omega$ and $j \succ^Q_\omega \mu(\omega)$. In other words, there cannot exist a blocking pair for the matching $\mu$. As a result, $\mu$ is a stable matching.

\paragraph{only if part}
Suppose the $\mu$ is a stable matching. Supose, for the purpose of contradiction, there exist some $\omega$ such that $\succ^Q_\omega$ is not compatible with $\succ^{P_s}_\omega$. In this, there must exist some $j\in F_\omega$ such that $j \succ^Q_\omega \mu(\omega)$. However, such $(\omega, j)$ pair forms a blocking pair, because $j \succ^Q_\omega \mu(\omega)$, and because $j\in F_\omega$ implies student $\omega$'s priority score $S_{\omega j}$ is higher than that of someone to which program $j$ is matched under matching $\mu$. This contradicts $\mu$ being a stable matching.

\subsection{Proof of Proposition \ref{Prop:undominated_strategy} }\label{proof:undominated_strategy}

\paragraph{if part}
Suppose, for any student $\omega$, $\succ^Q_\omega$ is compatible with $\succ^{P_u}_\omega$. Fix an arbitrary student $\omega$. We want to show $\succ^R_\omega$ is undominated given this $\succ^Q_\omega$. We consider two scenarios.
\begin{itemize}
	\item Suppose $|\succ^R_\omega| = K$. Suppose, for the purpose of contradiction, there exists $\succ^{R'}_\omega$ that dominates $\succ^R_\omega$. Then, there must exist some set $F \subseteq \mathcal{J}_0$ with $0\in F$ such that $\optchoice{\succ^{R'}_\omega}{F} \succ^Q_\omega \optchoice{\succ^R_\omega}{F}$. Because $\succ^{P_u}_\omega = \succ^R_\omega$ in this case, the compatibility between $\succ^Q_\omega$ and $\succ^{P_u}_\omega$ implies the compatibility between $\succ^R_\omega$ and $\succ^Q_\omega$. Because $\succ^Q_\omega$ is compatible with $\succ^R_\omega$, $\optchoice{\succ^{R'}_\omega}{F} $ cannot be in the domain of $\succ^R_\omega$. Because $\text{domain}(\succ^{R'}_\omega)$ contains an element not in $\text{domain}(\succ^{R}_\omega)$, and because $|\succ^R_\omega| = K$, there must exist  some $j^*\in \text{domain}(\succ^{R}_\omega) \setminus \text{domain}(\succ^{R'}_\omega)$. Construct $F' = \{j^*, 0\}$. Then, $j^* = \optchoice{\succ^{R}_\omega}{F'}$ and $0 = \optchoice{\succ^{R'}_\omega}{F'}$. Because $\succ^Q_\omega$ is compatible with $\succ^R_\omega$, $j^* \succ^Q_\omega 0$ so that  $\optchoice{\succ^{R}_\omega}{F'} \succ^Q_\omega \optchoice{\succ^{R'}_\omega}{F'}$. This contradicts to $\succ^R_\omega$ being dominated by $\succ^R_\omega$. Hence, there cannot exist any other $\succ^{R'}_\omega$ that dominates $\succ^R_\omega$.

	\item Suppose $|\succ^R_\omega| < K$. In this case, the compatibility between $\succ^Q_\omega$ and $\succ^{P_u}_\omega$ implies that the domain of $\succ^R_\omega$ includes all acceptable programs, i.e., the set of programs that student $\omega$ prefers over the outside option. This obseration, together with the fact $\succ^Q_\omega$ is compatible with $\succ^{P_u}_\omega$, implies that, for any $F\subseteq \mathcal{J}_0$ with $0\in F$, we have $\optchoice{\succ^R_\omega}{F} = \optchoice{\succ^{P_u}_\omega}{F} = \optchoice{\succ^Q_\omega}{F}$. Hence, there cannot exist any other $\succ^{R'}_\omega$ that dominates $\succ^R_\omega$.
	
\end{itemize}

\paragraph{only if part}
Suppose, for any student $\omega$, $\succ^R_\omega$ is undomianted given $\succ^Q_\omega$. Fix an arbitrary student $\omega$. We want to show that $\succ^Q_\omega$ and $\succ^{P_u}_\omega$ must be compatible. Again, we consider two scenarios.
\begin{itemize}
\item Suppose $|\succ^R_\omega| = K$. In this case, $\succ^{P_u}_\omega = \succ^{R}_\omega$. Suppose, for the purpose of contradiction, $\succ^R_\omega$ and $\succ^Q_\omega$ are not compatible. There there must exist $j_1, j_2$ with $j_1 \succ^R_\omega j_2$ and $j_2 \succ^Q_\omega j_1$. Consider the following two scenarios:
	\begin{itemize}
	\item Suppose $j_2 = 0$. Construct an alternative $\succ^{R'}_\omega$ the same as $\succ^R_\omega$ except dropping all non-acceptable programs (which includes $j_1$). Construct $F' = \{j_1, 0\}$. By construction, $\optchoice{\succ^{R'}_\omega}{F'} = j_2 \succ^Q_\omega j_1 = \optchoice{\succ^R_\omega}{F'}$. For any $F\subseteq \mathcal{J}_0$ with $0\in F$, we have either $\optchoice{\succ^{R}_\omega}{F} = \optchoice{\succ^{R'}_\omega}{F}$ or $\optchoice{\succ^{R}_\omega}{F} \neq \optchoice{\succ^{R'}_\omega}{F}$. In the latter case,  we must have $0 \succ^Q_\omega \optchoice{\succ^{R}_\omega}{F}$ and either $\optchoice{\succ^{R'}_\omega}{F} = 0$ or $\optchoice{\succ^{R'}_\omega}{F}\succ^Q_\omega 0$, because the only difference between $\succ^R_\omega$ and $\succ^{R'}_\omega$ is that $\succ^{R'}_\omega$ drops all the non-acceptable programs. Hence, $\succ^{R'}_\omega$ dominates $\succ^R_\omega$, contradicting to $\succ^R_\omega$ being undominated.
	\item Suppose $j_2 \neq 0$. Let $\succ^{R'}_\omega$ to be an ROL that has the same domain as $\succ^{R}_\omega$ yet order all programs according to $\succ^Q_\omega$. Construct $F' = \{j_1, j_2, 0\}$. Then, $\optchoice{\succ^{R'}_\omega}{F'} = j_2 \succ^{Q}_\omega j_1 = \optchoice{\succ^{R}_\omega}{F'}$. For any $F\subseteq \mathcal{J}_0$ with $0\in F$, we have either $\optchoice{\succ^{R}_\omega}{F} = \optchoice{\succ^{R'}_\omega}{F}$ or $\optchoice{\succ^{R}_\omega}{F} \neq \optchoice{\succ^{R'}_\omega}{F}$. In the latter case,  we must have $\optchoice{\succ^{R'}_\omega}{F} \succ^Q_\omega \optchoice{\succ^{R}_\omega}{F}$ because $\succ^{R'}_\omega$ orders $\optchoice{\succ^{R'}_\omega}{F}$ and $\optchoice{\succ^{R}_\omega}{F}$ according to $\succ^Q_\omega$. Hence,  $\succ^{R'}_\omega$ dominates $\succ^R_\omega$, contradicting to $\succ^R_\omega$ being undominated.
	\end{itemize}
Since we get contradiction in both cases, we must have $\succ^R_\omega$ and $\succ^Q_\omega$ to be compatible. 

\item Suppose $|\succ^R_\omega| < K$. Suppose, for the purpose of contradiction, $\succ^{P_u}_\omega$ and $\succ^Q_\omega$ are not compatible. Then, there must exist some $j_1$ and $j_2$ such that $j_1 \succ^{P_u}_\omega j_2$ and $j_2 \succ^Q_\omega j_1$. Consider three scenarios:
	\begin{itemize}
		\item Suppose $j_1 = 0$. Construct a $\succ^{R'}_\omega$ the same as $\succ^R_\omega$ except adding $j_2$ to the position just above the outside option, i.e., $j \succ^{R'}_\omega j'$ for any $j,j'$ with $j\succ^R_\omega j'$, and $j \succ^{R'}_\omega j_2 \succ^{R'}_\omega 0$ for any $j\in \text{domain}(\succ^R_\omega)\setminus\{0\}$. 
Construct $F' = \{j_2, 0\}$. Because $0 \succ^{P_u}_\omega j_2$, we know $j_2$ is not in the domain of $\succ^R_\omega$. Thus, $j_2 = \optchoice{\succ^{R'}_\omega}{F'} \succ^Q_\omega 0  = \optchoice{\succ^R_\omega}{F'}$.
For any $F\subseteq \mathcal{J}_0$ with $0\in F$, either $\optchoice{\succ^R_\omega}{F} = \optchoice{\succ^{R'}_\omega}{F}$ or $\optchoice{\succ^R_\omega}{F} \neq \optchoice{\succ^{R'}_\omega}{F}$. By the construction of  $\succ^{R'}_\omega$, whenever $\optchoice{\succ^R_\omega}{F} \neq \optchoice{\succ^{R'}_\omega}{F}$, we must have $\optchoice{\succ^{R'}_\omega}{F} = j_2 \succ^Q_\omega 0 = \optchoice{\succ^R_\omega}{F}$. Hence, $\succ^{R}_\omega$ is dominated by $\succ^{R'}_\omega$ contradicting to the hypoethesis that $\succ^R_\omega$ is undominated.
\item Suppose $j_1\neq 0$ and $j_2 = 0$. In this case, we have $j_1 \in \text{domain}(\succ^{R}_\omega)$. Construct a $\succ^{R'}_\omega$ the same as $\succ^{R}_\omega$ except dropping all non-acceptable programs (which includes $j_1$). Following the same argument we give in the case where $|\succ^{R}_\omega| = K$ and $j_2 = 0$, we can show that $\succ^R_\omega$ is dominated by $\succ^{R'}_\omega$, which leads to a contradiction.

\item Suppose $j_1 \neq 0$ and $j_2 \neq 0$. Let $\succ^{R'}_\omega$ to be an ROL that has the same domain as $\succ^{R}_\omega$ yet order all programs according to $\succ^Q_\omega$. Following the same argument we give in the case where $|\succ^{R}_\omega| = K$ and $j_2 \neq 0$, we can show that $\succ^R_\omega$ is dominated by $\succ^{R'}_\omega$, which leads to a contradiction.
	\end{itemize}
\end{itemize}
Thus, we always have a contradiction if $\succ^{P_u}_\omega$ is not compatible with $\succ^Q_\omega$. This completes the proof.

\subsection{Proof of Lemma \ref{thm:join}}
This result is a known result in the literature of orders and lattice, though described in different notations. Here, we prove this result for completeness.

First of all, by our definition of the compatibility, if $\succ^{P}$ and $\succ^{P'}$ are not compatible with each other, there cannot exist a total order compatible with both. 

Next, suppose $\succ^P$ and $\succ^{P'}$ are compatible. We want to show a total order $\succ$ is compatible with both if and only if $\succ$ is compatible with their join. 

\paragraph{if part} If $\succ$ is compatible with the join of $\succ^P$ and $\succ^{P'}$, then $\succ$ is compatible with $\succ^P$ because the join of $\succ^P$ and $\succ^{P'}$ includes all the binary relations in $\succ^P$. Similar argument shows that $\succ$ is compatible with $\succ^{P'}$. 

\paragraph{only if part} Suppose $\succ$ is compatible with both $\succ^{P}$ and $\succ^{P'}$. Let $j$ and $j'$ be any two programs with $j(\succ^P \vee \succ^{P'}) j'$. We want to show $j \succ j'$. By the definition of $\succ^P \vee \succ^{P'}$, there exists $j_1, ..., j_K$ such that $j_1 = j$, $j_K = j'$ and, for each $k=1,...,K-1$, $j_k \succ^P j_{k+1}$ or $j_{k} \succ^{P'} j_{k+1}$. Because $\succ$ is compatible with both $\succ^{P}$ and $\succ^{P'}$, we must have that, for each $k=1,...,K-1$, $j_k \succ j_{k+1}$. Because $\succ$ is transitive as a total order, we must have $j \succ j'$. This completes the proof.

\subsection{Proof of Proposition \ref{Prop:PS_undominated_strategy}}
We only need to prove that $\succ^{P_s}_\omega$ and $\succ^{P_u}_\omega$ are compatible with each other. The second part of the proposition is a corollary of Lemma \ref{thm:join}, Proposition \ref{Prop:PS}, and Proposition \ref{Prop:undominated_strategy}.

Fix an arbitrary student $\omega$. Construct a total order $\succ_\omega$ as follows:
\begin{itemize}
\item whenever $j \succ^R_\omega j'$, let $j \succ_\omega j'$. 
\item whenever $j \notin \text{domain}(\succ^R_\omega)$, let $0 \succ_\omega j$.
\item for any two $j_1$ and $j_2$ not in $\text{domain}(\succ^R_\omega)$ with $j_1 > j_2$, let $j_1 \succ_\omega j_2$.
\end{itemize}
This $\succ_\omega$ is a total order, where all programs in the domain of $\succ^R_\omega$ are at the top of $\succ_\omega$ and are ranked according to $\succ^R_\omega$, and all the programs outside of the domain of $\succ^R_\omega$ are at the bottom of $\succ_\omega$ and are ranked according to the integer order. 

This constructed $\succ_\omega$ is compatible with $\succ^{P_s}_\omega$, because the DA mechanism implies that 
\begin{equation*}
	\mu(\omega) = \optchoice{\succ^{R}_\omega}{F_\omega}. 
\end{equation*}
and $\mu(\omega)\in \text{domain}(\succ^R_\omega)$ so $\mu(\omega)$ must rank higher in $\succ_\omega$ than all the other feasible programs. This constructed $\succ_\omega$ is also compatible with $\succ^{P_u}_\omega$ as $\succ_\omega$ ranks all programs within the domain of $\succ^R_\omega$ in the same way as $\succ^R_\omega$ and all the programs outside of the domain of  $\succ^R_\omega$ are ranked below the programs within the domain of  $\succ^R_\omega$ by $\succ_\omega$. Hence, $\succ^{P_s}_\omega$ and $\succ^{P_u}_\omega$ are compatible with each other.

\section{Proofs and Additional Results for Section \ref{sec:counterfactual}}

\subsection{Proof of Theorem \ref{thm:linear_charact_stable}}
Our proof builds on Lemma 1 in \cite{Baiou2000}, which is stated in the following using our notations:
\begin{lemma}[Lemma 1 in \cite{Baiou2000}]\label{lem:baiou_lemma}
A matching matrix $d$ is a stable matching given preferences $(\succ_\omega: \omega\in \Omega_N)$, priority scores $(\tilde{S}_\omega: \omega \in \Omega_N)$ and capacities $\tilde{q}^N$ if and only if $d$ satisfies the following condition:
    \begin{equation}\label{eq:simple_stable_matching_known_prefernece}
        \forall \omega \in \Omega_N,\; \forall j\in \mathcal{J}_0,\quad 
        \tilde{q}^N_j d_{\omega j} 
        + \tilde{q}^N_j \sum_{l:\ l\succ_{\omega} j} d_{\omega l} 
        + \sum_{\omega':\,\tilde{S}_{\omega' j}>\tilde{S}_{\omega j}} d_{\omega' j}
        \geq \tilde{q}^N_j.
    \end{equation}
\end{lemma}
Let $d$ be an arbitary matching in $\mathcal{D}(\mathcal{C}^N, \tilde{q}^N, \tilde{S})$. We want to prove that it satisfies the inequalities in \eqref{eq:simple_stable_matching_general}. 

Since $d\in \mathcal{D}(\mathcal{C}^N, \tilde{q}^N, \tilde{S})$, by the definition of $\mathcal{D}(\mathcal{C}^N, \tilde{q}^N, \tilde{S})$, there exist preferences $(\succ_\omega: \omega\in \Omega)$ such that $\succ_\omega \in \mathcal{C}_\omega$ and $d$ is a stable matching given $(\succ_\omega: \omega\in \Omega_N)$, $(\tilde{S}_\omega: \omega \in \Omega_N)$ and $\tilde{q}^N$. By Lemma 
\ref{lem:baiou_lemma}, $d$ satisfies the inequalities in \eqref{eq:simple_stable_matching_known_prefernece}. Because $\succ_\omega \in \mathcal{C}_\omega$, we know $l \succ_\omega j$ would imply that there exists some $\tilde{\succ}_\omega \in \mathcal{C}_\omega$ such that there is $l \tilde{\succ}_\omega j$. As a result, 
\begin{equation*}
\tilde{q}^N_j \sum_{l:\ l\succ_{\omega} j} d_{\omega l} \le \tilde{q}^N_j \sum_{l:\,\exists\,\succ_{\omega}\in\mathcal{C}_{\omega}\text{ s.t. }l\succ_{\omega} j} d_{\omega l}. 
\end{equation*}
Hence, $d$ should also satisfy conditions in \eqref{eq:simple_stable_matching_general}. 

\subsection{Proof of Theorem \ref{thm:iff_stable_matching}}
First of all, note that, for any $j$ and $l$ in $\mathcal{J}_0$, there exists some $\succ_\omega$ compatible with $\succ^P_\omega$ such that $l \succ_\omega j$, if and only if there is no $j \succ^P_\omega l$. As a result, conditions in \eqref{eq:simple_stable_matching_general} is equivalent to those in \eqref{eq:simple_stable_matching_partial_order}. Thus, we only need to prove the \emph{if} part.

Let $d$ be a matching matrix that satisfies \eqref{eq:simple_stable_matching_partial_order}. We are going to show $d\in \mathcal{D}(\mathcal{C}^N, \tilde{q}^N, \tilde{S})$ by constructing a total order $\succ_\omega$ for each student $\omega$ such that $\succ_\omega \in \mathcal{C}_\omega$ and $d$ is a stable matching given $(\succ_\omega: \omega\in \Omega_N)$, $(\tilde{S}_\omega: \omega \in \Omega_N)$ and $\tilde{q}^N$.

To construct this total order $\succ_\omega$ for each student $\omega$, define $j_\omega$ as the program to which student $\omega$ is matched in $d$, i.e. $d_{\omega j} = 1$ for $j = j_\omega$. Define $\mathcal{J}^+_\omega$ and $\mathcal{J}^-_\omega$ as follows:
\begin{equation*}
	\mathcal{J}^+_\omega \coloneqq \{j \in \mathcal{J}_0: j \succ^P_\omega j_\omega \}\text{ and }\mathcal{J}^-_\omega \coloneqq \{j \in \mathcal{J}_0: j \neq j_\omega \text{ and } j\notin \mathcal{J}^+_\omega \}.
\end{equation*}
Let $\succ^{+}_\omega$ be a total order on $\mathcal{J}^+_\omega$ that is compatible with $\succ^P_\omega$ restricted to $\mathcal{J}^+_\omega$. Similarly, let $\succ^{-}_\omega$ be a total order on $\mathcal{J}^-_\omega$ that is compatible with $\succ^P_\omega$ restricted to $\mathcal{J}^-_\omega$. Such $\succ^{+}_\omega$ and $\succ^{-}_\omega$ always exist.

Then, construct a total order $\succ_\omega$ as follows: for any $j$ and $j'$ in $\mathcal{J}_0$ with $j\neq j'$, set $j \succ_\omega j'$ if and only if one of the following conditions hold:
\begin{itemize}
\item $j\in \mathcal{J}^+_\omega$ and $j'\notin \mathcal{J}^+_\omega$
\item $j = j_\omega$, $j' \in \mathcal{J}^-_\omega$
\item $j\in \mathcal{J}^+_\omega$, $j' \in \mathcal{J}^+_\omega$ and $j \succ^{+}_\omega j'$
\item $j\in \mathcal{J}^-_\omega$, $j'\in \mathcal{J}^-_\omega$ and $j \succ^-_\omega j'$.
\end{itemize}
This $\succ_\omega$ ranks all programs within $\mathcal{J}^+_\omega$ above $j_\omega$ which in turn is ranked higher than all the programs in $\mathcal{J}^-_\omega$. Among programs within $\mathcal{J}^+_\omega$, it ranks them by $\succ^+_\omega$. Among programs within $\mathcal{J}^-_\omega$, it ranks them by $\succ^-_\omega$. Note that, by construction, this $\succ_\omega$ is compatible with $\succ^P_\omega$. Moreover, $\mathcal{J}^+_\omega$ is the set of programs that student $\omega$ prefer over her matching $j_\omega$, i.e., $\mathcal{J}^+_\omega = \{j\in \mathcal{J}_0: j \succ_\omega j_\omega\}$.

Next, We are going to show that $d$ is a stable matching given the total orders $(\succ_\omega: \omega\in \Omega_N)$ constructed above as prefernces. Fix an arbitrary $\omega$. For program $j = j_{\omega}$, inequality \eqref{eq:simple_stable_matching_known_prefernece} holds because $d_{\omega j} = 1$ in this case. For program $j\in \mathcal{J}^-_\omega$, inequality \eqref{eq:simple_stable_matching_known_prefernece} holds because $j_\omega \succ_\omega j$ and $d_{\omega j_\omega} = 1$. For program $j\in \mathcal{J}^+_\omega$, we have the following results hold:
\begin{itemize}
	\item $d_{\omega j} = 0$ because $j \neq j_\omega$;
	\item  $d_{\omega l} = 0$ for all $l$ with $l\succ_\omega j$, because  $\{l: l\succ_\omega j\} \subseteq \mathcal{J}^+_\omega$ and $j_\omega \notin \mathcal{J}^+_\omega$;
	\item for any $l$ in  $\{l: l\neq j$ and no $j \succ^P_\omega l\}$, we have $d_{\omega l} = 0$. This is because $j_\omega \notin $  $\{l: l\neq j$ and no $j \succ^P_\omega l\}$, since there is $j \succ^P_\omega j_\omega$.
\end{itemize}
Combining all the above results implies that 
\begin{equation*}
\tilde{q}^N_j d_{\omega j} + \tilde{q}^N_j \sum_{l:\ l\succ_{\omega} j} d_{\omega l} = 
    \tilde{q}^N_j d_{\omega j} + \tilde{q}^N_j \sum_{l:\, l\neq j\text{ and no }j\succ^P_\omega l} d_{\omega l} = 0.
\end{equation*}
As a result, the left-hand side of \eqref{eq:simple_stable_matching_partial_order} is equal to the left-hand side of \eqref{eq:simple_stable_matching_known_prefernece} for this given $(j, \omega)$ pair. Because $d$ satisfies \eqref{eq:simple_stable_matching_partial_order} for this $(j, \omega)$ pair, it must also satisfy \eqref{eq:simple_stable_matching_known_prefernece} for this $(j, \omega)$ pair.

We have verified that $d$ satisfies  \eqref{eq:simple_stable_matching_known_prefernece} for all programs and for an arbitary student. This shows $d$ is a stable matching given the total orders $(\succ_\omega: \omega\in \Omega_N)$. As a result, $d\in \mathcal{D}(\mathcal{C}^N, \tilde{q}^N, \tilde{S})$.

\subsection{Proof of Theorem \ref{thm:DA_convergence_properties}}\label{proof:DA_program_proposing}

Fix an arbitrary $j\in \mathcal{J}_0$. Fix an arbitrary iteration $k \ge 0$. If $\sum_\omega \indicator[S_{\omega j} \geq c^{(k)}_{j},\ d^{(k)}_{\omega j} = 1] + \sum_\omega \indicator[S_{\omega j} < c^{(k)}_{j}] > q^N_j$, $c^{(k+1)}_j \le c^{(k)}_j$ by construction. If $\sum_\omega \indicator[S_{\omega j} \geq c^{(k)}_{j},\ d^{(k)}_{\omega j} = 1] + \sum_\omega \indicator[S_{\omega j} < c^{(k)}_{j}] \le q^N_j$, we have $c^{(k+1)}_j = 0$, which also imply $c^{(k+1)}_j \le c^{(k)}_j$. Thus, for any $k \ge 0$, we always have $c^{(k+1)}_j \le c^{(k)}_j$. Since priority scores are bounded between $0$ and $1$, $(c^{(k)}_j)_k$ is also bounded between $0$ and $1$. Therefore, as a bounded nonincreasing sequence, $(c^{(k)}_j)_k$ would converge to a limit, denoted as $\overline{c}_j$. Moreover, because $c^{(k)}_j$ in nonincreasing in $k$, and because $c^{(k)}_j$ can only take values within $\{S_{j,\omega}: \omega\in \Omega_N\}\cup \{1, 0\}$, $c^{(k)} = (c^{(k)}_j: j\in \mathcal{J}_0)$ would converge within at most $|\mathcal{J}_0| \times (N + 2)$ iterations. This completes the proof of the first part of the theorem.

Fix an arbitrary stable matching and its associated equilibirum cutoff vector $c$. We want to prove that, for each $k\ge 0$ and each $j\in \mathcal{J}_0$, $c^{(k)}_j\ge c_j$. We prove the this result by induction:
\begin{itemize}
	\item For each $j$, $c^{(k)}_j \ge c_j$ trivially for $k = 0$ because $c_j^{(0)} = 1$ and priority scores are bounded within $[0, 1]$. 
	\item Suppose that, for each $j\in \mathcal{J}_0$, $c^{(k)}_j \ge c_j$ for $k$. We want to prove that for each $j\in \mathcal{J}_0$, $c^{(k + 1)}_j \ge c_j$. Consider the following two scenarios:
	\begin{itemize}
		\item for $j$ with $c_j = 0$, we have $c^{(k + 1)}_j \ge c_j$ trivially, because priority scores are distributed within $[0, 1]$.
		\item for $j$ with $c_j > 0$, the fact $c$ is the cutoff associated to a stable matching implies the following condition:
			\begin{equation}\label{eq:q39nddas9}
				\sum_\omega \indicator(\optchoice{\succ^Q_\omega}{F(S_\omega, c)} = j) = q^N_j
			\end{equation}
			where $F(S_\omega, c)$ is the feasible set for student $\omega$ in the fixed stable matching, and $\indicator(\optchoice{\succ^Q_\omega}{F(S_\omega, c)} = j)$ equals $1$ if and only if student $\omega$ is matched to program $j$ in this fixed stable matching. Because $\optchoice{\succ^Q_\omega}{F(S_\omega, c)} = j$ implies $S_{\omega j}\ge c_j$, we can rewrite \eqref{eq:q39nddas9} as 
			\begin{equation*}
			\sum_\omega \indicator(S_{\omega j}\ge c_j,\  \optchoice{\succ^Q_\omega}{F(S_\omega, c)} = j) = q^N_j
			\end{equation*}
which can be rewritten as
			\begin{multline}\label{eq:moaq2093n}
			\sum_\omega \indicator(S_{\omega j}\ge c^{(k)}_j,\  \optchoice{\succ^Q_\omega}{F(S_\omega, c)} = j) \\
			+ \sum_\omega \indicator(S_{\omega j}\in  [c_j, c^{(k)}_j), \optchoice{\succ^Q_\omega}{F(S_\omega, c)} = j) = q^N_j.
			\end{multline}
	\end{itemize}
	Here, $[c_j, c^{(k)}_j)$ is well-defined because we know $c_j \le c^{(k)}_j$. 

	Because $c_{j'} \le c^{(k)}_{j'}$ for any $j'\in \mathcal{J}_0$, we know for any $j'\neq j$, $j'\in F(S_\omega, c^{(k)})$ implies $j'\in F(S_\omega, c)$. Hence, for any $\omega$, we know
	\begin{equation}\label{eq:nafoiefa12}
\indicator[S_{\omega j} \geq c^{(k)}_{j},\ d^{(k)}_{\omega j} = 1] \ge \indicator(S_{\omega j}\ge c^{(k)}_j,\  \optchoice{\succ^Q_\omega}{F(S_\omega, c)} = j).
\end{equation}
Thus, \eqref{eq:moaq2093n} and \eqref{eq:nafoiefa12} imply 
	\begin{equation*}
		\sum_\omega \indicator(S_{\omega j}\ge c^{(k)}_j,\  d^{(k)}_{\omega j} = 1)  + \sum_\omega \indicator(S_{\omega j}\in  [c_j, c^{(k)}_j)) \ge q^N_j.
	\end{equation*}
	As a result, we know $c^{(k+1)}_j$ is either equal to $c^{(k)}_j$ or the $\tau$-th highest value in $\{S_{\omega j}: \omega\in \Omega_N, S_{\omega j}< c^{(k)}_j\}$ where $\tau = q^{N}_j - \sum_\omega \indicator(S_{\omega j} \geq c^{(k)}_{j},\ d^{(k)}_{\omega j} = 1)$. In the latter case, $c^{(k+1)}_j$ satisfies:
	\begin{equation}\label{eq:nfoq2i3en}
	\sum_\omega \indicator[S_{\omega j} \geq c^{(k)}_{j},\ d^{(k)}_{\omega j} = 1] + \sum_\omega \indicator[S_{\omega j} \in [c^{(k+1)}_j,  c^{(k)}_{j})] = q^N_j
	\end{equation}

Combining \eqref{eq:nafoiefa12}  and \eqref{eq:nfoq2i3en} implies 
\begin{equation*}
\sum_\omega \indicator[S_{\omega j} \in [c^{(k+1)}_j,  c^{(k)}_{j})] \le\sum_\omega \indicator(S_{\omega j}\in  [c_j, c^{(k)}_j)).
\end{equation*}
This implies $c_j^{(k+1)} \ge c_j$. 
\end{itemize}
We have established that $c_j^{(k)} \ge c_j$ for any $k\ge 0$. Because $\overline{c}_j$ is $\lim_{k\to \infty}c_j^{(k)}$, we conclude $\overline{c}_j \ge c_j$. This completes the proof.

\subsection{Proof of Theorem \ref{thm:upper_bound_simple}}
The Proof of Theorem \ref{thm:upper_bound_simple} is similar to that of Theorem \ref{thm:DA_convergence_properties}. 

Fix an arbitrary $j\in \mathcal{J}_0$. Following the same argument as the one we use to prove the first part of Theorem \ref{thm:DA_convergence_properties}, we can show that $\overline{c}^{(k)}_j$ is nonincreasing in $k$ and it converges to some limit denoted as $\overline{c}_j$ as $k$ increases to some finite number. This completes the Proof of the first part of the theorem.

The Proof of the second part of the theorem is also similar to that of Theorem \ref{thm:DA_convergence_properties}. Fix an arbitrary matching $d\in \mathcal{D}(\mathcal{C}^N, q^N, S)$ and its associated equilibirum cutoff vector $c$. We want to prove that, for each $k\ge 0$ and each $j\in \mathcal{J}_0$, $\overline{c}^{(k)}_j\ge c_j$. We prove the this result by induction:
\begin{itemize}
	\item For each $j$, $\overline{c}^{(k)}_j \ge c_j$ trivially for $k = 0$ because $\overline{c}_j^{(0)} = 1$ and priority scores are bounded within $[0, 1]$. 
	\item Suppose that, for each $j\in \mathcal{J}_0$, $\overline{c}^{(k)}_j \ge c_j$ for $k$. We want to prove that for each $j\in \mathcal{J}_0$, $\overline{c}^{(k + 1)}_j \ge c_j$. Consider the following two scenarios:
	\begin{itemize}
		\item for $j$ with $c_j = 0$, we have $\overline{c}^{(k + 1)}_j \ge c_j$ trivially, because priority scores are distributed within $[0, 1]$.
		\item for $j$ with $c_j > 0$, we must have the following condition hold for some $(\succ^Q_\omega:\omega\in \Omega_N)$ with $\succ^Q_\omega \in \mathcal{C}_\omega$ for each $\omega\in \Omega_N$:
			\begin{equation}\label{eq:aq39nddas9}
				\sum_\omega \indicator(\optchoice{\succ^Q_\omega}{F(S_\omega, c)} = j) = q^N_j
			\end{equation}
			where $F(S_\omega, c)$ is the feasible set for student $\omega$ in the fixed stable matching, and $\indicator(\optchoice{\succ^Q_\omega}{F(S_\omega, c)} = j)$ equals $1$ if and only if student $\omega$ is matched to program $j$ in this matching $d$. Because $\optchoice{\succ^Q_\omega}{F(S_\omega, c)} = j$ implies $S_{\omega j}\ge c_j$, we can rewrite \eqref{eq:aq39nddas9} as 
			\begin{equation*}
			\sum_\omega \indicator(S_{\omega j}\ge c_j,\  \optchoice{\succ^Q_\omega}{F(S_\omega, c)} = j) = q^N_j
			\end{equation*}
			which further implies 
			\begin{multline}\label{eq:23mmoaq2093n}
			\sum_\omega \indicator(S_{\omega j}\ge \overline{c}^{(k)}_j,\  \optchoice{\succ^Q_\omega}{F(S_\omega, c)} = j) \\
			+ \sum_\omega \indicator(S_{\omega j}\in  [c_j, \overline{c}^{(k)}_j), \optchoice{\succ^Q_\omega}{F(S_\omega, c)} = j) = q^N_j.
			\end{multline}
	\end{itemize}
	Here, $[c_j, \overline{c}^{(k)}_j)$ is well-defined because we know $c_j \le \overline{c}^{(k)}_j$. Because $c_{j'} \le \overline{c}^{(k)}_{j'}$ for any $j'\in \mathcal{J}_0$, we know for any $j'\neq j$, $j'\in F(S_\omega, \overline{c}^{(k)})$ implies $j'\in F(S_\omega, c)$. Hence, for any $\omega$, we know
	\begin{equation}\label{eq:anafoiefa12}
\indicator[S_{\omega j} \geq c^{(k)}_{j},\ d^{(k)}_{\omega j} = 1] \ge \indicator(S_{\omega j}\ge c^{(k)}_j,\  \optchoice{\succ^Q_\omega}{F(S_\omega, c)} = j).
\end{equation}
	Because of \eqref{eq:23mmoaq2093n} and \eqref{eq:anafoiefa12}, we must have 
	\begin{equation*}
		\sum_\omega \indicator[S_{\omega j} \geq \overline{c}^{(k)}_{j},\ d^{(k)}_{\omega j} = 1]	  + \sum_\omega \indicator(S_{\omega j}\in  [c_j, \overline{c}^{(k)}_j)) \ge q^N_j.
	\end{equation*}
	As a result, we know $\overline{c}^{(k+1)}_j$ is either equal to $\overline{c}^{(k)}_j$ or the $\tau$-th highest value in $\{S_{\omega j}: \omega\in \Omega_N, S_{\omega j}< \overline{c}^{(k)}_j\}$ where $\tau = q^{N}_j - \sum_\omega \indicator(S_{\omega j} \geq \overline{c}^{(k)}_{j},\ d^{(k)}_{\omega j} = 1)$. In the latter case, $\overline{c}^{(k+1)}_j$ satisfies:
	\begin{equation}\label{eq:anfoq2i3en}
		\sum_\omega \indicator[S_{\omega j} \geq \overline{c}^{(k)}_{j},\ d^{(k)}_{\omega j} = 1] + \sum_\omega \indicator[S_{\omega j} \in [\overline{c}^{(k+1)}_j,  \overline{c}^{(k)}_{j})] = q^N_j
	\end{equation}

Combining \eqref{eq:anafoiefa12} and \eqref{eq:anfoq2i3en} implies 
\begin{equation*}
	\sum_\omega \indicator[S_{\omega j} \in [\overline{c}^{(k+1)}_j,  \overline{c}^{(k)}_{j})] \le 
	\sum_\omega \indicator(S_{\omega j}\in  [c_j, \overline{c}^{(k)}_j)).
\end{equation*}
This implies $\overline{c}_j^{(k+1)} \ge c_j$. 
\end{itemize}
We have established that $\overline{c}_j^{(k)} \ge c_j$ for any $k\ge 0$. Because $\overline{c}_j$ is $\lim_{k\to \infty}\overline{c}_j^{(k)}$, we conclude $\overline{c}_j \ge c_j$. This completes the proof.


\section{Proofs and Additional Results for Section \ref{sec:more_revealed_preference}}

\subsection{Supplementary results for section \ref{sec:robust_undominated}}\label{sec:supp_robust_undominated}

\begin{lemma}\label{lem:partial_order_u1_u2}
For each $\omega$, the $\succ^{P_{u1}}_\omega$ and $\succ^{P_{u2}}_\omega$ are partial orders. 
\end{lemma}
\begin{proof}
Fix an arbitrary student $\omega$. Let us first show that $\succ^{P_{u1}}_\omega$ is a partial order. Since $\succ^{P_{u1}}_\omega$ is constructed as the transitive closure of $\succ^{u1}_\omega$, $\succ^{P_{u1}}_\omega$ is transitive by construction. Thus, to show $\succ^{P_{u1}}_\omega$ is a partial order, we only need to show that $\succ^{u1}_\omega$ is acyclic. 

Suppose, for the purpose of contradiction, there exists some integer $M$ and $j_1, ..., j_M$ with $j_{k} \succ^{u1}_\omega j_{k+1}$ for each $k = 1,...,M-1$ and $j_M \succ^{u1}_\omega j_1$. By the construction of $\succ^{u1}_\omega$, there is $j \succ^{u1}_\omega j'$ if and only if there exists some $B\in \mathcal{F}_\omega$ such that $j, j'\in B$, $j\neq j'$ and $j = \optchoice{\succ^R_\omega}{B}$. Therefore, there must exists $B_1, ..., B_M\in \mathcal{F}_\omega$ such that, for each $k = 1,...,M-1$, the following two results hold: (\emph{i}) both $j_k$ and $j_{k+1}$ are in $B_k$, and (\emph{ii}) $j_k = \optchoice{\succ^R_\omega}{B_k}$. By these two results, we know that, for each $k=1,...,M-1$, $j_k \in \text{domain}(\succ^R_\omega)$ and $j_k \succ^R_\omega j_{k+1}$. Thus, $j_1 \succ^R_\omega j_M$. On the other hand, $j_M \succ^{P_{u1}}_\omega j_1$ implies $j_M \succ^R_{\omega} j_1$. This leads to a contradiction, because $\succ^R_\omega$ is a partial order. As a result, $\succ^{u1}_\omega$ is acyclic.
This completes the proof for $\succ^{P_{u1}}_\omega$.

Next, we want to prove that $\succ^{P_{u2}}_\omega$ is a partial order. Similar as before, because $\succ^{P_{u2}}_\omega$ is the transitive closure of $\succ^{u2}_\omega$, we only need to prove $\succ^{u2}_\omega$ is acyclic. 

Suppose, for the purpose of contradiction, there exists some integer $M$ and $j_1, ..., j_M$ with $j_{k} \succ^{u2}_\omega j_{k+1}$ for each $k = 1,...,M-1$ and $j_M \succ^{u2}_\omega j_1$. By the construction of $\succ^{u2}_\omega$, the following two statements are true:
\begin{itemize}
\item there does not exist $j\in \text{domain}(\succ^{u2}_\omega)$ such that $0\succ^{u2}_\omega j$ and $j \succ^{u2}_\omega j'$ for some $j'\in \mathcal{J}_0$.
\item there is $j, j'\in \mathcal{J}_0$ with $j'\neq 0$ and $j\succ^{u2}_\omega j'$ if and only if there exist some $B\in \mathcal{F}_\omega$ such that $j, j'\in B$, $j\neq j'$ and $j = \optchoice{\succ^R_\omega}{B}$.
\end{itemize}
By the first statement, we know that, for each $k= 1,...,M$, $j_k\neq 0$. By the second statement, there must exists $B_1, ..., B_M\in \mathcal{F}_\omega$ such that, for each $k = 1,...,M-1$, the following two results hold: (\emph{i}) both $j_k$ and $j_{k+1}$ are in $B_k$, and (\emph{ii}) $j_k = \optchoice{\succ^R_\omega}{B_k}$. Therefore, for each $k=1,...,M-1$, $j_k \in \text{domain}(\succ^R_\omega)$ and $j_k \succ^R_\omega j_{k+1}$. Thus, $j_1 \succ^R_\omega j_M$. On the other hand, $j_M \succ^{P_{u2}}_\omega j_1$ implies $j_M \succ^R_{\omega} j_1$. This leads to a contradiction. As a result, $\succ^{u2}_\omega$ is acyclic. This completes the proof for $\succ^{P_{u2}}_\omega$.
\end{proof}

\subsection{Proof of proposition \ref{prop:robust_undominance}}

\paragraph{if part}
Suppose 
\begin{enumerate}
\item for student $\omega$ with $|\succ^R_\omega| < K$, $\succ^Q_\omega$ is compatible with $\succ^{P_{u1}}_\omega$.
\item for student $\omega$ with $|\succ^R_\omega| = K$, $\succ^Q_\omega$ is compatible with at least one of $\succ^{P_{u1}}_\omega$ and $\succ^{P_{u2}}_\omega$. 
\end{enumerate}
Construct $\mathcal{F}^*_\omega$ as follows:
\begin{itemize}
\item for student $\omega$ with $|\succ^R_\omega| < K$, let $\mathcal{F}^*_\omega = \mathcal{F}_\omega$.
\item for student $\omega$ with $|\succ^{R}_\omega| = K$, let $\mathcal{F}^*_\omega = \mathcal{F}_\omega$ if $\succ^{Q}_\omega$ is compatible with $\succ^{P_{u1}}_\omega$, and let $\mathcal{F}^*_\omega = \mathcal{F}_\omega\cup \{\{j, 0\}: j\in \text{domain}(\succ^R_\omega)\setminus \{0\}\}$ if otherwise.
\end{itemize}
We are going to show that Assumption \ref{assu:robust_undominance} holds for this $(\mathcal{F}^*_\omega: \omega\in \Omega)$. By construction, $\mathcal{F}_\omega \subseteq \mathcal{F}^*_\omega $ for every $\omega \in \Omega$. Thus, we only need to show $\succ^{R}_\omega$ is undominated given $\mathcal{F}^*_\omega$. Consider two scenarios:
\begin{itemize}
	\item Fix an arbitrary student $\omega$ whose $\succ^Q_\omega$ is compatible with $\succ^{P_{u1}}_\omega$. In this case, $\mathcal{F}^*_\omega = \mathcal{F}_\omega$. Because $\succ^Q_\omega$ is compatible with $\succ^{P_{u1}}_\omega$, $\succ^Q_\omega$  is also compatible with $\succ^{u1}_\omega$. Hence, for any $B\in \mathcal{F}^*_\omega = \mathcal{F}_\omega$, $\optchoice{\succ^R_\omega}{B} = \optchoice{\succ^Q_\omega}{B}$, i.e., student $\omega$ is always matched to her favorite program in $B$. In this case, there cannot exist an alternative $\succ^{R'}_\omega$ such that $\optchoice{\succ^{R'}_\omega}{B}\succ^Q_\omega \optchoice{\succ^R_\omega}{B}$. Hence, $\succ^{R}_\omega$ is undominated given this $\mathcal{F}^*_\omega$.
	\item Fix an arbitrary student $\omega$ whose $\succ^Q_\omega$ is not compatible with $\succ^{P_{u1}}_\omega$. In this case, we must have $|\succ^R_\omega| = K$ and $\succ^Q_\omega$ is compatible with $\succ^{P_{u2}}_\omega$. Let $\succ^{R'}_\omega$ be an arbitrary alternative ROL. Consider the following two cases:
	\begin{itemize}
		\item Suppose $\text{domain}(\succ^{R'}_\omega)\neq \text{domain}(\succ^{R}_\omega)$. Because $|\succ^R_\omega| = K$, $\text{domain}(\succ^{R'}_\omega)\neq \text{domain}(\succ^{R}_\omega)$ implies that there exist some $j^*\in \text{domain}(\succ^{R}_\omega) \setminus \text{domain}(\succ^{R'}_\omega)$. Then, consider $B = \{j^*, 0\} \in \mathcal{F}^*_\omega$. For this $B$, $\optchoice{\succ^R_\omega}{B} = j^* \succ^Q_\omega 0 = \optchoice{\succ^{R'}_\omega}{B}$. Hence, $\succ^R_\omega$ cannot be dominated by $\succ^{R'}_\omega$ given $\mathcal{F}^*_\omega$.
		\item Suppose $\text{domain}(\succ^{R'}_\omega) =  \text{domain}(\succ^{R}_\omega)$. Because $\succ^Q_\omega$ is compatible with $\succ^{u2}_\omega$ in this case, we know for any $B\in \mathcal{F}_\omega$ and any $j'\in B\setminus\optchoice{\succ^R_\omega}{B}$, $\optchoice{\succ^R_\omega}{B} \succ^Q_\omega j'$. Therefore, for any $B\in \mathcal{F}_\omega$, either $\optchoice{\succ^R_\omega}{B} \succ^Q_\omega \optchoice{\succ^{R'}_\omega}{B}$ or $\optchoice{\succ^R_\omega}{B} = \optchoice{\succ^{R'}_\omega}{B} $. 
			For any $B\in \mathcal{F}^*_\omega \setminus \mathcal{F}_\omega$, we have $\optchoice{\succ^R_\omega}{B} = \optchoice{\succ^Q_\omega}{B}$. As a result, $\succ^R_\omega$ cannot be dominated by $\succ^{R'}_\omega$ given $\mathcal{F}^*_\omega$.
	\end{itemize}
Thus, $\succ^R_\omega$ cannot be dominated by $\succ^{R'}_\omega$ given $\mathcal{F}^*_\omega$. This completes the proof for $\succ^R_\omega$ being undominated given $\mathcal{F}^*_\omega$ in this case.
\end{itemize}
This completes the if part.

\paragraph{only if part}
Suppose Assumption \ref{assu:robust_undominance} holds for some $(\mathcal{F}^*_\omega: \omega\in \Omega)$. Consider two scenarios:

\begin{itemize}
\item Fix an arbitrary student $\omega$ with $|\succ^R_\omega| < K$. We want to show that $\succ^Q_\omega$ is compatible with $\succ^{P_{u1}}_\omega$. By the construction of $\succ^{P_{u1}}_\omega$, we only need to show that $\succ^Q_\omega$ is compatible with $\succ^{u1}_\omega$, i.e., for any $j, j'\in \mathcal{J}_0$ with $j \succ^{u1}_\omega j'$, we have $j \succ^{Q}_\omega j'$.

	Suppose, for the purpose of contradiction, there exists some $j_1$, $j_2 \in \mathcal{J}_0$ such that $j_1 \succ^{u1}_\omega j_2$ and $j_2 \succ^Q_\omega j_1$. By the construction of $\succ^{u1}_\omega$, there must exist some $B^*\in \mathcal{F}_\omega \subseteq \mathcal{F}^*_\omega$ such that $j_1, j_2 \in B^*$ and $j_1 = \optchoice{\succ^R_\omega}{B^*}$. Consider two scenarios:
\begin{itemize}
	\item Suppose $j_1 = 0$. Then, construct another ROL $\succ^{R'}_\omega$ the same as $\succ^{R}_\omega$ except appending $j_2$ just above the outside option in the ROL. That is, (i) $j \succ^{R'}_\omega j'$ for any $j,j'$ with $j \succ^R_\omega j'$, (ii) $j_2 \succ^{R'}_\omega 0$, and (iii) $j \succ^{R'}_\omega j_2$ for any $j\in \text{domain}(\succ^R_\omega)\setminus\{0\}$. By construction, $\optchoice{\succ^{R'}_\omega}{B^*} = j_2 \succ^Q_\omega 0 = \optchoice{\succ^R_\omega}{B^*}$. Moreover, for any $B\in \mathcal{F}^*_\omega$, either $\optchoice{\succ^{R'}_\omega}{B^*} = \optchoice{\succ^R_\omega}{B^*}$ or $\optchoice{\succ^{R'}_\omega}{B^*} \neq \optchoice{\succ^R_\omega}{B^*}$. If $\optchoice{\succ^{R'}_\omega}{B^*} \neq \optchoice{\succ^R_\omega}{B^*}$, we must have $\optchoice{\succ^{R'}_\omega}{B} = j_2 \succ^Q_\omega 0 = \optchoice{\succ^R_\omega}{B}$ by the construction of $\succ^{R'}_\omega$. As a result, $\succ^R_\omega$ is dominated by $\succ^{R'}_\omega$ given $\mathcal{F}^*_\omega$, contradicting to Assumption \ref{assu:robust_undominance}.

	\item Suppose $j_1 \neq 0$. Then, construct another ROL $\succ^{R'}_\omega$ with the same domain as $\succ^R_\omega$ but ordering them according to $\succ^Q_\omega$. That is, $j \succ^{R'}_\omega j'$ if and only if $j,j'\in \text{domain}(\succ^R_\omega)$ and $j \succ^Q_\omega j'$. Then, by construction, either $\optchoice{\succ^{R'}_\omega}{B^*} = j_2$ or $\optchoice{\succ^{R'}_\omega}{B^*} \succ^Q_\omega j_2$. In both cases, we have $\optchoice{\succ^{R'}_\omega}{B^*} \succ^Q_\omega  \optchoice{\succ^R_\omega}{B^*}$, because $\optchoice{\succ^R_\omega}{B^*} = j_1$. Moreover, for any $B\in \mathcal{F}^*_\omega$, either $\optchoice{\succ^{R'}_\omega}{B^*} = \optchoice{\succ^R_\omega}{B^*}$ or $\optchoice{\succ^{R'}_\omega}{B^*} \neq \optchoice{\succ^R_\omega}{B^*}$. If $\optchoice{\succ^{R'}_\omega}{B^*} \neq \optchoice{\succ^R_\omega}{B^*}$, we must have $\optchoice{\succ^{R'}_\omega}{B} \succ^Q_\omega \optchoice{\succ^R_\omega}{B}$ because $\succ^{R'}_\omega$ is compatible with $\succ^Q_\omega$ by construction. As a result, $\succ^R_\omega$ is dominated given $\mathcal{F}^*_\omega$, contradicting to Assumption \ref{assu:robust_undominance}.
\end{itemize}
Thus, we always have a contradiction for student $\omega$ with $|\succ^R_\omega| < K$.

\item Fix an arbitrary student $\omega$ with $|\succ^R_\omega| = K$. We want to show that $\succ^Q_\omega$ is compatible with at least one of $\succ^{P_{u1}}_\omega$ and $\succ^{P_{u2}}_\omega$. Define $\mathcal{A} \coloneqq \{\optchoice{\succ^Q_\omega}{B}: B\in \mathcal{F}^*_\omega\} \cup \{0\}$. Consider two scenarios:
	\begin{itemize}
		\item Suppose the number of elements in $\mathcal{A}$, denoted as $|\mathcal{A}|$, is less than or equal to $K$. In this case, we want to show $\succ^Q_\omega$ is compatible with $\succ^{P_{u1}}_\omega$. 
Suppose, for the purpose of contradiction, that $\succ^Q_\omega$ is not compatible with $\succ^{P_{u1}}_\omega$. 
Then, $\succ^Q_\omega$ is not compatible with $\succ^{u1}_\omega$. 
As a result, there must exist some $B^*\in \mathcal{F}_\omega\subseteq \mathcal{F}^*_\omega$ and some $j_1, j_2\in B$ such that $j_1 = \optchoice{\succ^R_\omega}{B^*}$, $j_1 \succ^{u1}_\omega j_2$ and $j_2 \succ^Q_\omega j_1$. 
Construct an alternative $\succ^{R'}_\omega$ where $j \succ^{R'}_\omega j'$ if and only if $j \succ^Q_\omega j'$ and $j, j' \in \mathcal{A}$. By construction, for any $B\in \mathcal{F}^*_\omega$, $\optchoice{\succ^{R'}_\omega}{B} = \optchoice{\succ^Q_\omega}{B}$. Hence, for any $B\in \mathcal{F}^*_\omega$, either $\optchoice{\succ^{R'}_\omega}{B} = \optchoice{\succ^R_\omega}{B}$ or $\optchoice{\succ^{R'}_\omega}{B} \succ^Q_\omega \optchoice{\succ^R_\omega}{B}$. Moreover, either $\optchoice{\succ^{R'}_\omega}{B^*} = j_2$ or $\optchoice{\succ^{R'}_\omega}{B^*} \succ^Q_\omega j_2$. In both cases, $\optchoice{\succ^{R'}_\omega}{B^*} \succ^Q_\omega j_1 = \optchoice{\succ^R_\omega}{B}$. As a result, $\succ^R_\omega$ is dominated by $\succ^{R'}_\omega$ given $\mathcal{F}^*_\omega$, contradicting to Assumption \ref{assu:robust_undominance}.

\item Suppose $|\mathcal{A}|$ is greater than $K$. In this case, we want to show $\succ^Q_\omega$ is compatible with $\succ^{P_{u2}}_\omega$. Suppose, for the purpose of contradiction, that $\succ^Q_\omega$ is not compatible with $\succ^{P_{u2}}_\omega$. Then, $\succ^Q_\omega$ is not compatible with $\succ^{u2}_\omega$. Then, one of the following two cases is true:
 \begin{itemize}
 \item there exists some $j^\dagger\in \text{domain}(\succ^R_\omega)$ such that $0 \succ^Q_\omega j^\dagger$. Since $|\mathcal{A}| > |\succ^R_\omega| = K$, there exists some $B^*\in \mathcal{F}^*_\omega$ such that $j^* \coloneqq \optchoice{\succ^Q_\omega}{B^*} \notin \text{domain}(\succ^R_\omega)$. In particular, $j^* \succ^Q_\omega 0$. 
	 Construct an alternative $\succ^{R'}_\omega$ as follows: let $\succ^{R'}_\omega$'s domain be $( \text{domain}(\succ^R_\omega)\setminus \{j^\dagger\} ) \cup \{j^*\}$, and for any $j, j'$ in its domain, $j \succ^{R'}_\omega j'$ if and only if $j \succ^Q_\omega j'$. 
	Then, for any $B\in \mathcal{F}^*_\omega$, either $\optchoice{\succ^{R'}_\omega}{B} = \optchoice{\succ^R_\omega}{B}$ or $\optchoice{\succ^{R'}_\omega}{B} \succ^Q_\omega \optchoice{\succ^R_\omega}{B}$. 
	Moreover, $\optchoice{\succ^{R'}_\omega}{B^*} = j^* \succ^Q_\omega \optchoice{\succ^{R}_\omega}{B^*}$. As a result, $\succ^R_\omega$ is dominated by $\succ^{R'}_\omega$ given $\mathcal{F}^*_\omega$, contradicting to Assumption \ref{assu:robust_undominance}.

\item there exists some $B^* \in \mathcal{F}_\omega \subseteq \mathcal{F}^*_\omega$ and some $j_1, j_2 \in B\cap\text{domain}(\succ^R_\omega)$ such that $j_2 \succ^Q_\omega j_1 = \optchoice{\succ^R_\omega}{B^*}$. Construct an alternative ROL $\succ^{R'}_\omega$ by setting its domain to be the same as $\text{domain}(\succ^R_\omega)$ and ordering the programs according to $\succ^Q_\omega$.
That is, $j \succ^{R'}_\omega j'$ if and only if $j,j'\in \text{domain}(\succ^R_\omega)$ and $j \succ^Q_\omega j'$. Because $\succ^{R'}_\omega$ and $\succ^R_\omega$ have the same domain and $\succ^{R'}_\omega$ is compatible with $\succ^Q_\omega$, for any $B\in \mathcal{F}^*_\omega$, either $\optchoice{\succ^{R'}_\omega}{B} = \optchoice{\succ^R_\omega}{B}$ or $\optchoice{\succ^{R'}_\omega}{B} \succ^Q_\omega \optchoice{\succ^R_\omega}{B}$. 
	Moreover, $\optchoice{\succ^{R'}_\omega}{B^*} \succ^Q_\omega j_1 = \optchoice{\succ^{R}_\omega}{B^*}$. As a result, $\succ^R_\omega$ is dominated by $\succ^{R'}_\omega$ given $\mathcal{F}^*_\omega$, contradicting to Assumption \ref{assu:robust_undominance}.
 \end{itemize}
\end{itemize}
	
Thus, we always have shown that $\succ^Q_\omega$ should be compatible with at least one of $\succ^{P_{u1}}_\omega$ and $\succ^{P_{u2}}_\omega$ for students with $|\succ^R_\omega| = K$.
\end{itemize}

\subsection{Supplementary results for section \ref{sec:selective}}\label{sec:lemma_sel}
\begin{lemma}\label{lem:select}
Suppose $\gg$ and $(\mathcal{F}_\omega:\omega\in \Omega)$ are compatible with each other. Moreover, suppose $\succ^{R}_\omega$ and $\gg$ are compatible with each other. Then, $\succ^{P_{sel}}_\omega$ is a partial order.
\end{lemma}

\begin{proof}
First of all, note that, by the assumptions imposed by the lemma and by the construction of $\succ^{sel}_\omega$, the following statements are true:
\begin{itemize}
	\item there does not exist $j\in \mathcal{J}_0$ such that $0 \succ^{sel}_\omega j$, because for each $B\subseteq \mathcal{J}_0$ with $0\in B$, $\optchoice{\succ^R_\omega}{B} = 0$ would imply $B\cap \text{domain}(\succ^R_\omega) = \{0\}$.
	\item for any $j, j'\in \mathcal{J}_0$ with $j \succ^{sel}_\omega j'$, we have $j\in \text{domain}({\succ^R_\omega})\setminus\{0\}$.
		
	\item for any $j, j'\in \text{domain}(\succ^R_\omega)$ with $j'\neq 0$, there is $j \succ^{sel}_\omega j'$ if and only if  and one of the following two conditions are true: (\emph{a}) $j\gg j'$ and $j \succ^R_\omega j'$; (\emph{b})  $j = \optchoice{\succ^R_\omega}{B}$ for some $B\in \mathcal{F}_\omega$ with $j'\in B$. In both cases, $j \succ^R_\omega j'$. 
\end{itemize}

Because $\succ^{P_{sel}}_\omega$ is the transitive closure of $\succ^{sel}_\omega$, we only need to show that $\succ^{sel}_\omega$ is acyclic. Suppose, for the purpose of contraiction, there exists an integer $M$ and $j_1, ..., j_M\in \mathcal{J}_0$ such that $j_k \succ^{sel}_\omega j_{k+1}$ for each $k = 1, ..., M-1$ and $j_M \succ^{sel}_\omega j_1$. By the three results noted at the beginning of the proof, we know that 
\begin{itemize}
	\item $j_k \in \text{domain}(\succ^{R}_i)\setminus\{0\}$ for each $k=1,...,M$
	\item for each $k=1,...,M-1$, $j_k \succ^R_\omega j_{k+1}$.
	\item $j_M \succ^R_\omega j_1$.
\end{itemize}
Because $\succ^{R}_\omega$ is a partial order and hence acyclic, the last two points leads to a contradiction. This completes the proof.
\end{proof}

\subsection{Proof of proposition \ref{prop:selectiveness}}
\paragraph{if part}
Suppose the following conditions are satisfied
\begin{enumerate}
\item for student $\omega$ with $|\succ^R_\omega| < K$ or $\succ^R_\omega$ not compatile with $\gg$, $\succ^Q_\omega$ is compatible with $\succ^{P_{u1}}_\omega$
\item for student $\omega$ with $|\succ^R_\omega| = K$ and $\succ^R_\omega$ compatile with $\gg$, $\succ^Q_\omega$ is compatible with at least one of  $\succ^{P_{u1}}_\omega$ and $\succ^{P_{sel}}_\omega$.
\end{enumerate}
For each student $\omega$, construct $\mathcal{F}^*_\omega$ as follows:
\begin{itemize}
\item for student $\omega$ whose $\succ^Q_\omega$ is compatible with $\succ^{P_{u1}}_\omega$, let $\mathcal{F}^*_\omega = \mathcal{F}_\omega$.
\item for student $\omega$ whose $\succ^Q_\omega$ is not compatible with $\succ^{P_{u1}}_\omega$, let $\mathcal{F}^*_\omega = \mathcal{F}_\omega\cup \{\{j\}\cup\{j'\in \mathcal{J}_0: j \gg j'\}: j\in \text{domain}(\succ^R_\omega)\setminus \{0\}\}$.

\end{itemize}
We are going to show that Assumption \ref{assu:robust_undominance} holds for this $(\mathcal{F}^*_\omega: \omega\in \Omega)$. By construction, $(\mathcal{F}^*_\omega: \omega\in \Omega)$ satisfies Assumption \ref{assu:selectiveness} because $\mathcal{F}_\omega$ and $\gg$ are compatible. Also, by construction, $\mathcal{F}_\omega \subseteq \mathcal{F}^*_\omega $ for every $\omega \in \Omega$. Thus, we only need to show that, for every student $\omega$, $\succ^{R}_\omega$ is undominated given $\mathcal{F}^*_\omega$. To show this, consider two scenarios:

\begin{itemize}
	\item Fix an arbitrary student $\omega$ whose $\succ^Q_\omega$ is compatible with $\succ^{P_{u1}}_\omega$. In this case, $\mathcal{F}^*_\omega = \mathcal{F}_\omega$. Because $\succ^Q_\omega$ is compatible with $\succ^{P_{u1}}_\omega$, $\succ^Q_\omega$  is also compatible with $\succ^{u1}_\omega$. Hence, for any $B\in \mathcal{F}^*_\omega = \mathcal{F}_\omega$, $\optchoice{\succ^R_\omega}{B} = \optchoice{\succ^Q_\omega}{B}$, i.e., student $\omega$ is always matched to her favorite program in $B$. In this case, there cannot exist an alternative $\succ^{R'}_\omega$ such that $\optchoice{\succ^{R'}_\omega}{B}\succ^Q_\omega \optchoice{\succ^R_\omega}{B}$. Hence, $\succ^{R}_\omega$ is undomianted given this $\mathcal{F}^*_\omega$.

	\item Fix an arbitrary student $\omega$ whose $\succ^Q_\omega$ is not compatible with $\succ^{P_{u1}}_\omega$. In this case, we must have $|\succ^R_\omega| = K$, $\succ^R_\omega$ compatible with $\gg$ and $\succ^Q_\omega$ is compatible with $\succ^{P_{sel}}_\omega$. Let $\succ^{R'}_\omega$ be an arbitrary alternative ROL. Consider the following two cases:
	\begin{itemize}
		\item Suppose $\text{domain}(\succ^{R'}_\omega)\neq \text{domain}(\succ^{R}_\omega)$. Because $|\succ^R_\omega| = K$, $\text{domain}(\succ^{R'}_\omega)\neq \text{domain}(\succ^{R}_\omega)$ implies that there exist some $j^*\in \text{domain}(\succ^{R}_\omega) \setminus \text{domain}(\succ^{R'}_\omega)$. Then, consider $B = \{j^*\}\cup\{j'\in \mathcal{J}_0: j^* \gg j'\}\in \mathcal{F}^*_\omega$. For this $B$, because $j^* \in \text{domain}(\succ^R_\omega)$ and $\succ^R_\omega$ is compatible with $\gg$, we know $\optchoice{\succ^R_\omega}{B} = j^*$. On the other hand, because $j^*\notin \text{domain}(\succ^{R'}_\omega)$ and $\succ^Q_\omega$ compatible with $\succ^{sel}_\omega$, we must have $j^* \succ^Q_\omega \optchoice{\succ^{R'}_\omega}{B}$. Hence, $\succ^R_\omega$ cannot be dominated by $\succ^{R'}_\omega$ given $\mathcal{F}^*_\omega$.
		\item Suppose $\text{domain}(\succ^{R'}_\omega) =  \text{domain}(\succ^{R}_\omega)$. Because $\succ^Q_\omega$ is compatible with $\succ^{sel}_\omega$ in this case, we know for any $B\in \mathcal{F}_\omega$ and any $j'\in B\setminus\optchoice{\succ^R_\omega}{B}$, $\optchoice{\succ^R_\omega}{B} \succ^Q_\omega j'$. Therefore, for any $B\in \mathcal{F}_\omega$, either $\optchoice{\succ^R_\omega}{B} \succ^Q_\omega \optchoice{\succ^{R'}_\omega}{B}$ or $\optchoice{\succ^R_\omega}{B} = \optchoice{\succ^{R'}_\omega}{B} $. 
Moreover, for any $B\in \mathcal{F}^*_\omega \setminus \mathcal{F}_\omega$, we have $\optchoice{\succ^R_\omega}{B} = \optchoice{\succ^Q_\omega}{B}$, because  $\succ^Q_\omega$ is compatible with $\succ^{sel}_\omega$. As a result, $\succ^R_\omega$ cannot be dominated by $\succ^{R'}_\omega$ given $\mathcal{F}^*_\omega$.
	\end{itemize}
Thus, $\succ^R_\omega$ cannot be dominated by $\succ^{R'}_\omega$ given $\mathcal{F}^*_\omega$. This completes the proof for $\succ^R_\omega$ being undominated given $\mathcal{F}^*_\omega$ in this case.
\end{itemize}
This completes the if part.

\paragraph{only if part}
Suppose Assumption \ref{assu:robust_undominance} and Assumption \ref{assu:selectiveness} hold for some $(\mathcal{F}^*_\omega: \omega \in \Omega)$. Consider two scenarios:
\begin{itemize}
\item Fix an arbitrary $\omega$ with $|\succ^{R}_\omega| < K$ or $\succ^R_\omega$ not compatible with $\gg$. We want to show $\succ^Q_\omega$ must be compatible with $\succ^{P_{u1}}_\omega$, which is equivalent to show $\succ^Q_\omega$ is compatible with $\succ^{u1}_\omega$.

Suppose, for the purpose of contradiction, there exists some $j_1, j_2 \in \mathcal{J}_0$ such that $j_1 \succ^{u1}_\omega j_2$ and $j_2 \succ^{Q}_\omega j_1$. By the construction of $\succ^{u1}_\omega$, there must exist some $B^*\in \mathcal{F}_\omega \subseteq \mathcal{F}^*_\omega$ such that $j_1, j_2 \in B^*$ and $j_1 = \optchoice{\succ^R_\omega}{B^*}$. Consider two scenarios:
\begin{itemize}
	\item Suppose $j_1 = 0$. Then, construct another ROL $\succ^{R'}_\omega$ as follows: 
		\begin{itemize}
		\item suppose $|\succ^{R}_\omega| < K$. Construct ROL $\succ^{R'}_\omega$ the same as $\succ^{R}_\omega$ except appending $j_2$ to ROL at the position just above the outside option. That is, (i) $j \succ^{R'}_\omega j'$ for any $j,j'$ with $j \succ^R_\omega j'$, (ii) $j_2 \succ^{R'}_\omega 0$, and (iii) $j \succ^{R'}_\omega j_2$ for any $j\in \text{domain}(\succ^R_\omega)\setminus\{0\}$; 
		\item suppose $|\succ^{R}_\omega| = K$. In this case, $\succ^R_\omega$ is not compatible with $\gg$. Thus, there must exist some $j^*, j^\dagger \in \text{domain}(\succ^R_\omega)$ such that $j^* \succ^R_\omega j^\dagger$ and $j^\dagger \gg j^*$. Here, $j^*\neq 0$ and $j^\dagger\neq 0$, because $j^*\succ^R_\omega j^\dagger$ implies $j^* \neq j^\dagger$ and $j^*\neq 0$, and because $j^\dagger \gg j^*$ and $j^*\neq 0$ imply $j^\dagger\neq 0$.  In this case, the $j^\dagger$ listed in $\succ^R_\omega$ is redudndant, as there will never exist a $B\in \mathcal{F}^*_\omega$ such that  $j^\dagger = \optchoice{\succ^R_\omega}{B}$. Thus, we can construct ROL $\succ^{R'}_\omega$ the same as $\succ^{R}_\omega$ except dropping $j^\dagger$ from the ROL and appending $j_2$ to ROL at the position just above the outside option. Note that $j_1 \neq j^\dagger$ and $j_2\neq j^{\dagger}$, because $j_1 = 0$ and $0 = j_1\succ^{u1}_\omega j_2$ implies $j_2 \notin \text{domain}(\succ^R_\omega)$.
	\end{itemize}
	By construction, $\optchoice{\succ^{R'}_\omega}{B^*} = j_2 \succ^Q_\omega 0 = \optchoice{\succ^R_\omega}{B^*}$. Moreover, for any $B\in \mathcal{F}^*_\omega$, either $\optchoice{\succ^{R'}_\omega}{B^*} = \optchoice{\succ^R_\omega}{B^*}$ or $\optchoice{\succ^{R'}_\omega}{B^*} \neq \optchoice{\succ^R_\omega}{B^*}$. If $\optchoice{\succ^{R'}_\omega}{B^*} \neq \optchoice{\succ^R_\omega}{B^*}$, we must have $\optchoice{\succ^{R'}_\omega}{B} = j_2 \succ^Q_\omega 0 = \optchoice{\succ^R_\omega}{B}$ by the construction of $\succ^{R'}_\omega$. As a result, $\succ^R_\omega$ is dominated by $\succ^{R'}_\omega$ given $\mathcal{F}^*_\omega$, contradicting to Assumption \ref{assu:robust_undominance}. 

	\item Suppose $j_1 \neq 0$. Then, construct another ROL $\succ^{R'}_\omega$ as follows:
    \begin{itemize}
        \item if $\succ^R_\omega| < K$, construct $\succ^{R'}_\omega$ as ordering $\text{domain}(\succ^R_\omega)\cup\{j_2\}$.
        \item if $\succ^R_\omega| = K$, then $\succ^R_\omega$ must be not compatible with $\gg$. As shown above, there must exist some $j^\dagger$ such that there never exist a $B\in \mathcal{F}^*_\omega$ such that $j^\dagger = \optchoice{\succ^R_\omega}{B}$. Construct $\succ^{R'}_\omega$ by selecting its domain as $\{j_2\}\cup\text{domain}(\succ^R_\omega)\setminus\{j^\dagger\}$ and ordering programs within domain according to $\succ^Q_\omega$.
    \end{itemize}
   By construction, $j_2\in \text{domain}(\succ^{R'}_\omega)$ and     
    $j \succ^{R'}_\omega j'$ only if  $j \succ^Q_\omega j'$. Thus, by construction, either $\optchoice{\succ^{R'}_\omega}{B^*} = j_2$ or $\optchoice{\succ^{R'}_\omega}{B^*} \succ^Q_\omega j_2$. In both cases, we have $\optchoice{\succ^{R'}_\omega}{B^*} \succ^Q_\omega  \optchoice{\succ^R_\omega}{B^*}$, because $\optchoice{\succ^R_\omega}{B^*} = j_1$. Moreover, for any $B\in \mathcal{F}^*_\omega$, either $\optchoice{\succ^{R'}_\omega}{B^*} = \optchoice{\succ^R_\omega}{B^*}$ or $\optchoice{\succ^{R'}_\omega}{B^*} \neq \optchoice{\succ^R_\omega}{B^*}$. If $\optchoice{\succ^{R'}_\omega}{B^*} \neq \optchoice{\succ^R_\omega}{B^*}$, we must have $\optchoice{\succ^{R'}_\omega}{B} \succ^Q_\omega \optchoice{\succ^R_\omega}{B}$ because $\succ^{R'}_\omega$ is compatible with $\succ^Q_\omega$ by construction. As a result, $\succ^R_\omega$ is dominated given $\mathcal{F}^*_\omega$, contradicting to Assumption \ref{assu:robust_undominance}.
\end{itemize}
Thus, we always have a contradiction for student $\omega$ with $|\succ^R_\omega| < K$.

\item Fix an arbitrary $\omega$ with $|\succ^{R}_\omega| = K$ and $\succ^R_\omega$ compatible with $\gg$. We want to show $\succ^Q_\omega$ must be compatible with  $\succ^{P_{u1}}_\omega$ or $\succ^{P_{sel}}_\omega$. Define $\mathcal{A} \coloneqq \{\optchoice{\succ^Q_\omega}{B}: B\in \mathcal{F}^*_\omega\} \cup \{0\}$. Consider two scenarios:
	\begin{itemize}
	\item Suppose the number of elements in $\mathcal{A}$, denoted as $|\mathcal{A}|$, is less than or equal to $K$. In this case, we want to show $\succ^Q_\omega$ is compatible with $\succ^{P_{u1}}_\omega$. 
Suppose, for the purpose of contradiction, that $\succ^Q_\omega$ is not compatible with $\succ^{P_{u1}}_\omega$. 
Then, $\succ^Q_\omega$ is not compatible with $\succ^{u1}_\omega$. 
As a result, there must exist some $B^*\in \mathcal{F}_\omega\subseteq \mathcal{F}^*_\omega$ and some $j_1, j_2\in B^*$ such that $j_1 = \optchoice{\succ^R_\omega}{B^*}$, $j_1 \succ^{u1}_\omega j_2$ and $j_2 \succ^Q_\omega j_1$. 
Construct an alternative $\succ^{R'}_\omega$ where $j \succ^{R'}_\omega j'$ if and only if $j \succ^Q_\omega j'$ and $j, j' \in \mathcal{A}$. By construction, for any $B\in \mathcal{F}^*_\omega$, $\optchoice{\succ^{R'}_\omega}{B} = \optchoice{\succ^Q_\omega}{B}$. Hence, for any $B\in \mathcal{F}^*_\omega$, either $\optchoice{\succ^{R'}_\omega}{B} = \optchoice{\succ^R_\omega}{B}$ or $\optchoice{\succ^{R'}_\omega}{B} \succ^Q_\omega \optchoice{\succ^R_\omega}{B}$. Moreover, either $\optchoice{\succ^{R'}_\omega}{B^*} = j_2$ or $\optchoice{\succ^{R'}_\omega}{B^*} \succ^Q_\omega j_2$. In both cases, $\optchoice{\succ^{R'}_\omega}{B^*} \succ^Q_\omega j_1 = \optchoice{\succ^R_\omega}{B}$. As a result, $\succ^R_\omega$ is dominated by $\succ^{R'}_\omega$ given $\mathcal{F}^*_\omega$, contradicting to Assumption \ref{assu:robust_undominance}.

\item Suppose $|\mathcal{A}|$ is greater than $K$. In this case, we want to show $\succ^Q_\omega$ is compatible with $\succ^{P_{sel}}_\omega$. Suppose, for the purpose of contradiction, that $\succ^Q_\omega$ is not compatible with $\succ^{P_{sel}}_\omega$. Then, $\succ^Q_\omega$ is not compatible with $\succ^{sel}_\omega$. Then, one of the following two cases is true:
 \begin{itemize}
	 \item there exists some $j_1\in \text{domain}(\succ^R_\omega)$ and some $j_2$ such that (\emph{i}) $j_1 \gg j_2$, (\emph{ii}) there is no $j_2 \succ^R_\omega j_1$, and (\emph{iii}) $j_2 \succ^Q_\omega j_1$. We discuss the following cases:
	\begin{enumerate}
		\item Suppose $j_2 = 0$. Since $|\mathcal{A}| > |\succ^R_\omega| = K$, there exists some $B^*\in \mathcal{F}^*_\omega$ such that $j^* \coloneqq \optchoice{\succ^Q_\omega}{B^*} \notin \text{domain}(\succ^R_\omega)$. In particular, $j^* \succ^Q_\omega 0$. 
	 Construct an alternative $\succ^{R'}_\omega$ as follows: $\succ^{R'}_\omega$'s domain be $( \text{domain}(\succ^R_\omega)\setminus \{j_1\} ) \cup \{j^*\}$, and for any $j, j'$ in its domain, $j \succ^{R'}_\omega j'$ if and only if $j \succ^Q_\omega j'$. 
	Then, for any $B\in \mathcal{F}^*_\omega$, either $\optchoice{\succ^{R'}_\omega}{B} = \optchoice{\succ^R_\omega}{B}$ or $\optchoice{\succ^{R'}_\omega}{B} \succ^Q_\omega \optchoice{\succ^R_\omega}{B}$. 
	Moreover, $\optchoice{\succ^{R'}_\omega}{B^*} = j^* \succ^Q_\omega \optchoice{\succ^{R}_\omega}{B^*}$. As a result, $\succ^R_\omega$ is dominated by $\succ^{R'}_\omega$ given $\mathcal{F}^*_\omega$, contradicting to Assumption \ref{assu:robust_undominance}.

	\item Suppose $j_2 \neq 0$ and there does not exists any $B\in \mathcal{F}^*_\omega$ such that $j_1 = \optchoice{\succ^R_\omega}{B}$. In this case, $j_1$ is a redundnat program listed in $\succ^R_\omega$. We can construct the same $\succ^{R'}_\omega$ as the previous case. Following the same argument, we can show $\succ^{R'}_\omega$ dominates $\succ^R_\omega$, contradicting to Assumption \ref{assu:robust_undominance}.

	\item Suppose $j_2 \neq 0$ and there exists some $B^*\in \mathcal{F}^*_\omega$ such that $j_1 = \optchoice{\succ^R_\omega}{B^*}$. Construct an alternative $\succ^{R'}_\omega$ as follows: $\succ^{R'}_\omega$'s domain be $( \text{domain}(\succ^R_\omega)\setminus \{j_1\} ) \cup \{j_2\}$, and for any $j, j'$ in its domain, $j \succ^{R'}_\omega j'$ if and only if $j \succ^Q_\omega j'$. 

	Then, either $\optchoice{\succ^{R'}_\omega}{B^*} \succ^Q_\omega j_2$ or $\optchoice{\succ^{R'}_\omega}{B^*} = j_2$. In both cases, $\optchoice{\succ^{R'}_\omega}{B^*} \succ^Q_\omega \optchoice{\succ^{R}_\omega}{B^*} $, because $j_2 \succ^Q_\omega j_1 = \optchoice{\succ^{R}_\omega}{B^*}$. 

	Moreover, for an arbitrary $B\in \mathcal{F}^*_\omega$, either $\optchoice{\succ^{R'}_\omega}{B} = \optchoice{\succ^R_\omega}{B}$ or $\optchoice{\succ^{R'}_\omega}{B} \neq \optchoice{\succ^R_\omega}{B}$. 
	In the latter case, we must have $\optchoice{\succ^{R'}_\omega}{B} \succ^Q_\omega \optchoice{\succ^R_\omega}{B}$ because $\succ^{R'}_\omega$ replace $j_1$ with $j_2$, $j_2\succ^Q_\omega j_1$ and $\succ^{R'}_\omega$ orders the programs in its domain according to $\succ^Q_\omega$. As a result, $\succ^R_\omega$ is dominated by $\succ^{R'}_\omega$ given $\mathcal{F}^*_\omega$, contradicting to Assumption \ref{assu:robust_undominance}.

	\end{enumerate}
	Therefore, we have a contradiction in all of the above three cases.

\item there exists some $B^* \in \mathcal{F}_\omega \subseteq \mathcal{F}^*_\omega$ and some $j_1, j_2 \in B\cap\text{domain}(\succ^R_\omega)$ such that $j_2 \succ^Q_\omega j_1 = \optchoice{\succ^R_\omega}{B^*}$. Construct an alternative ROL $\succ^{R'}_\omega$ by setting its domain to be the same as $\text{domain}(\succ^R_\omega)$ and ordering the programs according to $\succ^Q_\omega$.
That is, $j \succ^{R'}_\omega j'$ if and only if $j,j'\in \text{domain}(\succ^R_\omega)$ and $j \succ^Q_\omega j'$. Because $\succ^{R'}_\omega$ and $\succ^R_\omega$ have the same domain and $\succ^{R'}_\omega$ is compatible with $\succ^Q_\omega$, for any $B\in \mathcal{F}^*_\omega$, either $\optchoice{\succ^{R'}_\omega}{B} = \optchoice{\succ^R_\omega}{B}$ or $\optchoice{\succ^{R'}_\omega}{B} \succ^Q_\omega \optchoice{\succ^R_\omega}{B}$. 
	Moreover, $\optchoice{\succ^{R'}_\omega}{B^*} \succ^Q_\omega j_1 = \optchoice{\succ^{R}_\omega}{B^*}$. As a result, $\succ^R_\omega$ is dominated by $\succ^{R'}_\omega$ given $\mathcal{F}^*_\omega$, contradicting to Assumption \ref{assu:robust_undominance}.
 \end{itemize}

\end{itemize}
Thus, we always have shown that $\succ^Q_\omega$ should be compatible with at least one of $\succ^{P_{u1}}_\omega$ and $\succ^{P_{sel}}_\omega$ for students with $|\succ^R_\omega| = K$.
	\end{itemize}

\end{appendices}

\end{document}